\documentclass[journal,12pt,draftcls,onecolumn]{IEEEtran}

\usepackage{epsfig}
\usepackage{subfigure}
\usepackage{graphicx}
\usepackage{color}

\usepackage{amsmath}
\usepackage{amssymb}

\def\BibTeX{{\rm B\kern-.05em{\sc i\kern-.025em b}\kern-.08em
    T\kern-.1667em\lower.7ex\hbox{E}\kern-.125emX}}

\newtheorem{corollary}{Corollary}
\newtheorem{theorem}{Theorem}
\newtheorem{definition}{Definition}
\newtheorem{lemma}{Lemma}

\newtheorem{remark}{Remark}
\newtheorem{baseline}{Baseline}

\newcommand{\blue}{\textcolor[rgb]{0,0,1}}

\newcommand{\nt}{\tilde{n}}
\newcommand{\mt}{\tilde{m}}
\newcommand{\at}{\tilde{\alpha}}
\newcommand{\rt}{\tilde{R}}
\newcommand{\wt}{\tilde{W}}
\newcommand{\kt}{\tilde{k}}
\newcommand{\ct}{\tilde{C}}

\date{}

\title{Two-Way Interference Channel Capacity: \\ How to Have the Cake and Eat it Too}

\author{Changho Suh, Jaewoong Cho and David Tse \\
\thanks{C. Suh and J. Cho are with the School of Electrical Engineering at Korea Advanced Institute of Science and Technology, South Korea (Email: $\mathsf{ \{chsuh, cjw2525 \}@kaist.ac.kr}$)}
\thanks{D. Tse is with the Electrical Engineering Department at Stanford University, CA, USA (Email: $\mathsf{dntse@stanford.edu}$).} }

\begin{document}

\maketitle

%\IEEEaftertitletext{
\begin{abstract}
%\textbf{(TBD)}

Two-way communication is prevalent and its fundamental limits are first studied in the point-to-point setting by Shannon~\cite{shannon:two-way}. One natural extension is a two-way  interference channel (IC) with four independent messages: two associated with each direction of communication. In this work, we explore a deterministic two-way IC which captures key properties of the wireless Gaussian channel. Our main contribution lies in the complete capacity region characterization of the two-way IC (w.r.t. the forward and backward sum-rate pair) via a new achievable scheme and a new converse. One surprising consequence of this result is that not only we can get an interaction gain over the one-way non-feedback capacities, we can sometimes \emph{get all the way to perfect feedback capacities in both directions simultaneously}. In addition, our novel outer bound characterizes channel regimes in which interaction has no bearing on capacity.

\end{abstract}

\section{Introduction}

Two-way communication, where two nodes want to communicate data to each other, is prevalent. The first study of such two-way channels was done by Shannon~\cite{shannon:two-way}  in the setting of point-to-point memoryless channels. When the point-to-point channels in the two directions are orthogonal (such as when the two directions are allocated different time slots or different frequency bands, or when the transmitted signal can be canceled perfectly as in full-duplex communication), the problem is not interesting as feedback does not increase point-to-point capacity. Hence, communication in one direction cannot increase the capacity of the other direction and no {\em interaction gain} is possible. One can achieve no more than the one-way capacity in each direction. 

The situation changes in network scenarios where feedback can increase capacity. In these scenarios, communication in one direction can potentially increase the capacity of the other direction by providing feedback in addition to communicating data. One scenario of particular interest is the setting of the two-way interference channel (two-way IC), modeling two interfering two-way communication links (Fig.~\ref{fig:system}). Not only is this scenario common in wireless communication networks, it has also been demonstrated that feedback provides a significant gain for communication over (one-way) IC's~\cite{Kramer:it02,SuhTse,AlirezaSuhAves}. In particular,~\cite{SuhTse} reveals that the feedback gain can be  unbounded, i.e., the gap between the feedback and non-feedback capacities can be arbitrarily large for certain channel parameters. This suggests the potential of significant interaction gain in two-way IC's. On the other hand,  the feedback result~\cite{SuhTse} assumes a dedicated infinite-capacity feedback link. In the two-way setting, any feedback needs to be transmitted through a backward IC,  which also needs to carry its own backward data traffic. The question is when we take in consideration the competition with the backward traffic, whether there is still any net interaction gain through feedback?

 \begin{figure}[t]
\begin{center}
{\epsfig{figure=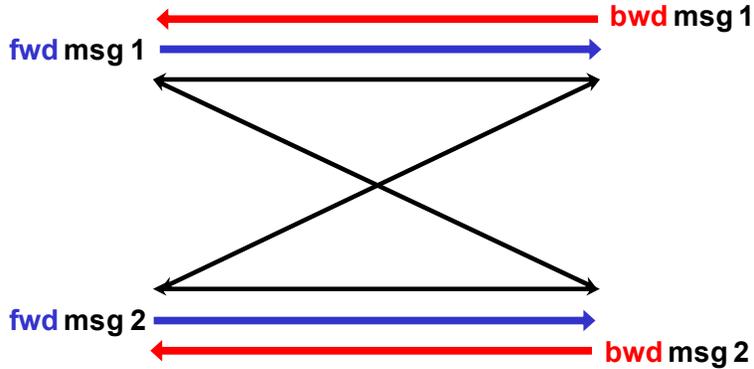, angle=0, width=0.6\textwidth}}
\end{center}
\caption{Two interfering two-way communication links, consisting of two IC's, one in each direction. The IC's are orthogonal to each other and do not necessarily have the same channel gains.}
\label{fig:system}
\vspace*{-0.1in}
\end{figure}

To answer this question,~\cite{SuhWangTse} investigated a two-way IC under the linear deterministic model~\cite{Salman:IT11}, which approximates a Gaussian channel. A scheme is proposed to demonstrate a net interaction gain, i.e., one can simultaneously achieve better than the non-feedback capacities in both directions.  While an outer bound is also derived, it has a gap to the lower bound. Hence, there has been limited understanding on the maximal gain that can be reaped by feedback. In particular, whether or not one can get all the way to perfect feedback capacities in both directions has been unanswered. 
Recently Cheng-Devroye~\cite{Devroye:IT2014} derived an outer bound, but it does not give a proper answer as the result assumes a partial interaction scenario in which interaction is enabled only at two nodes, while no interaction is permitted at the other two nodes. 

%Under the full interaction scenario of interest, a new outer bound is established in~\cite{SuhTseCho}. However, the improved bound still comes with a gap to the lower bound and thus the maximal feedback gain has been unknown so far. 

\begin{figure}[t]
\begin{center}
{\epsfig{figure=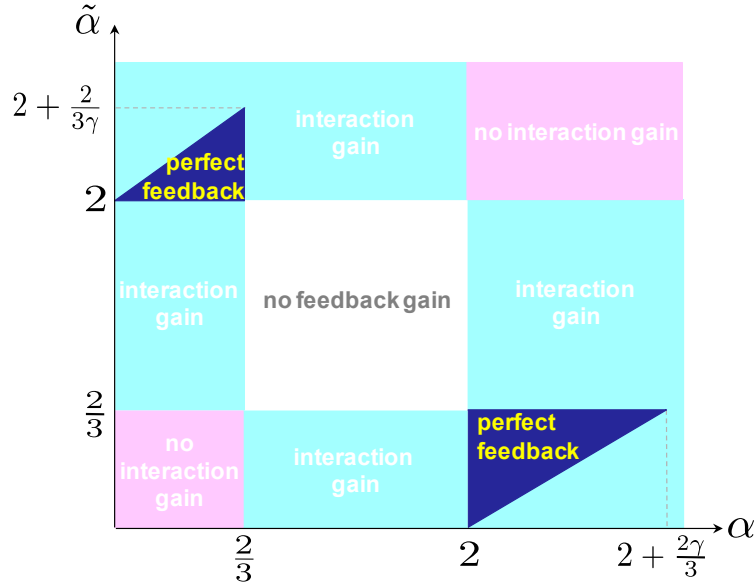, angle=0, width=0.6\textwidth}}
\end{center}
\caption{{\bf When can one have the cake and eat it too?} The plot is over two channel parameters of the deterministic model, $\alpha$ and $\tilde{\alpha}$, where $\alpha$ is the ratio of the interference-to-noise ratio (in dB) to the signal-to-noise ratio (in dB) of the IC in the forward direction and $\tilde{\alpha}$ is the corresponding quantity of the IC in the backward direction. The parameter $\gamma$ is the ratio of the backward signal-to-noise ratio (in dB) to the forward signal-to-noise ratio (in dB), and is fixed to be a value between $1$ and $4$. White region: feedback does not increase capacity in either direction and thus interaction is not useful. Purple: feedback does increase capacity but interaction cannot provide such increase. Light blue: feedback can be provided through interaction and there is a net interaction gain. Dark blue: interaction is so efficient that one can achieve perfect feedback capacity simultaneously in both directions. This implies that one can obtain the maximal feedback gain without any sacrifice for feedback transmission (have the cake and eat it too).}
\label{fig:freeride_vs_nofreeride}
\vspace*{-0.15in}
\end{figure}

In this work, we settle this open problem and completely characterize the capacity region of the deterministic two-way IC via a new capacity-achieving transmission scheme as well as a novel outer bound. For simplicity, we  assume the IC in each direction is symmetrical between the two users; however the IC's in the two directions are not necessarily the same (for example, they may use different frequency bands). For some channel gains, the new scheme simultaneously achieves the perfect feedback sum-capacities of the IC's in both directions. % even when feedback transmission competes with independent message transmission.
This occurs even when feedback offers gains in both directions and thus feedback must be explicitly or implicitly carried over each IC while sending the traffic in its own direction.  Fig.~\ref{fig:freeride_vs_nofreeride} shows for what channel gains this happens.
 
In the new scheme, feedback allows the exploitation of  the following as side information: (i) past received signals; (ii) users' own messages; (iii) even the \emph{future} information via \emph{retrospective decoding} (to be detailed later; see Remark~\ref{remark:whyfreelunch} in particular).
While the first two were already shown to offer a feedback gain in literature, the third is newly exploited. It turns out this new exploitation leads us to achieve the perfect feedback capacities in both directions, which can never be done by the prior schemes~\cite{SuhTse,AlirezaSuhAves,SuhWangTse}.

Our new outer bound leads to the characterization of channel regimes in which interaction provides no gain in capacity.
The bound is neither cutset nor more sophisticated bounds such as genie-aided bounds~\cite{ElGamal:it82, Kramer:it02, SuhTse, Rini:it11, WangTse:it11, Vinod:it11, Sahai:ITW09} and the generalized network sharing bound~\cite{Kamath:allerton14}. We employ a notion called triple mutual information, also known as interaction information~\cite{InteractionInfo:54}. In particular, we exploit one key property of the notion, commutativity, to derive the bound.

\section{Model}
\label{sec:model}

Fig.~\ref{fig:model} describes a two-way deterministic IC where user $k$ wants to send its own message $W_k$ to user $\tilde{k}$, while user $\tilde{k}$ wishes to send its own message $\wt_k$ to user $k$, $k=1,2$. We assume that $(W_1, W_2, \wt_1, \wt_2)$ are independent and uniformly distributed. For simplicity, we consider a setting where both forward and backward ICs are symmetric but not necessarily the same. In the forward IC, $n$ and $m$ indicate the number of signal bit levels for direct and cross links respectively. The corresponding values in the backward IC are denoted by $(\tilde{n}, \tilde{m})$. Let $X_k \in \mathbb{F}_2^{\max(n,m)}$ be user $k$'s transmitted signal and $V_{k} \in \mathbb{F}_2^m$ be a part of $X_k$ visible to user $\tilde{j} (\neq \tilde{k})$. Similarly let $\tilde{X}_k$ be user $\tilde{k}$'s transmitted signal and $\tilde{V}_k$ be a part of $\tilde{X}_k$ visible to user $j (\neq k)$.
The deterministic model abstracts broadcast and superposition of signals in the wireless Gaussian channel. See~\cite{Salman:IT11} for explicit details.
A signal bit level observed by both users is broadcasted. If multiple signal levels arrive at the same signal level at a user, we assume modulo-2-addition.
The encoded signal $X_{ki}$ of user $k$ at time $i$ is a function of its own message and past received signals:
%\begin{align}
%\label{eq:Xfunc_relation}
$X_{ki}= f_{ki} (W_k, \tilde{Y}_{k}^{i-1})$.
%\end{align}
We define $\tilde{Y}_{k}^{i-1}:= \{ \tilde{Y}_{kt} \}_{t=1}^{i-1}$ where $\tilde{Y}_{kt}$ denotes user $k$'s received signal at time $t$, offered through the backward IC. Similarly the encoded signal $\tilde{X}_{k i}$ of user $\kt$ at time $i$ is a function of its own message and past received signals:
%\begin{align}
%\label{eq:Xtildefunc_relation}
$\tilde{X}_{ki} = \tilde{f}_{ki} (\wt_k, {Y}_{k}^{i-1})$.
%\end{align}

\begin{figure}[t]
\begin{center}
{\epsfig{figure=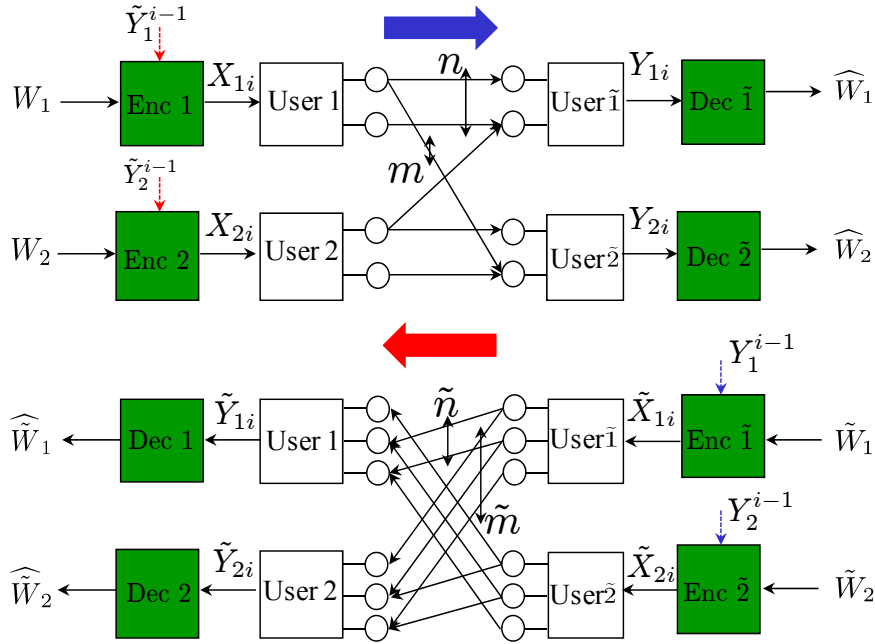, angle=0, width=0.7\textwidth}}
\end{center}
\caption{A two-way deterministic interference channel (IC).} \label{fig:model}
\vspace*{-0.1in}
\end{figure}

A rate tuple $(R_1,R_2, \rt_1, \rt_2)$ is said to be achievable if there exists a family of codebooks and  encoder/decoder functions such that the decoding error probabilities go to zero as code length $N$ tends to infinity. 

For simplicity, we focus on a sum-rate pair regarding the forward and bacward ICs: $(R,\rt):=(R_1+R_2, \rt_1 + \rt_2)$.\footnote{The extension to the four-rate tuple case is not that challenging although it requires a complicated yet tedious analysis. Given our results (to be presented soon) and the tradeoff w.r.t. $(R_1, R_2)$ (or $(\rt_1,\rt_2)$)  already characterized in~\cite{SuhTse}, the extension does not provide any additional insights. Hence, here we consider a simpler sum-rate pair setting.} The capacity region is defined as the closure of the set of achievable sum-rate pairs:
${\cal C} =  {\rm closure}  \{ (R, \rt): (R_1,R_2, \rt_1, \rt_2) \in \mathcal{C}_{\rm high} \}$ where ${\cal C}_{\rm high}$ denotes the one w.r.t. the high-dimensional rate tuple.
%Define $C^{\rm sum} = \sup \{ R+ \rt: (R, \rt) \in {\cal C}  \}$.

\section{Main Results}
\label{sec:MainResults}

Our main contribution lies in characterization of the capacity region of the two-way IC, formally stated below.

\begin{theorem}[Capacity region]
\label{thm:capacityregion}
The capacity region ${\cal C}$ of the two-way IC is the set of $(R, \rt)$ such that
\vspace*{-0.065in}
\begin{align}
R & \leq \max (2n - m, m) =: C_{\sf pf}  \label{eq:S3outerbound_1} \\
\tilde{R} & \leq   \max (2\tilde{n} - \tilde{m}, \tilde{m}) =: \tilde{C}_{\sf pf} \label{eq:S3outerbound_2} \\
R + \tilde{R} &\leq 2 (n + \tilde{n}) \label{eq:S3outerbound_3} \\
R + \tilde{R} &\leq 2 \max (n-m, m) + 2 \max (\tilde{n} - \tilde{m}, \tilde{m}) \label{eq:S3outerbound_4}
\end{align}
where $C_{\sf pf}$ and $\tilde{C}_{\sf pf}$ indicate the perfect feedback sum-capacities of the forward and backward IC's, respectively~\cite{SuhTse}. 
\end{theorem}
\begin{proof}
The achievability proof relies on two novel transmission schemes. In particular, we highlight key features of the second scheme - that we call~\emph{retrospective decoding} - which plays a crucial role to achieve perfect feedback capacities in both directions. The first feature is that it consists of two stages, each comprising a sufficiently large number $L$ of time slots. The second feature is that in the second stage, feedback-aided successive refinement w.r.t. the fresh symbols sent in the first stage occurs in a \emph{retrospective} manner: the fresh symbol sent in time $i$ of stage I is refined in time $2L+2-i$ of stage II where $1\leq i \leq L$. See Section~\ref{sec:achievability_capacityregion} for the detailed proof. 

For the converse proof, we first note that the first two bounds~\eqref{eq:S3outerbound_1} and~\eqref{eq:S3outerbound_2} 
match the perfect-feedback bound~\cite{SuhTse,Sahai:ITW09,SuhWangTse}. So one can prove them with a simple modification to the proof in the references. The third bound is due to cutset: $R_1 + \tilde{R}_2 \leq n + \tilde{n}$ and $R_2 + \tilde{R}_1 \leq n + \tilde{n}$. Our contribution lies in the derivation of the last bound. See Section~\ref{sec:proofnovelbound} for the proof.
\end{proof}

We state two baselines for comparison to our main result.

\begin{baseline}
%[Non-interaction capacity region~\cite{ElGamal:it82,bresler:europe}]
[\cite{ElGamal:it82,bresler:europe}]
\label{baseline:nonfeedback}
The capacity region ${\cal C}_{\sf no}$ for the non-interactive scenario is the set of $(R, \rt )$ such that
%$R \leq C_{\sf no}$ and $\tilde{R} \leq \tilde{C}_{\sf no}$ where
\vspace*{-0.065in}
\begin{align*}
R & \leq \min \left\{  2 \max (n-m, m), \max (2n-m, m), 2n \right \} =: C_{\sf no}  \\
\rt & \leq \min \left\{  2 \max ( \tilde{n}-\tilde{m}, \tilde{m}), \max (2 \tilde{n}- \tilde{m}, \tilde{m} ), 2 \tilde{n} \right \} = : \tilde{C}_{\sf no}.
\end{align*}
\end{baseline}

\vspace*{-0.05in}
\begin{baseline}
%[Perfect feedback capacity region~\cite{SuhTse}]
[\cite{SuhTse}]
\label{baseline:perfectfeedback}
The capacity region  for the perfect feedback scenario is ${\cal C}_{\sf pf} = \{ (R, \rt ): R \leq C_{\sf pf}, \rt \leq \tilde{C}_{\sf pf} \}$. 
\end{baseline}

%\begin{corollary}[Sum capacity]
%\label{thm:sumcapacity}
%The sum capacity of the two-way IC is:
%\vspace*{-0.05in} 
%\begin{align}
%\label{eq:sumcapacity}
%\begin{split}
%& C^{\rm sum}  =  \min   \left \{  \max (2n- m, m) + \max (2 \nt - \mt, \mt), \right . \\
%& \left . \; \;\; \;  2n + 2\tilde{n},  2\max ( n - m, m) + 2 \max ( \nt - \mt, \mt) \right \}.
%\end{split}
%\end{align}
%\end{corollary}
%\begin{proof}
%The proof is immediate from Theorem~\ref{thm:capacityregion}.
%\end{proof}

With Theorem~\ref{thm:capacityregion} and Baseline~\ref{baseline:nonfeedback}, one can readily see that feedback gain (in terms of capacity region) occurs as long as $(\alpha \notin [\frac{2}{3},2], \tilde{\alpha} \notin[\frac{2}{3},2])$, where $\alpha:=\frac{m}{n}$ and $\tilde{\alpha} := \frac{\tilde{m}}{\tilde{n}}$. 
%{\bf Feedback can come for free}: 
A careful inspection reveals that there are channel regimes in which one can enhance ${C}_{\sf no}$ (or $\tilde{C}_{\sf no}$) without sacrificing the other counterpart. This implies a net interaction gain.  

%\begin{figure}[t]
%\begin{center}
%{\epsfig{figure=./Figures/fig_comparisonS1S2.eps, angle=0, width=0.28\textwidth}}
%\end{center}
%\caption{Gain-vs-nogain picture in sum capacity:
%%Sum-rate comparisons between non-interactive and interactive scenarios.
%The blue regime indicates the case in which $C^{\rm sum}> C_{\sf no}^{\rm sum}$, while the red regime denotes the case where $C^{\rm sum} = C_{\sf no}^{\rm sum}$.
%} \label{fig:comparisonS1S2}
%\vspace*{-0.1in}
%\end{figure}

%
%\begin{figure}[t]
%\begin{center}
%{\epsfig{figure=./Figures/fig_freeride_vs_nofreeride2.eps, angle=0, width=0.28\textwidth}}
%\end{center}
%\caption{Free ride vs. no free ride:
%%Sum-rate comparisons between non-interactive and interactive scenarios.
%The blue regime indicates the case in which feedback comes for free in only one direction; the dark blue regime indicates the case in which one can get to the full feedback capacities in both directions (that we call \emph{strong free ride}) under some channel condition identified in Corollary~\ref{cor:S2sumrate_sumcapacity}; the purple regime indicates the case in which there is non-zero feedback cost and $C^{\rm sum} = C_{\sf no}^{\rm sum}$ while ${\cal C} \supsetneq {\cal C}_{\sf no}$; and the red regime denotes the case where ${\cal C} = {\cal C}_{\sf no}$.
%} \label{fig:freeride_vs_nofreeride}
%\vspace*{-0.1in}
%\end{figure}

\begin{definition}[Interaction gain]
\label{def:freeride_vs_nofreeride}
We say that an interaction gain occurs if one can achieve $(R,\rt) = (C_{\sf no} + \delta, \tilde{C}_{\sf no} + \tilde{\delta} )$ for some $\delta \geq 0$ and $\tilde{\delta} \geq 0$ such that $\max (\delta, \tilde{\delta} ) >0$. 
%If only one among such $\delta$ and $\tilde{\delta}$ is positive, then it is said to be \emph{one-way}; otherwise (both are positive), it is said to be \emph{two-way}.
\end{definition}
A tedious yet straightforward calculation with this definition leads us to identify channel regimes which exhibit an interaction gain, marked in light blue in Fig.~\ref{fig:freeride_vs_nofreeride}.

We also find the regimes in which feedback does increase capacity but interaction cannot provide such increase, meaning that whenever $\delta >0$, $\tilde{\delta}$ must be $-\delta$ and vice versa. These are $(\alpha \leq \frac{2}{3}, \tilde{\alpha} \leq \frac{2}{3})$ and $(\alpha \geq 2, \tilde{\alpha} \geq 2)$ marked in purple in Fig.~\ref{fig:freeride_vs_nofreeride}. The cutset bound~\eqref{eq:S3outerbound_3} proves this for $(\alpha \geq 2, \tilde{\alpha} \geq 2)$. The regime of $(\alpha \leq \frac{2}{3}, \tilde{\alpha} \leq \frac{2}{3})$ has been open as to whether both $\delta$ and $\tilde{\delta}$ can be non-negative. Our novel bound~\eqref{eq:S3outerbound_4} cracks the open regime, demonstrating that there is no interaction gain in the regime.

%obtain the picture of free ride vs. no free ride. This is shown in Fig.~\ref{fig:freeride_vs_nofreeride} for a practically-relevant scenario where $\gamma:= \frac{\tilde{n}}{n}=1$. We see that a free ride occurs in a wide range of regimes. 
%except for $(\alpha \leq \frac{2}{3}, \tilde{\alpha} \leq \frac{2}{3})$, $( \alpha \in [\frac{2}{3}, 2], \tilde{\alpha} \in [\frac{2}{3}, 2])$, and $(\alpha \geq 2, \tilde{\alpha} \geq 2)$.
%The regimes of $(\alpha \leq \frac{2}{3}, \tilde{\alpha} \leq \frac{2}{3})$, and $(\alpha \geq 2, \tilde{\alpha} \geq 2)$ are the ones in which  feedback offers a gain, but it comes at a cost. 
%It turns out one bit of feedback costs one bit in the regimes, admitting one-to-one tradeoff.  We also see that the definition of free ride is equivalent to having a feedback gain in sum capacity. Hence, whenever feedback offers a gain in sum capacity, feedback costs nothing; otherwise, no free lunch comes. We note that whether or not a free ride occurs depends only on $\alpha$ and $\tilde{\alpha}$, irrelevant to $\gamma$.

\vspace*{-0.02in}
{\bf Achieving perfect feedback capacities:} One interesting observation is that
there are channel regimes in which both $\delta$ and $\tilde{\delta}$ can be strictly positive. This is unexpected because it implies that not only feedback does not sacrifice one transmission for the other, it can actually improve both simultaneously. More interestingly, $\delta$ and $\tilde{\delta}$ can reach up to the maximal feedback gains, reflected in $C_{\sf pf} - C_{\sf no}$ and $\tilde{C}_{\sf pf} - \tilde{C}_{\sf no}$. 
%In Section~\ref{sec:achievability_capacityregion}, we will explain why this can happen, while describing our achievable scheme; see Remark~\ref{remark:whyfreelunch} in particular. 
The dark blue regimes in Fig.~\ref{fig:freeride_vs_nofreeride} indicate such channel regimes when $1 \leq \gamma := \frac{ \tilde{n}}{n} \leq 4$. 
%Fig.~\ref{fig:freeride_vs_nofreeride} demonstrates channel regimes (marked in dark blue) in which one can get to the perfect feedback capacities in both directions when $\gamma =1$. 
Note that such regimes depend on $\gamma$. The amount of feedback that one can send is limited by available resources offered by the backward (or forward) IC. Hence, the feedback gain can be saturated depending on availability of the resources, which is affected by the channel asymmetry parameter $\gamma$.
% captures the channel asymmetry between forward and backward ICs, so it also affects whether or not \emph{strong free ride} occurs.  
One point to note here is that for any $\gamma$, there always exists a non-empty set of $(\alpha, \tilde{\alpha})$ in which perfect feedback capacities can be achieved. Corollary~\ref{cor:S2sumrate_sumcapacity} stated below exhibits all of such channel regimes.
% in which the strong interaction occurs. 

%\emph{Example: $(n,m)=(2,1), (\tilde{n},\mt)=(0,1)$}: In this example, Theorem~\ref{baseline:nonfeedback} gives $C_{\sf no}^{\rm sum}=C_{\sf no}+ \tilde{C}_{\sf no} = 2+0 = 2$, while Theorem~\ref{thm:capacityregion} yields $C^{\rm sum} = C_{\sf pf} + \tilde{C}_{\sf pf} = 3 + 1 = 4$. Observe that this is a non-trivial case in which feedback helps in both ICs; 
%% (i.e., $C_{\sf pf} > C_{\sf no}$ and $\tilde{C}_{\sf pf} > \tilde{C}_{\sf no}$); 
% hence, feedback transmission must occur in both directions to achieve $C_{\sf pf}+ \tilde{C}_{\sf pf}$. Here achieving $C_{\sf pf} + \tilde{C}_{\sf pf}$ is quite  unexpected because it implies that feedback free-rides over the ICs. The tension between feedback and independent message transmissions are completely resolved, so the transmission comes as if there is no feedback cost at all. Corollary~\ref{cor:S2sumrate_sumcapacity} stated below exhibits such promising channel regimes in details. $\Box$
%% most promising channel regimes in which feedback cost is zero in both directions. 

\vspace*{-0.02in}
\begin{corollary}
\label{cor:S2sumrate_sumcapacity}
Consider a case in which feedback helps in both ICs:  $C_{\sf pf} > C_{\sf no}$ and $\tilde{C}_{\sf pf} > \tilde{C}_{\sf no}$. In this case, the channel regimes in which ${\cal C} = {\cal C}_{\sf pf}$ are: 
%(I) $\alpha < 2/3,  \tilde{\alpha} > 2$, 
%$C_{\sf pf} - C_{\sf no} \leq 2 \mt - \tilde{C}_{\sf pf}$, 
%$\tilde{C}_{\sf pf} - \tilde{C}_{\sf no} \leq 2n - C_{\sf pf}$;
%(II) $\at < 2/3,    \alpha > 2$,
%$\tilde{C}_{\sf pf} - \tilde{C}_{\sf no} \leq 2 m - C_{\sf pf}$,
%$ C_{\sf pf} - C_{\sf no} \leq  2 \tilde{n} - \tilde{C}_{\sf pf}$.
\begin{align*}
 & \textrm{(I)} \;\;  \alpha < 2/3,  \tilde{\alpha} > 2, 
\; C_{\sf pf} - C_{\sf no} \leq 2 \mt - \tilde{C}_{\sf pf}, \; \tilde{C}_{\sf pf} - \tilde{C}_{\sf no} \leq 2n - C_{\sf pf}; \\
% C_{\sf pf} - C_{\sf no} + \tilde{C}_{\sf pf} - \tilde{C}_{\sf no} \leq \min \{ 2 \mt - \tilde{C}_{\sf no}, 2n - C_{\sf no} \}; \\
&  \textrm{(II)} \;  \at < 2/3,    \alpha > 2, \;
\tilde{C}_{\sf pf} - \tilde{C}_{\sf no} \leq 2 m - C_{\sf pf}, \;
C_{\sf pf} - C_{\sf no} \leq  2 \tilde{n} - \tilde{C}_{\sf pf}.
%C_{\sf pf} - C_{\sf no} + \tilde{C}_{\sf pf} - \tilde{C}_{\sf no} \leq \min \{ 2 m - C_{\sf no}, 2 \tilde{n} - \tilde{C}_{\sf no} \}.
\end{align*} 
%\vspace*{-0.05in}
\end{corollary}
\vspace*{-0.02in}
\begin{proof}
A tedious yet straightforward calculation with Theorem~\ref{thm:capacityregion} completes the proof.
\end{proof}
\vspace*{-0.02in}
\begin{remark}[Why the Perfect Feedback Regimes?]
 When $\alpha <2/3$ and $\at >2$, $2 \mt$ indicates the total number of resource levels at the receivers in the backward channel. 
% which can also be viewed as the fully cooperative MIMO capacity. 
 Hence, one can interpret $2 \mt - \tilde{C}_{\sf pf}$ as the remaining resource levels (resource holes) that can potentially be utilized to aid forward transmission. It turns out feedback can maximize resource utilization by filling up the resource holes under-utilized in the non-interactive case. Note that $C_{\sf pf} - C_{\sf no}$ represents the amount of feedback that needs to be sent for achieving $C_{\sf pf}$. Hence, the condition $C_{\sf pf}-C_{\sf no} \leq 2 \mt - \tilde{C}_{\sf pf}$ (similarly $\tilde{C}_{\sf pf} - \tilde{C}_{\sf no} \leq 2 n - C_{\sf pf}$) in Corollary~\ref{cor:S2sumrate_sumcapacity} implies that as long as we have enough resource holes, we can get all the way to perfect feedback capacity.
% implies that the perfect feedback regimes come from the fact that feedback enables us sometimes to fully utilize all the available resource levels at receivers. 
 We will later provide an intuition as to why feedback can do so while describing our achievability; see Remark~\ref{remark:whyfreelunch} in particular. 
  $\Box$
\end{remark}

\section{Achievability Proof of Theorem~\ref{thm:capacityregion}}
\label{sec:achievability_capacityregion}

%\textbf{(TBD)} We first illustrate two achievable schemes that form the basis of achievability that encompasses arbitrary $(n,m,\nt, \mt)$. To this end, we focus on two toy examples which turn out to act as fundamental building blocks as well as the simplest examples for illustration of the two schemes. 

We first illustrate new transmission schemes via two toy examples in which the key ingredients of our achievability idea are well presented. Once the description of the schemes is done via the examples, we will then outline the proof for generalization while leaving a detailed proof for arbitrary channel parameters in Appendix~\ref{append:achievability_remaining}.

\subsection{Example 1: $(n,m)=(2,1), (\nt, \mt) = (1,2)$}
\label{sec:scheme1}

See Fig.~\ref{fig:perfectfeedback} for the channel structure of the example. The claimed rate region in this example reads $\{(R, \rt): R \leq C_{\sf pf} = 3, \rt \leq \tilde{C}_{\sf no} = 2 \}$. This is the case in which one can achieve $C_{\sf pf}$ while maintaining $\tilde{C}_{\sf no}$. We introduce a new transmission scheme (that we call \emph{Scheme 1}) to achieve the claimed rate region.

\begin{figure}[t]
\begin{center}
{\epsfig{figure=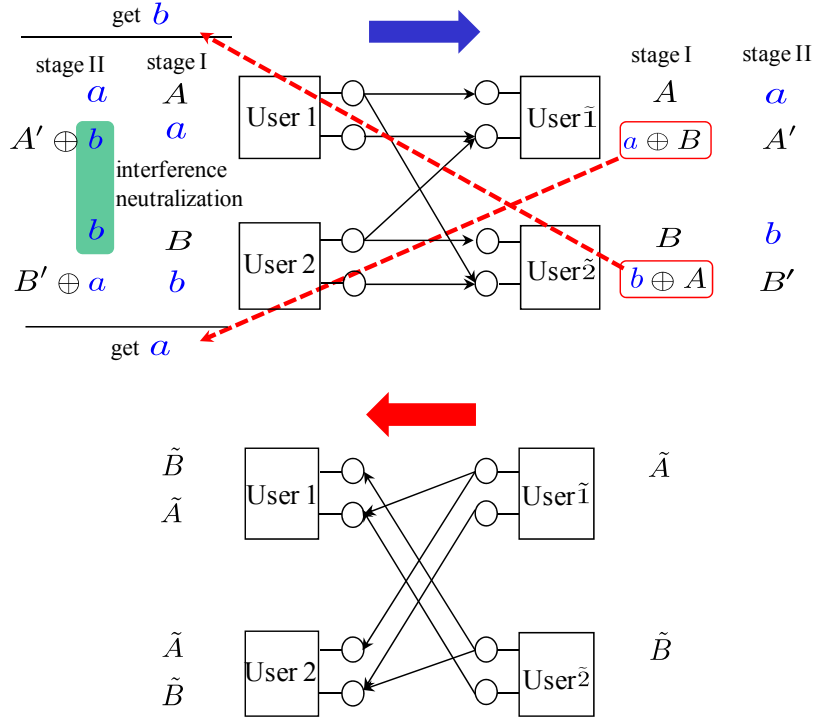, angle=0, width=0.65\textwidth}}
\end{center}
\caption{A perfect feedback scheme for $(n,m)=(2,1)$ where $C_{\sf pf} = 3$ (top); a nonfeedback scheme for $(\nt,\mt) = (1,2)$ where $\tilde{C}_{\sf pf} = \tilde{C}_{\sf no} = 2$ (bottom).
} \label{fig:perfectfeedback}
\vspace*{-0.1in}
\end{figure}

{\bf Perfect feedback scheme:} A perfect feedback scheme was presented in~\cite{SuhTse}. Here we consider a different scheme which allows us to resolve the tension between feedback and independent messages when translated into a two-way scenario. The scheme operates in two stages. See Fig.~\ref{fig:perfectfeedback}. In stage I, four fresh symbols ($(A,a)$ from user 1 and $(B,b)$ from user 2) are transmitted. The scheme in~\cite{SuhTse} feeds  $a \oplus B$ back to user 1, so that user 1 can decode $B$ which turns out to help refining the corrupted symbol $a$ in stage II. On the other hand, here we send $a \oplus B$ back to \emph{user 2}. This way, user 2 can get $a$ by removing its own symbol $B$. Similarly user 1 can get $b$. Now in stage II, user 2 intends to re-send $b$ on top, as the $b$ is corrupted due to $A$ in stage I. But here a challenge arises. The challenge is that the $b$ causes interference to user $\tilde{1}$ at the bottom level. But here the symbol $b$ obtained via feedback at user 1 can play a role. The idea of interference neutralization~\cite{Mohajer:it11} 
%(or dirty paper coding~\cite{Costa:it83}) 
comes into play. User 1 sending the $b$ on bottom enables neutralizing the interference. This then allows user 1 to transmit another fresh symbol, say $A'$, without being interfered. Similarly user 2 can carry $B'$ interference-free. This way, we send 6 symbols during two time slots, thus achieving $C_{\sf pf}=3$. As for the backward IC, we employ a nonfeedback scheme in~\cite{bresler:europe}. User $\tilde{1}$ and $\tilde{2}$ send $(\tilde{A}, \tilde{B})$ on top levels. This yields $\tilde{C}_{\sf no} = 2$.

%{\bf Scheme 1 (XORing with interference neutralization)}: 
We are now ready to illustrate our achievability. Like the perfect feedback scheme, it still operates in two stages and the operation of stage I remains unchanged. A new idea comes in feedback strategy. Recall that $a \oplus B$ is the one that is desired to be fed back to user 2. But the $a \oplus B$ has a conflict with transmission of $\tilde{A}$. It seems an explicit selection needs to be made between the two competing transmissions. But it turns out the two transmissions come without the conflict. 
%feeding back $a_1 \oplus B_1$ to user $2$; sending $\tilde{A}_1$ to user 1. 
The idea is to combine the XORing scheme introduced in network coding literature~\cite{ahlswede:it} with interference neutralization~\cite{Mohajer:it11}. See Fig.~\ref{fig:example1}. User $\tilde{1}$ simply sends the XOR of $a \oplus B$ and $\tilde{A}$ on top. User 1 can then extract $\tilde{A} \oplus B$ by using its own symbol $a$ as side information. But it is still interfered with by $B$. Here a key observation is that $B$ is also available at user $\tilde{2}$ - it was received cleanly at the top level in stage I. User $\tilde{2}$ sending the $B$ on bottom enables user 1 to achieve interference neutralization at the bottom level, thereby decoding $\tilde{A}$ of interest. Now consider user 2 side. User 2 can exploit $B$ to obtain $a \oplus \tilde{A}$. Note that $a \oplus \tilde{A}$ is not the same as $a$ wanted by user 2 in the perfect feedback scheme. Nonetheless $a \oplus \tilde{A}$ can serve the same role as $a$ and this will be clearer soon. Similarly, user $\tilde{2}$ sending $\tilde{B} \oplus (b \oplus A)$ on top while user $\tilde{1}$ sending $A$ (already delivered via the forward IC) on bottom, user $2$ can decode $\tilde{B}$ of interest and user 1 can get $b \oplus \tilde{B}$.

\begin{figure}[t]
\begin{center}
{\epsfig{figure=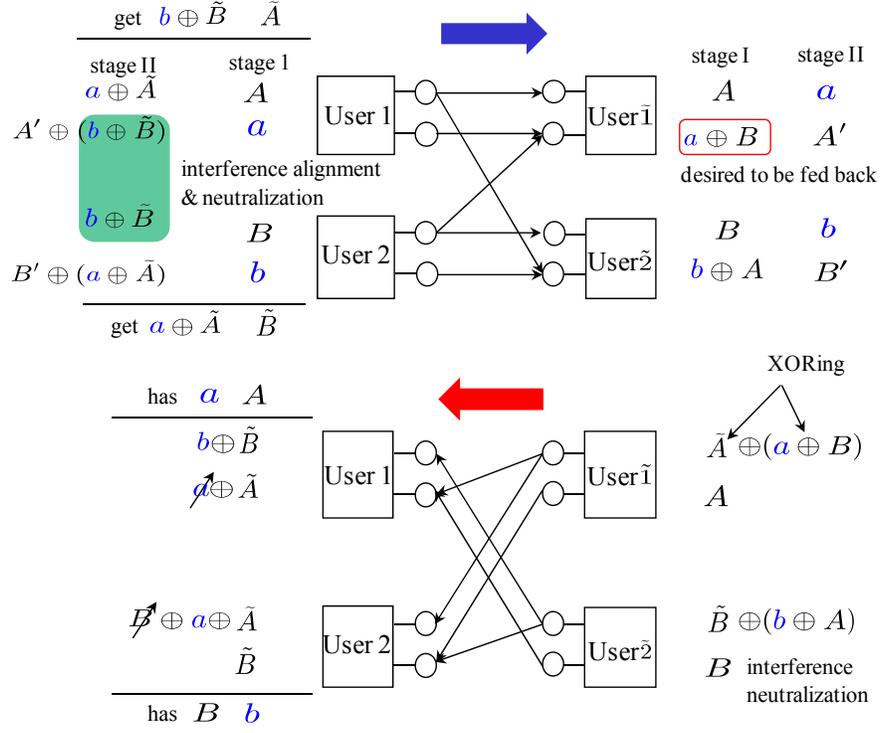, angle=0, width=0.7\textwidth}}
\end{center}
\caption{XORing with interferene neutralization for feedback strategy; Employing interference alignment and neutralization for refinement of the past corrupted symbols.} \label{fig:example1}
\vspace*{-0.1in}
\end{figure}

Now in stage II, we take a similar approach as in the perfect feedback case. User 2 intends to re-send $b$ on top. Recall in the perfect feedback scheme that user 1 sent the fedback symbol $b$ on bottom, in order to remove the interference caused to user $\tilde{1}$. But the situation is different here. User 1 has $b \oplus \tilde{B}$ instead. It turns out this can also play the same role. The idea is to use \emph{interference alignment} and \emph{neutralization}~\cite{MaddahAli:it08,Jafar:it08,Mohajer:it11}. User 1 sends $b \oplus \tilde{B}$ on bottom. Here $\tilde{B}$ seems to cause interference to user $\tilde{1}$. But this can be canceled as $\tilde{B}$ is already decoded at user 2 - see the bottom level at user 2 in the backward channel. User 2 sending $b \oplus \tilde{B}$ on top enables interference neutralization. This allows user 1 to send another fresh symbol $A'$ on bottom interference-free. Note that $b \oplus \tilde{B}$ can be viewed as the aligned interference w.r.t. $A$. Similarly user 1 sending $a \oplus \tilde{A}$ on top and user 2 sending $B' \oplus (a \oplus \tilde{A})$ on bottom, user $\tilde{1}$ and $\tilde{2}$ can decode $a$ and $B'$ respectively. This way, we achieve $C_{\sf pf }=3$ as in the perfect feedback case while maintaining $\tilde{C}_{\sf no}=2$. Hence, the claimed rate region is achieved. $\Box$

\begin{remark}[Exploiting Side Information]
\label{remark:scheme1}
% Notice in this example that $2 \tilde{m} = 4$ indicates the total number of resource levels at the receivers in the backward channel and thus $2 \tilde{m} - \tilde{C}_{\sf no} = 2$ can be interpreted as the remaining resource levels that can potentially be exploited to send feedback signals w.r.t. forward symbols. As illustrated in Fig.~\ref{fig:example1} (bottom), our scheme fully utilizes the remaining resources to send two feedback signals $(b \oplus \tilde{B}, a \oplus \tilde{A})$. 
 Note in Fig.~\ref{fig:example1} (bottom) that the two backward symbols $(\tilde{A}, \tilde{B})$ and the two feedback signals $(a \oplus B, b \oplus A)$ can be transmitted through 2-bit-capacity backward IC.
 This is because each user can cancel the seemingly interfering information by exploiting what has been received and its own symbols as side information. The side information allows the backward IC to have an effectively larger capacity, thus yielding a gain. This gain equalizes feedback cost, which in turn enables feedback to come for free in the end. The nature of the gain offered by side information coincides with that of the two-way relay channel~\cite{Wu:05} and many other examples~\cite{Katti:SIGCOMM06,BarYossef:FOCS06,SuhTse,MaddahAli:allterton10,MaddahAli:it13}. 
$\Box$
\end{remark} 

\subsection{Example 2: $(n,m)=(2,1), (\nt, \mt)=(0,1)$}
\label{sec:scheme2}

Scheme 1 is intended for the regimes in which feedback provides a gain only in one direction, e.g., $C_{\sf pf}> C_{\sf no}$ and $\tilde{C}_{\sf pf} = \tilde{C}_{\sf no}$. For the regimes feedback helps in both directions, we develop another transmission scheme (that we call \emph{Scheme 2}) which enables us to get sometimes all the way to perfect feedback capacities. In this section, we illustrate the scheme via Example 2 in which $(C_{\sf pf} = 3 > 2 = C_{\sf no}, \tilde{C}_{\sf pf} =1 > 0 = \tilde{C}_{\sf no})$ and one can achieve $(R,\tilde{R}) = (C_{\sf pf}, \tilde{C}_{\sf pf})$.
See Fig.~\ref{fig:example2_stage1} for the channel structure of the example.

Our scheme operates in two stages. But one noticeable distinction is that each stage comprises a sufficiently large number of time slots. Specifically stage I consists of $L$ time slots, while stage II uses $L+1$ time slots. 
It turns out our scheme ensures transmission of $6L$ forward symbols and $2L$ backward symbols, thus yielding:
\vspace*{-0.05in} 
\begin{align*}
(R, \rt) = \left( \frac{ 6L}{2L+1}, \frac{2L }{2L +1} \right) \longrightarrow (3,1) = (C_{\sf pf}, \tilde{C}_{\sf pf}). 
\end{align*}
as $L \rightarrow \infty$. Here are details.

Before describing details, let us review the perfect feedback scheme of the backward IC~\cite{SuhTse} which takes a relaying idea. User $\tilde{1}$ delivers a backward symbol, say $\tilde{a}$, to user 1 via the feedback-assisted path: user $\tilde{1}$ $\rightarrow$ user 2 $\rightarrow$ feedback $\rightarrow$ user $\tilde{2} \rightarrow$ user 1. Similarly user $\tilde{2}$ sends $\tilde{b}$ to user 2. This yields $\tilde{C}_{\sf pf}=1$.

\begin{figure}[!t]
\begin{center}
{\epsfig{figure=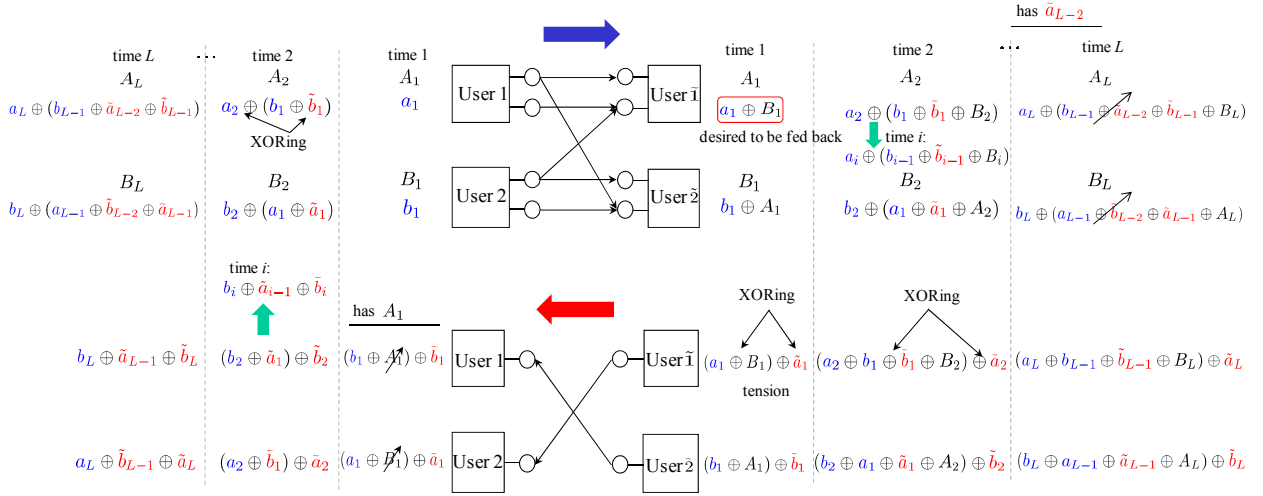, angle=0, width=1.0\textwidth}}
\end{center}
\caption{Stage I: Employ $L$ time slots. The operation in  each time slot is similar to stage I's operation in the perfect feedback case. We simply forward the XOR of a feedback signal and a new independent symbol. Here we see the tension between them. 
%A noticeable distinction arises in handling the tension between feedback and independent message transmissions. 
%Here the tension between a feedback signal and a new independent symbol is simply ignored. We forward the XOR of them.
} \label{fig:example2_stage1}
\vspace*{-0.1in}
\end{figure}

{\bf Stage I:} We employ $L$ time slots.  In each time slot, we mimick the perfect feedback scheme although we have the tension between feedback and independent message transmissions. 
%Simply we send the XOR of a feedback signal and a fresh independent symbol. See below for details. 

\emph{Time 1}: Four fresh symbols are transmitted over the forward IC. User $\tilde{1}$ then extracts the one that is desired to be fed back: $a_1 \oplus B_1$. Next we send the XOR of $a_1 \oplus B_1
$ and a backward symbol, say $\tilde{a}_1$. 
%Aiming to achieve $C_{\sf pf} + \tilde{C}_{\sf pf}$, one cannot avoid the tension between $a_1 \oplus B$ and 
% $\tilde{a}_1$. It turns out this tension can be resolved later in stage II. 
 Similarly user $\tilde{2}$ transmits $(b_1 \oplus A_1 ) \oplus \tilde{b}_1$. 
User 1 then gets $b_1 \oplus \tilde{b}_1$ using its own symbol $A_1$. Similarly user 2 gets $a_1 \oplus \tilde{a}_1$. 

\emph{Time 2}: User 1 superimposes $b_1 \oplus \tilde{b}_1$ with another new symbol, say $a_2$, sending the XOR on bottom. On top is another fresh symbol $A_2$ transmitted. Similarly user 2 sends $(B_2, b_2 \oplus (a_1 \oplus \tilde{a}_1))$. 
%Next we repeat the same as before. 
User $\tilde{1}$ transmits $(a_2 \oplus b_1 \oplus \tilde{b}_1 \oplus B_2)\oplus \tilde{a}_2$. 
%We see a tension between many symbols. Again this will be resolved later in stage II.  
Similarly user $\tilde{2}$ sends $(b_2 \oplus a_1 \oplus \tilde{a}_1 \oplus A_2) \oplus \tilde{b}_2$. User 1 then gets $b_2 \oplus \tilde{a}_1 \oplus \tilde{b}_2$ by using its own signal $a_1 \oplus A_2$. Similarly user 2 obtains $a_2 \oplus \tilde{b}_1 \oplus \tilde{a}_2$.
%Observe what user 1 and 2 obtain at the end of time 2.
% that user 1 and 2 get $b_2 \oplus \tilde{a}_1 \oplus \tilde{b}_2$ and $a_2 \oplus \tilde{b}_1 \oplus \tilde{a}_2$, respectively. 
 Repeating the above, one can readily verify that at time $i \in \{2, \cdots, L \}$, user 1 and 2 get $b_{i} \oplus \tilde{a}_{i-1} \oplus \tilde{b}_{i-1}$ and $a_{i} \oplus \tilde{b}_{i-1} \oplus \tilde{a}_{i-1}$ respectively; similarly user $\tilde{1}$ and $\tilde{2}$ get $a_i \oplus b_{i-1} \oplus \tilde{b}_{i-1} \oplus B_i $ and $b_i \oplus a_{i-1} \oplus \tilde{a}_{i-1} \oplus A_i $ on bottom, respectively. See Fig.~\ref{fig:example2_stage1}.

 \begin{figure*}[!htp]
\begin{center}
{\epsfig{figure=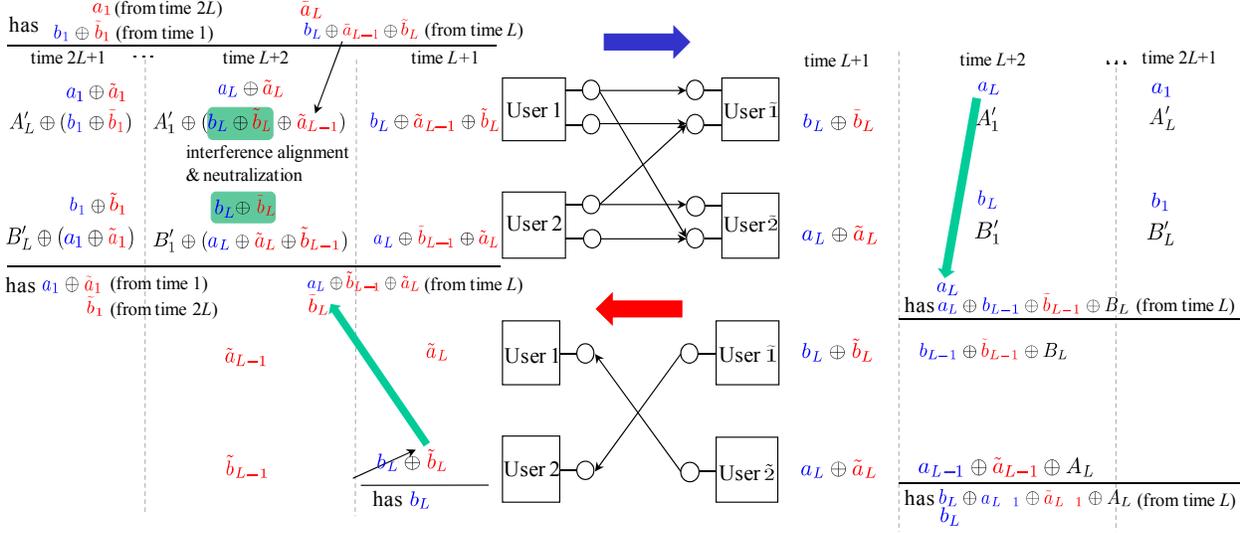, angle=0, width=1.0\textwidth}}
\end{center}
\caption{Stage II: Time $L+1$ aims at decoding $(\tilde{a}_L, \tilde{b}_L)$. At time $L+1+i$, given $(\tilde{a}_{L+1-i}, \tilde{b}_{L+1-i})$ (decoded in time $L+i$), we decode $(a_{L+1-i},b_{L+1-i})$ which in turn helping decoding $(\tilde{a}_{L-i},\tilde{b}_{L-i})$. We iterate this from $i=1$ to $i=L$.} \label{fig:example2_stage2}
\vspace*{-0.1in}
\end{figure*}

%\emph{Time $L$}: Suppose at time $L-1$ that user $1$ gets $b_{L-1} \oplus \tilde{a}_{L-2} \oplus \tilde{b}_{L-1}$. In time $L$, user 1 sends $ a_L \oplus (b_{L-1} \oplus \tilde{a}_{L-2} \oplus \tilde{b}_{L-1})$ on bottom, while transmitting $A_L$ on top. Similarly user 2 sends $(B_L, b_L \oplus (a_{L-1} \oplus \tilde{b}_{L-2} \oplus \tilde{a}_{L-1}  ))$. Using its own symbol $\tilde{a}_{L-2}$, user $\tilde{1}$ obtains $a_L \oplus b_{L-1} \oplus  \tilde{b}_{L-1} \oplus B_L$ on bottom. 
%%Similarly user $\tilde{2}$ gets $b_L \oplus a_{L-1} \oplus  \tilde{a}_{L-1} \oplus A_L$. 
%Next user $\tilde{1}$ transmits $(a_L \oplus b_{L-1} \oplus \tilde{b}_{L-1} \oplus B_L ) \oplus \tilde{a}_L$.
%User 2 then gets $a_{L} \oplus \tilde{b}_{L-1} \oplus \tilde{a}_{L}$.
%Similarly user 1 gets $b_{L} \oplus \tilde{a}_{L-1} \oplus \tilde{b}_{L}$. 
%%This proves the above conjecture. 

%\vspace*{-0.09in} 
{\bf Stage II:} We employ $L+1$ time slots. We perform refinement w.r.t. the fresh symbols sent in stage I. The novel feature here is that the successive refinement occurs in a \emph{retrospective} manner: the fresh symbol sent at time $i$ is refined at time $2L+2-i$ in stage II where $1\leq i \leq L$. Here one key point to emphasize is that the refined symbol in stage II acts as \emph{side information}, which in turn helps refining other past symbols in later time. 
%In other words, side information is dynamically changing and this enables peeling all the mess accumulated in stage I. 
%It turns out exploiting the side information plays a central role to enable refining all the symbols corrupted in stage I.  
% A distinctive feature here is that the refinement of the past corrupted symbols is done in a retrospective manner. 
In the example, the decoding order reads: 
%\vspace*{-0.07in} 
\begin{align}
\label{eq:decodingorder}
(\tilde{a}_L, \tilde{b}_L) \rightarrow (a_L, b_L) \rightarrow 
%\tilde{a}_{L-1},\tilde{b}_{L-1}) \rightarrow 
%(a_{L-1},b_{L-1}) \rightarrow 
\cdots \rightarrow (\tilde{a}_1, \tilde{b}_1) \rightarrow (a_1,b_1).
\end{align}

%\vspace*{-0.09in} 
\emph{Time $L+1$}: 
User 1 sends $b_L \oplus \tilde{a}_{L-1} \oplus \tilde{b}_L$ (received at time $L$) on bottom. It turns out this acts as \emph{ignition} for refining all the corrupted symbols in the past. 
%A new symbol $A_{L+1}$ is additionally transmitted on top. 
Similarly user 2 sends $a_L \oplus \tilde{b}_{L-1} \oplus \tilde{a}_L$ on bottom. %User $\tilde{1}$ then receives $b_L \oplus \tilde{a}_{L-1} \oplus \tilde{b}_L \oplus B_{L+1}$ on bottom. 
%Exploiting its own symbol $\tilde{a}_{L-1}$, 
User $\tilde{1}$ can then obtain $b_L \oplus \tilde{b}_L$ which would be forwarded to user 2. User 2 can then decode $\tilde{b}_L$ of interest.
% with the help of its own symbols $(b_L, B_{L+1})$. 
 Similarly $\tilde{a}_L$ is delivered to user 1.  

\emph{Time $L+2$}: The decoded symbols $(\tilde{a}_L,\tilde{b}_L)$ turn out to play a key role to refine past forward transmission. Remember that $b_L$ sent by user 2 at time $L$ in stage I was corrupted. User 2 re-transmits the $b_L$ on top as in the perfect feedback case. But here the problem is that the situation is different from that in the perfect feedback case where $b_L$ was available at user 1 and helped nulling interference. Note that $b_L$ is not available here. Instead user 1 has an interfered version: $b_L \oplus  \tilde{b}_L \oplus \tilde{a}_{L-1} $. Nonetheless we can effectively do the same as in the perfect feedback case. 
% by \emph{exploiting the decoded symbol $\tilde{b}_L$ as side information}. 
 User 1 sends $b_L \oplus \tilde{b}_L \oplus \tilde{a}_{L-1}$ on bottom.
%  in an effort to neutralize the inference caused by user 2's transmission of $b_L$ on top. 
  Clearly the neutralization is not perfect as it contains $\tilde{b}_L$.
  %which cannot be canceled at user $\tilde{1}$. 
%Here comes a new idea. 
Here the idea is to exploit the $\tilde{b}_L$ as side information to enable interference alignment and neutralization~\cite{MaddahAli:it08,Jafar:it08,Mohajer:it11}. 
Note that user 2 can exploit the knowledge of $\tilde{b}_L$ to construct the \emph{aligned interference} $b_L \oplus \tilde{b}_L$. Sending the $b_L \oplus \tilde{b}_L$ on top, user 2 can completely neutralize the interference as in the perfect feedback case. This enables user 1 to deliver $A_{1}'$ interference-free on bottom. Similarly we can deliver $(a_L, B_{1}')$.
%\emph{Time $L+2$ (backward)}: 
%Remember that $a_L \oplus b_{L-1} \oplus \tilde{b}_{L-1} \oplus B_L$ is the one that user $\tilde{1}$ received at time $L$ in stage I. 
On the other hand, exploiting $a_L$ (decoded right before) as side information, user $\tilde{1}$ can extract $b_{L-1} \oplus \tilde{b}_{L-1} \oplus B_L$ from the one received at time $L$. Sending this then allows user 2 to decode $\tilde{b}_{L-1}$. Similarly $\tilde{a}_{L-1}$ can be decoded at user 1. 

\emph{Time $L+3$ $\sim$ Time $2L+1$}: We repeat the same as before. At  time $L+1+i$ where $2 \leq i \leq L$, exploiting $(\tilde{a}_{L+1-i}, \tilde{b}_{L+1-i})$ decoded in time $L+i$, we decode $(a_{L+1-i}, b_{L+1-i})$, which in turn helps decoding $(\tilde{a}_{L-i}, \tilde{b}_{L-i})$. %Decoding $(\tilde{a}_0, \tilde{b}_0)$ is ignored when $i=L$.

Now let us compute an achievable rate. In stage I, we sent $(4L,2L)$ fresh forward and backward symbols. In stage II, we sent only $2L$ fresh forward symbols. This yields the desired rate in the limit of $L \rightarrow \infty$.

\begin{remark}[Exploiting Future Symbols as Side Information]
\label{remark:whyfreelunch}
Note in Fig.~\ref{fig:example2_stage1} the two types of tension: (1) forward-symbol feedback vs. backward symbols; 
(2) the other counterpart. As illustrated in Fig.~\ref{fig:example2_stage2}, our scheme leads us to resolve both tensions. This then enables us to fully utilize the remaining resource level $2 \mt - \tilde{C}_{\sf pf} = 1$ for sending the forward-symbol feedback of $C_{\sf pf}- C_{\sf no}=1$, thereby achieving $C_{\sf pf}$. Similarly we can fill up the resource holes $2n - C_{\sf pf} = 1$ with the backward-symbol feedback of $\tilde{C}_{\sf pf} - \tilde{C}_{\sf no} = 1$. This comes from the fact that our feedback scheme exploits the following as side information: (i) past received signals; (ii) users' own symbols; (iii) partially decoded symbols. While the first two were already shown to be beneficial in the prior works~\cite{SuhTse,SuhWangTse} (as well as in Example 1), the third type of information  is the newly exploited one which turns out to yield the strong interaction gain. One can view this as \emph{future} information. Recall the decoding order~\eqref{eq:decodingorder}. When decoding $(\tilde{a}_{L-1}, \tilde{b}_{L-1})$, we exploited $(a_L,b_L)$ (future symbols w.r.t. $(\tilde{a}_{L-1}, \tilde{b}_{L-1})$) as side information. A conventional belief is that feedback allows us to know only about the \emph{past}. In contrast, we discover a new viewpoint on the role of feedback. Feedback enables exploiting future information as well via retrospective decoding. $\Box$
\end{remark}

\subsection{Proof Outline}
\label{sec:proofoutline}

We categorize regimes depending on the values of channel parameters. Notice that ${\cal C} = {\cal C}_{\sf no}$ when $(\alpha \in [\frac{2}{3}, 2], \at \in [\frac{2}{3}, 2])$. Also by symmetry, it suffices to consider only five regimes - see Fig.~\ref{fig:regimes_to_check}:
\begin{align*}
&{\rm (R1) } \; \alpha >2 , \at >2; \\
&{\rm (R2) } \; \alpha \in  (0, 2/3), \at \in (0, 2/3); \\
&{\rm (R3) } \; \alpha >2, \at \in [2/3,2 ]; \\
&{\rm (R4) } \; \alpha \in ( 0, 2/3), \at \in [2/3, 2 ]; \\
&{\rm (R5) } \; \alpha \in (0, 2/3), \at > 2.
\end{align*}

\begin{figure}[t]
\begin{center}
{\epsfig{figure=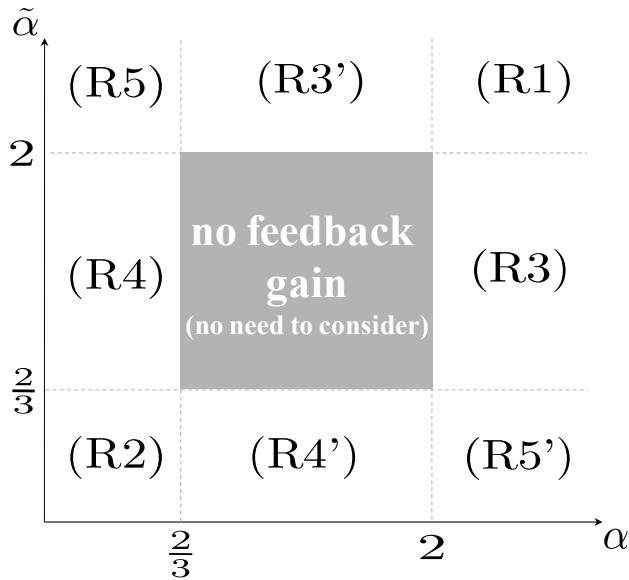, angle=0, width=0.5\textwidth}}
\end{center}
\caption{Regimes to check for achievability proof. By symmetry, it suffices to consider (R1), (R2), (R3), (R4), (R5).} \label{fig:regimes_to_check}
\end{figure}

As figured out in Fig.~\ref{fig:freeride_vs_nofreeride}, (R1) and (R2) are the ones in which there is no interaction gain. The proof builds only upon the perfect feedback scheme~\cite{SuhTse}. One thing to note here is that there are many subcases depending on whether or not available resources offered by a channel are enough to achieve the perfect feedback bound. Hence, a tedious yet careful analysis is required to cover all such subcases. On the other hand, (R3) and (R4) are the ones in which there is an interaction gain but only in one direction. So in this case, the nonfeedback scheme suffices for the backward IC while a non-trivial scheme needs to be employed for the forward IC. It turns out Scheme 1 illustrated in Example 1 plays a key role in proving the claimed achievable region. (R5) is the one in which there is an interaction gain and sometimes one can get to perfect feedback capacities. We fully utilize the ideas presented in Scheme 1 and Scheme 2 to prove the claimed rate region. One key feature to emphasize is that the idea of \emph{network decomposition} developed in~\cite{SuhGoelaGastpar:it16} is utilized to provide a conceptually simpler proof for generalization. Here we illustrate the network decomposition idea via Example 3, while leaving a detailed proof in Appendix~\ref{append:achievability_remaining}.

\begin{figure}[t]
\begin{center}
{\epsfig{figure=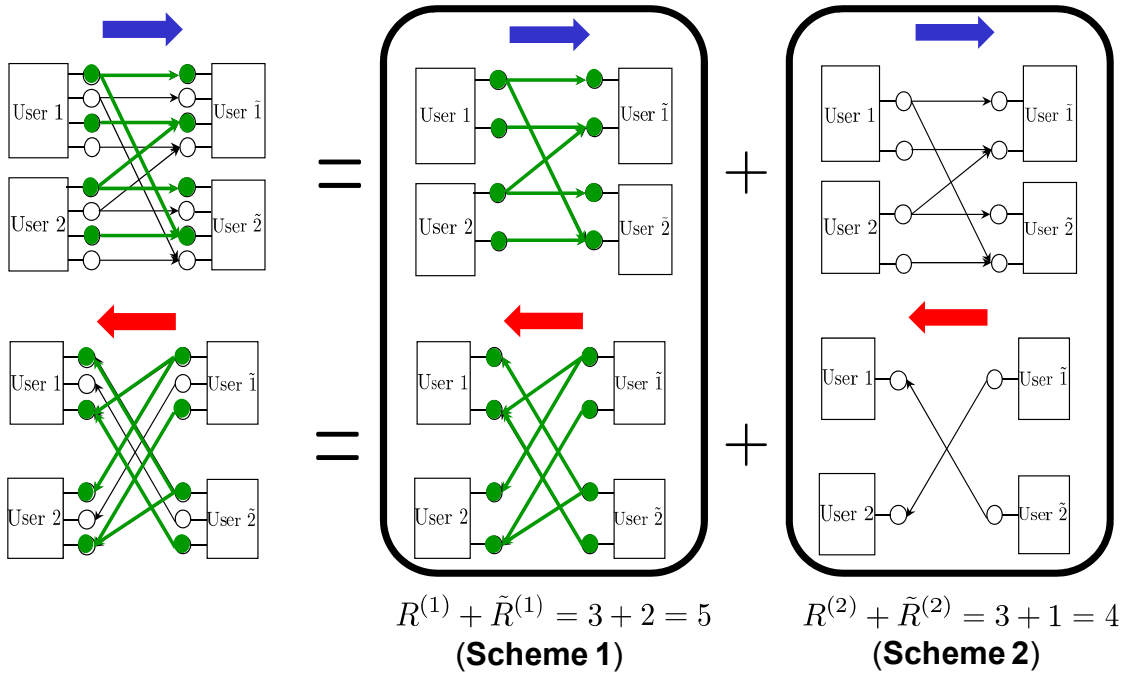, angle=0, width=0.9\textwidth}}
\end{center}
\caption{Achievaility for $(n,m)=(4,2), (\nt,\mt) = (1,3)$ via network decomposition.} \label{fig:example3}
\vspace*{-0.1in}
\end{figure}

{\bf Example 3: $(n,m)=(4,2), (\nt,\mt) = (1,3)$}:  Network decomposition relies on graph coloring. See Fig.~\ref{fig:example3}. For the forward IC, we assign a color (say green) to level 1 and the levels connected to level 1. The green-colored graph then represents a subchannel, say $(n^{(1)}, m^{(1)}) = (2,1)$, which has no overlap with the remaining uncolored subchannel $(n^{(2)}, m^{(2)}) = (2,1)$. Following the notation in~\cite{SuhGoelaGastpar:it16}, we represent this  by: $(4,2) \longrightarrow (2,1) \times (2,1)$. Similarly the backward channel can be decomposed as: $(1,3) \longrightarrow (1,2) \times (0,1)$. We then pair up one forward-subchannel $(2,1)$ and one backward-subchannel, say $(1,2)$, and apply \emph{Scheme 1} for the pair as in Fig.~\ref{fig:example1}. This gives $(R^{(1)}, \tilde{R}^{(1)}) = (3, 2)$. For the remaining pair of $(2,1)$ and $(0,1)$, we perform \emph{Scheme 2} independently. This yields $(R^{(2)}, \tilde{R}^{(2)}) = (3, 1)$. Combining these two achieves the desired rate region: $\{(R,\rt): R \leq C_{\sf pf} = 6, \rt \leq \tilde{C}_{\sf pf} = 3 \}$. $\Box$

\section{Converse Proof of Theorem~\ref{thm:capacityregion}}
\label{sec:S3_outerbound}

The first two~\eqref{eq:S3outerbound_1} and~\eqref{eq:S3outerbound_2} are the perfect-feedback bounds~\cite{SuhTse,Sahai:ITW09,SuhWangTse}. So the proof is immediate via a slight modification. The third bound~\eqref{eq:S3outerbound_3} is cutset: $R_1 + \tilde{R}_2 \leq n + \tilde{n}$ and $R_2 + \tilde{R}_1 \leq n + \tilde{n}$. The last is a new bound. For completeness, we will provide detailed proof for the cutset and perfect feedback bounds in the subsequent section. We will then derive the new bound in Section~\ref{sec:proofnovelbound}. 

\subsection{Proof of the Cutset \& Perfect Feedback Bound}
 
\textbf{Proof of~(\ref{eq:S3outerbound_3})}: Starting with Fano's inequality, we get
\begin{align*}
\begin{split}
N&(R_1 + \rt_2 - \epsilon_N) \leq I(W_1, \wt_2 ;Y_1^{N}, \tilde{W}_1, \tilde{Y}_2^N, W_2) \\
& \overset{(a)}{=} \sum H(Y_{1i}, \tilde{Y}_{2i} | \tilde{W}_1, W_2, Y_1^{i-1}, \tilde{Y}_2^{i-1}, X_{2i}) \\
& \overset{(b)}{=} \sum H(Y_{1i} | \tilde{W}_1, W_2, Y_1^{i-1}, \tilde{Y}_2^{i-1}, X_{2i})  + \sum H(\tilde{Y}_{2i} | \tilde{W}_1, W_2, Y_1^{i}, \tilde{Y}_2^{i-1}, X_{2i}, \tilde{X}_{1i}) \\
&\overset{(c)}{\leq} \sum H(Y_{1i} | X_{2i} )  + \sum H(\tilde{Y}_{2i} | \tilde{X}_{1i}) \\
& \overset{(d)}{\leq}  N (n + \nt)
\end{split}
\end{align*}
where $(a)$ follows from the fact that $(W_1, \wt_2)$ is independent of $(W_2, \tilde{W}_1)$, and $X_{2i}$ is a function of $(W_2,\tilde{Y}_2^{i-1})$; $(b)$ follows from the fact that $\tilde{X}_{1i}$ is a function of $(\wt_1, Y_1^{i-1})$; $(c)$ follows from the fact that conditioning reduces entropy; $(d)$ follows from the fact that the right-hand-side is maximized when $(X_1, X_2, \tilde{X}_1, \tilde{X}_2)$ are uniformly distributed and independent. Similarly one can show $N(R_2 + \rt_1 - \epsilon_N) \leq N( n + \tilde{n})$. If $(R_1,R_2, \rt_1, \rt_2)$ is achievable, then $\epsilon_N \rightarrow 0$ as $N$ tends to infinity. Therefore, we get the desired bound.

\textbf{Proof of~(\ref{eq:S3outerbound_1})}: Starting with Fano's inequality, we get
\begin{align*}
\begin{split}
&N(R_1 + R_2- \epsilon_N) \overset{(a)}{ \leq} I(W_1;Y_1^{N}| \tilde{W}_1, W_2, \wt_2) + I(W_2;Y_2^{N} | \tilde{W}_2, \wt_1)  \\
& = H(Y_1^{N}|\tilde{W}_1,  W_2, \wt_2) + H(Y_2^{N}| \tilde{W}_2, \wt_1)  \\
& \quad  -\left \{  H(Y_1^N, Y_2^N | \tilde{W}_1, W_2, \wt_2) -  H(Y_1^{N}|  \wt_1, \tilde{W}_2, W_2, Y_2^N) \right \} \\
& = H(Y_1^{N}|\wt_1, \tilde{W}_2, W_2, Y_2^N) -  H(Y_2^{N}|  \tilde{W}_1, W_2, \wt_2, Y_1^N) + H(Y_2^N | \tilde{W}_2, \wt_1 ) \\
& \leq H(Y_1^{N}|\wt_1, \tilde{W}_2, W_2, Y_2^N) + H(Y_2^N | \tilde{W}_2, \wt_1) \\
& \overset{(b)}{=} \sum H(Y_{1i} | \wt_1, \tilde{W}_2, W_2, Y_2^N, Y_1^{i-1},  \tilde{X}_1^i, \tilde{X}_{2i}, \tilde{Y}_2^i, X_{2i}, V_{1i}) + H(Y_2^N | \tilde{W}_2, \wt_1)  \\
& \overset{(c)}{\leq}  \sum H(Y_{1i} | V_{1i}, X_{2i} ) + \sum H(Y_{2i}) \\
& \leq N \left\{ (n-m)^+ + \max(n,m) \right \} = N \max(2n-m, m)
\end{split}
\end{align*}
where $(a)$ follows from the independence of $(W_1,W_2, \tilde{W}_1, \tilde{W}_2)$; $(b)$ follows from the fact that $\tilde{X}_{1}^{i}$ is a function of $(\wt_1, Y_1^{i-1})$, $X_{2i}$ is a function of $(W_2, \tilde{Y}_2^{i-1} )$, and $V_{1i}$ is a function of $(X_{2i}, Y_{2i})$; $(c)$ follows from the fact that conditioning reduces entropy. This completes the proof.

\subsection{Proof of a Novel Outer Bound}
\label{sec:proofnovelbound}

The proof hinges upon several lemmas stated below. The proof is streamlined with the help of a key notion, called \emph{triple mutual information} (or interaction information~\cite{InteractionInfo:54}), which is defined as 
\begin{align}
\label{eq:triplet}
I(X;Y;Z) := I (X;Y) - I(X;Y|Z).
\end{align}
It turns out that the commutative property of the notion plays a crucial role in deriving several key steps in the proof: 
\begin{align}
\label{eq:commutativity}
I(X;Y;Z) = I(X;Z;Y) = \cdots = I(Z;Y;X).
\end{align}
Using this notion and starting with Fano's inequality, we get
%\vspace*{-0.15in}
%\small
\begin{align*}
\begin{split}
&N(R_1 + R_2- \epsilon_N) \leq I(W_1;Y_1^{N}, \tilde{W}_1  ) + I(W_2;Y_2^{N}, \tilde{W}_2 )  \\
& \leq I (W_1; Y_{1}^N, V_1^N | \tilde{W}_1) + I (W_2 ; Y_{2}^N, V_2^N | \tilde{W}_2)   \\
& = \sum \left \{ I (W_1; Y_{1i}, V_{1i} |\tilde{W}_1,  Y_1^{i-1}, V_1^{i-1} )  + I (W_2 ; Y_{2i}, V_{2i} |\tilde{W}_2, Y_2^{i-1}, V_2^{i-1} )  \right \} \\
& = \sum \left \{ I (V_{1i}; W_1  | \tilde{W}_1,  Y_1^{i-1}, V_1^{i-1}) + I (Y_{1i};  W_1   | \tilde{W}_1, Y_1^{i-1}, V_1^{i}) \right. \\
& \quad \;\; \left. + I (V_{2i}; W_2 |\tilde{W}_2, Y_2^{i-1}, V_2^{i-1}) + I (Y_{2i}; W_2 |\tilde{W}_2, Y_2^{i-1}, V_2^{i})  \right \} \\
& \overset{(a)}{=} \sum \left \{  I ( Y_{1i}; W_1, \blue{ W_2, \tilde{W}_2 } |\tilde{W}_1, Y_1^{i-1}, V_1^{i}) + I (Y_{2i}; W_2, \blue{ W_1, \tilde{W}_1} |\tilde{W}_2, Y_2^{i-1}, V_2^{i}) \right. \\
& \left. + I (V_{1i}; W_1 | \tilde{W}_1,  Y_1^{i-1}, V_1^{i-1})  - I ( Y_{1i}; \blue{ W_2, \tilde{W}_2 } |W_1, \tilde{W}_1, Y_1^{i-1}, V_1^{i}) \right. \\
& \left. + I (V_{2i}; W_2 |\tilde{W}_2, Y_2^{i-1}, V_2^{i-1}) - I (Y_{2i}; \blue{ W_1, \tilde{W}_1} |W_2,\tilde{W}_2,  Y_2^{i-1}, V_2^{i}) \right \} \\
& \leq \sum \left \{ H( Y_{1i} | V_{1i} ) + H(Y_{2i} | V_{2i} ) \right. \\
& \left. + I (V_{1i}; W_1  |  \tilde{W}_1, Y_1^{i-1}, V_1^{i-1}) - I ( Y_{1i}; W_2, \tilde{W}_2 |W_1, \tilde{W}_1, Y_1^{i-1}, V_1^{i}) \right . \\
& \left. + I ( V_{2i}; W_2 | \tilde{W}_2, Y_2^{i-1}, V_2^{i-1}) -  I (Y_{2i}; W_1, \tilde{W}_1 |W_2, \tilde{W}_2, Y_2^{i-1}, V_2^{i})  \right \}
%& \overset{(b)}{\leq} \sum \left \{ H(Y_{1i} |  {V_{1i}}) + H(Y_{2i} | { V_{2i}} )  +  H( \tilde{Y}_{i1} | \tilde{V}_{1i})  + H( \tilde{Y}_{2i} | \tilde{V}_{2i}) \right \} \\
%&  \leq 2N \max (n-m, m) + 2N \max (\tilde{n}- \tilde{m}, \tilde{m})
\end{split}
\end{align*}
%\normalsize
where $(a)$ follows from a chain rule. By symmetry, we get:
\begin{align*}
\begin{split}
&N(\tilde{R}_1 + \tilde{R}_2- \epsilon_N) \leq  \sum \left \{ H( \tilde{Y}_{1i} | \tilde{V}_{1i} ) + H( \tilde{Y}_{2i} | \tilde{V}_{2i} ) \right. \\
& \left. + I (\tilde{V}_{1i}; \tilde{W}_1  | W_1, \tilde{Y}_1^{i-1}, \tilde{V}_1^{i-1}) - I ( \tilde{Y}_{1i}; W_2, \tilde{W}_2 |W_1, \tilde{W}_1, \tilde{Y}_1^{i-1}, \tilde{V}_1^{i}) \right . \\
& \left. + I ( \tilde{V}_{2i}; \tilde{W}_2 | W_2, \tilde{Y}_2^{i-1}, \tilde{V}_2^{i-1}) -  I (\tilde{Y}_{2i}; W_1, \tilde{W}_1 |W_2, \tilde{W}_2, \tilde{Y}_2^{i-1}, \tilde{V}_2^{i})  \right \}.
\end{split}
\end{align*}

Now adding the above two and using Lemma~\ref{lemma:blancebound} stated below, we get:
\begin{align*}
\begin{split}
&N(R_1 + R_2 + \tilde{R}_1 + \tilde{R}_2 - \epsilon_N) \\
& \leq \sum \left \{ H(Y_{1i} |  {V_{1i}}) + H(Y_{2i} | { V_{2i}} )  +  H( \tilde{Y}_{1i} | \tilde{V}_{1i})  + H( \tilde{Y}_{2i} | \tilde{V}_{2i}) \right \} \\
&  \leq 2N \max (n-m, m) + 2N \max (\tilde{n}- \tilde{m}, \tilde{m}).
\end{split}
\end{align*}
%If $(R_1,R_2, \rt_1, \rt_2)$ is achievable, then $\epsilon_N \rightarrow 0$ as $N$ tends to infinity. 
Hence, we get the desired bound.

\begin{lemma}
\label{lemma:blancebound}
%\small
\begin{align*}
\begin{split}
%\vspace*{-0.2in}
& \sum \left \{ I (V_{1i}; W_1 | \tilde{W}_1, Y_1^{i-1}, V_1^{i-1} ) -  I (Y_{1i}; W_2, \tilde{W}_2 | W_1, \tilde{W}_1, Y_1^{i-1}, V_1^{i} ) \right. \\
& \left . \quad +  I (V_{2i}; W_2 | \tilde{W}_2, Y_2^{i-1}, V_2^{i-1}) - I (Y_{2i}; W_1, \tilde{W}_1 | W_2, \tilde{W}_2, Y_2^{i-1}, V_2^i ) \right . \\
& \left . \quad + I (\tilde{V}_{1i}; \tilde{W}_1 | {W}_1, \tilde{Y}_1^{i-1}, \tilde{V}_1^{i-1} ) -  I (\tilde{Y}_{1i}; \tilde{W}_2, {W}_2 | \tilde{W}_1, {W}_1, \tilde{Y}_1^{i-1}, \tilde{V}_1^{i} ) \right. \\
& \left . \quad +  I (\tilde{V}_{2i}; \tilde{W}_2 | {W}_2, \tilde{Y}_2^{i-1}, \tilde{V}_2^{i-1}) - I (\tilde{Y}_{2i}; \tilde{W}_1, {W}_1 | \tilde{W}_2, {W}_2, \tilde{Y}_2^{i-1}, \tilde{V}_2^i ) \right\} \leq 0 
\end{split}
\end{align*}
%\normalsize
\end{lemma}
%\begin{remark}[Capturing the Tension]
%dddd. $\Box$
%\end{remark}

\subsection{Proof of Lemma~\ref{lemma:blancebound}}

First consider:
%the 1st and 2nd terms in summation of LHS:
%\small
\begin{align*}
\begin{split}
\scriptstyle
&\textrm{(1st and 2nd terms in summation of LHS)} \\
%&\sum \left \{ I (V_{1i}; W_1 |\tilde{W}_1, Y_1^{i-1}, V_1^{i-1} ) - I (Y_{1i}; W_2, \tilde{W}_2 | W_1, \tilde{W}_1, Y_1^{i-1}, V_1^{i}) \right \} \\
&\overset{(a)}{=} \sum \left \{ I (V_{1i}; W_1 | \tilde{W}_1,Y_1^{i-1}, V_1^{i-1} ) - I (Y_{1i}; W_2, \tilde{W}_2, \blue{\tilde{Y}_1^{i}} | W_1, \tilde{W}_1, Y_1^{i-1}, V_1^{i}) \right. \\
&\overset{(b)}{=} \sum \left \{ I (V_{1i}, \blue{ \tilde{V}_{1i} }; W_1 |  \tilde{W}_1, Y_1^{i-1}, V_1^{i-1}, \blue{ \tilde{V}_{1}^{i-1} }) \right. \\
& \left. \quad - I ( Y_{1i} ; {\tilde{Y}_1^{i}} | W_1, \tilde{W}_1, Y_1^{i-1}, V_1^{i},\blue{\tilde{V}_1^{i}}) - I (Y_{1i}; W_2, \tilde{W}_2 | W_1, \tilde{W}_1, {\tilde{Y}_1^{i}}, Y_1^{i-1}) \right\} \\
&\overset{(c)}{=} \sum \left \{ I (V_{1i}, \tilde{V}_{1i} ; W_1 |  \tilde{W}_1, V_1^{i-1},  \tilde{V}_{1}^{i-1} )  - I (V_{1i}, \tilde{V}_{1i} ; W_1; \blue{Y_1^{i-1}} | \tilde{W}_1, V_1^{i-1},  \tilde{V}_{1}^{i-1} ) \right. \\
& \left.\quad - I ( Y_{1i} ; {\tilde{Y}_1^{i}} | W_1, \tilde{W}_1, Y_1^{i-1}, V_1^{i},\tilde{V}_1^{i}) - I (Y_{1i}; W_2, \tilde{W}_2 | W_1, \tilde{W}_1, {\tilde{Y}_1^{i}}, Y_1^{i-1}) \right\} 
\end{split}
\end{align*}
%\normalsize
where $(a)$ follows from the fact that $\tilde{Y}_1^{i}$ is a function of $(W_1,\tilde{W}_1,W_2,\tilde{W}_2)$; $(b)$ follows from the fact that $\tilde{V}_{1}^{i}$ is a function of $(\tilde{W}_1, Y_1^{i-1})$; and $(c)$ is due to the definition of triple mutual information~\eqref{eq:triplet}.
%; and $(d)$ follows from Lemma~\ref{lemma:inside_detail1} (see below).

Using Lemma~\ref{lemma:inside_detail1} stated at the end of this section, we get:
\begin{align*}
\begin{split}
&\textrm{(1st and 2nd terms in summation of LHS)} \\
&\leq \sum \left \{ I (V_{1i}, \tilde{V}_{1i} ; W_1 |  \tilde{W}_1, V_1^{i-1},  \tilde{V}_{1}^{i-1} ) 
 + I (\tilde{Y}_{1i}; Y_1^{i-1}  | W_1, \tilde{W}_1, \tilde{Y}_{1}^{i-1}, \tilde{V}_1^{i})  \right . \\
& \left.\quad - I (  \tilde{V}_{1i}; W_1, \tilde{Y}_1^{i-1} | \tilde{W}_1, V_1^{i-1}, \tilde{V}_1^{i-1} )  - I (Y_{1i}; W_2, \tilde{W}_2 | W_1, \tilde{W}_1, {\tilde{Y}_1^{i}}, Y_1^{i-1}) \right\}.
\end{split}
\end{align*}
Now combining this with the 5th and 6th terms in summation of LHS gives:
%\small
\begin{align*}
\begin{split}
&\textrm{(1st, 2nd, 5th and 6th terms of LHS in the claimed bound)} \\
&\overset{(a)}{\leq} \sum \left \{ I (V_{1i}, \tilde{V}_{1i} ; W_1 |  \tilde{W}_1, V_1^{i-1},  \tilde{V}_{1}^{i-1} ) 
 + I (\tilde{Y}_{1i}; Y_1^{i-1}  | W_1, \tilde{W}_1, \tilde{Y}_{1}^{i-1}, \tilde{V}_1^{i})  \right . \\
& \left.\quad - I (  \tilde{V}_{1i}; W_1, \tilde{Y}_1^{i-1} | \tilde{W}_1, V_1^{i-1}, \tilde{V}_1^{i-1} )  - I (Y_{1i}; W_2, \tilde{W}_2 | W_1, \tilde{W}_1, {\tilde{Y}_1^{i}}, Y_1^{i-1}) \right\} \\
& +\sum \left\{ I (\tilde{V}_{1i}; \tilde{W}_1 | {W}_1, \tilde{Y}_1^{i-1},\blue{V_1^{i-1}} ,\tilde{V}_1^{i-1} )  -  I (\tilde{Y}_{1i}; \tilde{W}_2, {W}_2, \blue{Y_1^{i-1}} | \tilde{W}_1, {W}_1, \tilde{Y}_1^{i-1}, \tilde{V}_1^{i} )\right\}\\
&\overset{(b)}{\leq} \sum \left \{ I (V_{1i}, \tilde{V}_{1i} ; W_1 |  \tilde{W}_1, V_1^{i-1},  \tilde{V}_{1}^{i-1} ) \right. \\
& \left.\quad - I (  \tilde{V}_{1i}; W_1, \tilde{Y}_1^{i-1} | \tilde{W}_1, V_1^{i-1}, \tilde{V}_1^{i-1} )  - I (Y_{1i}; W_2, \tilde{W}_2 | W_1, \tilde{W}_1, {\tilde{Y}_1^{i}}, Y_1^{i-1}) \right\} \\
& +\sum \left\{ I (\tilde{V}_{1i}; \tilde{W}_1 ,\blue {{W}_1, \tilde{Y}_1^{i-1}}|V_1^{i-1}, \tilde{V}_1^{i-1} )   -  I (\tilde{Y}_{1i}; \tilde{W}_2, {W}_2| \tilde{W}_1, {W}_1, \tilde{Y}_1^{i-1}, {Y_1^{i-1}} )\right\}\\
&\overset{(c)}{\leq} \sum \left \{ I (V_{1i}, \tilde{V}_{1i} ; W_1 |  \tilde{W}_1, V_1^{i-1},  \tilde{V}_{1}^{i-1} ) \right. \\
& \left. \quad  - I (\blue{Y_{1i},\tilde{Y}_{1i}}; W_2, \tilde{W}_2 | W_1, \tilde{W}_1, {\tilde{Y}_1^{i-1}}, Y_1^{i-1})  + I (\blue{V_{1i}},\tilde{V}_{1i}; \blue{\tilde{W}_1}|V_1^{i-1}, \tilde{V}_1^{i-1} )  \right \} \\
%& \overset{(e)}{ = } \sum \left \{  I (V_{1i},  \tilde{V}_{1i} ; W_1, \tilde{W}_1 |   V_1^{i-1},  \tilde{V}_{1}^{i-1} )  \right . \\
%& \left.\quad  - I (Y_{1i},\tilde{Y}_{1i}; W_2, \tilde{W}_2 | W_1, \tilde{W}_1, {\tilde{Y}_1^{i-1}}, Y_1^{i-1}) \right\} \\
%&\overset{}{=}I(V_1^N,\tilde{V}_1^{N};W_1,\tilde{W}_1)
%-I(Y_1^N,\tilde{Y}_1^N;W_2,\tilde{W}_2|W_1,\tilde{W}_1)\\
&\overset{(d)}{=}I(V_1^N,\tilde{V}_1^{N};W_1,\tilde{W}_1)
-I(Y_1^N,\tilde{Y}_1^N,\blue{V_2^N,\tilde{V}_2^{N}};W_2,\tilde{W}_2|W_1,\tilde{W}_1)\\
&\overset{}{\leq}I(V_1^N,\tilde{V}_1^{N};W_1,\tilde{W}_1)
-I({V_2^N,\tilde{V}_2^{N}};W_2,\tilde{W}_2|W_1,\tilde{W}_1)
\end{split}
\end{align*}
%\normalsize
where $(a)$ follows from the fact that $V_{1}^{i-1}$ and $Y_1^{i-1}$ are functions of $(W_1,\tilde{Y}_1^{i-1})$ and $(W_1,W_2,\tilde{W}_1,\tilde{W}_2)$, respectively; $(b)$ follows from a chain rule (applied on the last term) and the non-negativity of mutual information; $(c)$ follows from a chain rule (combining the 2nd and 4th terms; also combining the 3rd and 5th terms) and the non-negativity of mutual information; $(d)$ follows from a chain rule (combining the 1st and 3rd terms) and the fact that $(V_2^N,\tilde{V}_2^{N})$ is a function of $(W_1,\tilde{W}_1,Y_1^N,\tilde{Y}_1^N)$. 

Applying the same to the 3rd, 4th, 7th and 8th terms in summation of LHS, we get: 
%\small
\begin{align*}
\begin{split}
&\textrm{(LHS in the claimed bound)} \\
&\overset{}{\leq}I(V_1^N,\tilde{V}_1^{N};W_1,\tilde{W}_1)
-I({V_2^N,\tilde{V}_2^{N}};W_2,\tilde{W}_2|W_1,\tilde{W}_1)\\
&+I(V_2^N,\tilde{V}_2^{N};W_2,\tilde{W}_2)
-I({V_1^N,\tilde{V}_1^{N}};W_1,\tilde{W}_1|W_2,\tilde{W}_2)  \\
&\overset{}{\leq}I(\blue{W_2,\tilde{W}_2},V_1^N,\tilde{V}_1^{N};W_1,\tilde{W}_1)
-I({V_2^N,\tilde{V}_2^{N}};W_2,\tilde{W}_2|W_1,\tilde{W}_1)\\
&+I(\blue{W_1,\tilde{W}_1},V_2^N,\tilde{V}_2^{N};W_2,\tilde{W}_2)
-I({V_1^N,\tilde{V}_1^{N}};W_1,\tilde{W}_1|W_2,\tilde{W}_2) = 0.
\end{split}
\end{align*}
This completes the proof.
%\normalsize
%where $(a)$ follows from the fact that $I(W_1,\tilde{W}_1;W_2,\tilde{W}_2)=0$.   
\begin{lemma}
\label{lemma:inside_detail1}
%\small
\begin{align*}
\begin{split}
& - \sum \left \{ I (V_{1i}, \tilde{V}_{1i}; W_1; Y_1^{i-1} | \tilde{W}_1, V_1^{i-1}, \tilde{V}_{1}^{i-1})  +I (Y_{1i}; \tilde{Y}_{1}^{i} | W_1, \tilde{W}_1, Y_1^{i-1}, V_1^i,  \tilde{V}_1^{i} ) \right \} \\
&\leq  \sum \left \{ I (\tilde{Y}_{1i}; Y_1^{i-1} | W_1, \tilde{W}_1, \tilde{Y}_{1}^{i-1}, \tilde{V}_1^{i})  - I (  \tilde{V}_{1i}; W_1, \tilde{Y}_1^{i-1} | \tilde{W}_1, V_1^{i-1}, \tilde{V}_1^{i-1} ) \right \}.
\end{split}
\normalsize
\end{align*}
\end{lemma}
\begin{proof}
See Section~\ref{sec:proofoflemma2}.
\end{proof}

\subsection{Proof of Lemma~\ref{lemma:inside_detail1}}
\label{sec:proofoflemma2}

\begin{align*}
\begin{split}
& - \sum \left \{ I (V_{1i}, \tilde{V}_{1i}; W_1;Y_1^{i-1} | \tilde{W}_1, V_1^{i-1}, \tilde{V}_{1}^{i-1})  +I (Y_{1i}; \tilde{Y}_{1}^{i} | W_1, \tilde{W}_1,Y_1^{i-1}, V_1^i,  \tilde{V}_1^{i} ) \right \} \\
& \overset{(a)}{=} \sum \left \{ I (V_{1i}, \tilde{V}_{1i}; Y_1^{i-1}  | W_1, \tilde{W}_1, V_1^{i-1}, \tilde{V}_{1}^{i-1}) \right. \\
& \left . \qquad \;\;\;\; - I (V_{1i}, \tilde{V}_{1i}; Y_1^{i-1}  | \tilde{W}_1,V_1^{i-1}, \tilde{V}_{1}^{i-1}) - I (Y_{1i}; \tilde{Y}_{1}^{i} | W_1, \tilde{W}_1, Y_1^{i-1}, V_1^i,  \tilde{V}_1^{i} ) \right \} \\
& \overset{(b)}{\leq} \sum \left \{ I (V_{1i}, \tilde{V}_{1i}; Y_1^{i-1}  | W_1, \tilde{W}_1,V_1^{i-1}, \tilde{V}_{1}^{i-1}) \right. \\
& \left . \qquad \;\;\;\; - I( \tilde{V}_{1i};Y_1^{i-1} | \tilde{W}_1, V_1^{i-1}, \tilde{V}_{1}^{i-1}) - I (Y_{1i}; \tilde{Y}_{1}^{i} | W_1, \tilde{W}_1,Y_1^{i-1}, V_1^i,  \tilde{V}_1^{i} ) \right \} \\
& \overset{(c)}{=} \sum \left \{ I ( \tilde{Y}_1^{i}, \tilde{V}_{1i}; Y_1^{i-1}  | W_1, \tilde{W}_1, V_1^{i-1}, \tilde{V}_{1}^{i-1})  - I ( \tilde{Y}_1^{i}; Y_1^{i-1} | W_1, \tilde{W}_1, V_1^i, \tilde{V}_1^{i} ) \right . \\
& \left . \qquad \;\;\;\; - I( \tilde{V}_{1i};Y_1^{i-1} | \tilde{W}_1, V_1^{i-1}, \tilde{V}_{1}^{i-1})  - I (Y_{1i}; \tilde{Y}_{1}^{i} | W_1, \tilde{W}_1, Y_1^{i-1}, V_1^i,  \tilde{V}_1^{i} ) \right \} \\
& \overset{(d)}{=} \sum \left \{ I ( \tilde{Y}_1^{i}, \tilde{V}_{1i}; Y_1^{i-1}  | W_1, \tilde{W}_1, V_1^{i-1}, \tilde{V}_{1}^{i-1}) \right. \\
& \left. \qquad \;\;\;\; - I ( \tilde{Y}_1^{i}; Y_1^{\blue{i}} | W_1, \tilde{W}_1, V_1^i, \tilde{V}_1^{i} )  - I( \tilde{V}_{1i};Y_1^{i-1} | \tilde{W}_1, V_1^{i-1}, \tilde{V}_{1}^{i-1}) \right \} \\
& \overset{(e)}{=} \sum \left \{ I ( \tilde{Y}_{1i}, \tilde{V}_{1i}; Y_1^{i-1}  | W_1, \tilde{W}_1, \tilde{Y}_1^{i-1},  V_1^{i-1}, \tilde{V}_{1}^{i-1}) + I ( \tilde{Y}_1^{i-1}; Y_1^{i-1} | W_1, \tilde{W}_1,V_1^{i-1}, \tilde{V}_1^{i-1}) \right. \\
& \left. \qquad \;\;\;\; -  I ( \tilde{Y}_1^{i}; Y_1^{i} | W_1, \tilde{W}_1, V_1^i, \tilde{V}_1^{i} )  - I( \tilde{V}_{1i};Y_1^{i-1} | \tilde{W}_1, V_1^{i-1}, \tilde{V}_{1}^{i-1}) \right \} \\
& \overset{(f)}{\leq} \sum \left \{ I ( \tilde{Y}_{1i}, \tilde{V}_{1i}; Y_1^{i-1}  | W_1, \tilde{W}_1, \tilde{Y}_1^{i-1},  V_1^{i-1}, \tilde{V}_{1}^{i-1}) - I( \tilde{V}_{1i};Y_1^{i-1} | \tilde{W}_1, V_1^{i-1}, \tilde{V}_{1}^{i-1}) \right \} \\
%& \blue{slack: - I (\tilde{Y}_1^{N-1}; Y_1^N | W_1, V_1^N, \tilde{V}_1^{N-1} )}  \\
& \overset{(g)}{=} \sum \left \{ I ( \tilde{Y}_{1i} ; Y_1^{i-1}  | W_1, \tilde{W}_1, \tilde{Y}_1^{i-1},  V_1^{i-1}, \blue{\tilde{V}_{1}^{i}})  - I( \tilde{V}_{1i}; Y_1^{i-1};\blue{W_1, \tilde{Y}_1^{i-1}} | \tilde{W}_1, V_1^{i-1}, \tilde{V}_{1}^{i-1}) \right \} \\
& \overset{(h)}{=} \sum \left \{ I ( \tilde{Y}_{1i} ; Y_1^{i-1}  | W_1, \tilde{W}_1, \tilde{Y}_1^{i-1} ,\tilde{V}_{1}^{i}) - I( \tilde{V}_{1i}; W_1, \tilde{Y}_1^{i-1} | \tilde{W}_1, V_1^{i-1}, \tilde{V}_{1}^{i-1}) \right \} \\
\end{split}
\end{align*}
where $(a)$ follows from the definition of triple mutual information~\eqref{eq:triplet}; $(b)$ follows from the non-negativity of mutual information; $(c)$ follows from a chain rule and the fact that $V_{1i}$ is a function of $(W_1, \tilde{Y}_1^{i})$ (see below)
\begin{align*}
&\sum  I (V_{1i}, \tilde{V}_{1i}; Y_1^{i-1}  | W_1, \tilde{W}_1,V_1^{i-1}, \tilde{V}_{1}^{i-1})  \\
 =& \sum I (V_{1i}, \blue{\tilde{Y}_1^i }, \tilde{V}_{1i}; Y_1^{i-1}  | W_1, \tilde{W}_1,V_1^{i-1}, \tilde{V}_{1}^{i-1})  - \sum I ( \blue{\tilde{Y}_1^i }; Y_1^{i-1}  | W_1, \tilde{W}_1,V_1^{i}, \tilde{V}_{1}^{i}) \\
 = & \sum I (\tilde{Y}_1^i , \tilde{V}_{1i}; Y_1^{i-1}  | W_1, \tilde{W}_1,V_1^{i-1}, \tilde{V}_{1}^{i-1})  - \sum I ( \tilde{Y}_1^i; Y_1^{i-1}  | W_1, \tilde{W}_1,V_1^{i}, \tilde{V}_{1}^{i}); 
\end{align*} 
$(d)$ follows from a chain rule (combining the 2nd and 4th terms); $(e)$ follow from a chain rule (applying to the 1st term); $(f)$ follows from 
\begin{align*}
&\sum_{i=1}^{N} I ( \tilde{Y}_1^{i-1}; Y_1^{i-1} | W_1, \tilde{W}_1, V_1^{i-1}, \tilde{V}_1^{i-1})  \\ 
=& \sum_{i=0}^{N-1} I ( \tilde{Y}_1^{i}; Y_1^{i} | W_1, \tilde{W}_1, V_1^{i}, \tilde{V}_1^{i}) \leq \sum_{i=1}^{N} I ( \tilde{Y}_1^{i}; Y_1^{i} | W_1, \tilde{W}_1, V_1^{i}, \tilde{V}_1^{i});
\end{align*} 
$(g)$ follows from a chain rule and the definition of triple mutual information; $(h)$ follows from the fact that $V_1^{i-1}$ and $\tilde{V}_{1i}$ are a function of $(W_1, \tilde{Y}_1^{i-2})$ and $(\tilde{W}_1,Y_1^{i-1})$, respectively.

\section{Discussion}
\label{section:Comparison}

%\subsection{Free Ride of Feedback Transmission}
%\label{disc:freeride}
%
%Recall in the example of $(n,m,\nt,\mt)=(2,1,1,3)$ that we can achieve $(R, \tilde{R}) = (2,3) = (C_{\sf nf}, C_{\sf pf})$. See Fig.~\ref{fig:freeride}. This means that we use the forward channel for free in sending feedback signals w.r.t. backward messages. Now one natural question that we can ask is: Can we achieve $(C_{\sf pf}, \tilde{C}_{\sf pf})$ in the forward and backward channels respectively in which both channels have capacity gains due to feedback? In other words, can we do \emph{free rides in both directions} in sending feedback signals? Can feedback come for free completely? Our outer bound in Theorem~\ref{thm:S3outerbound} suggests that there is hope for the free-ride in both directions, while our scheme in in Theorem~\ref{thm:S3innerbound} does not ensure the achievability. Whether or not we can achieve the free ride in both directions remains open.

\subsection{System Implication}
\label{disc:systemimplication}

\begin{figure}[t]
\begin{center}
{\epsfig{figure=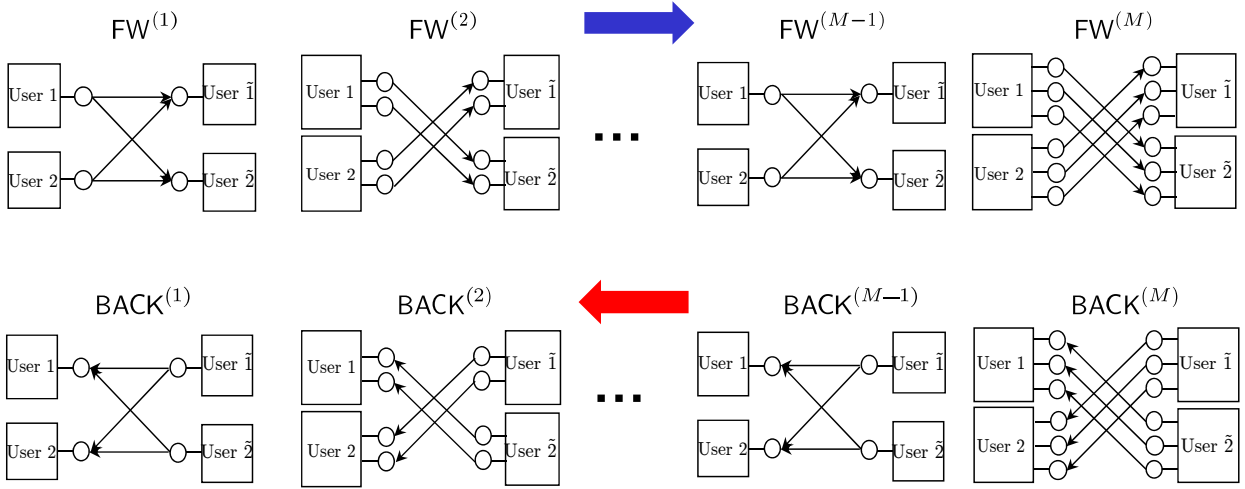, angle=0, width=1.0\textwidth}}
\end{center}
\caption{Two-way parallel ICs. The rich diversity on channel gains across many parallel subchannels can often occur in broadband systems.} \label{fig:broadband}
\end{figure}

As suggested in Fig.~\ref{fig:freeride_vs_nofreeride}, an interaction gain occurs when forward and backward ICs are somewhat different. This asymmetry occurs naturally in FDD systems where the forward and backward channels are on completely different bands. Even in TDD systems, the asymmetry can occur since the forward and backward channels can be on different subcarriers or different coherent time. Also one can create this asymmetry by opportunistically pairing subbands for the forward and backward transmissions. While this asymmetry is not likely to occur in narrowband systems, it can often occur in broadband systems where there are a multitude of subchannels with a wide dynamic range of channel gains. For example, in 4G-LTE and 5G systems, one can easily expect rich diversity on channel gains, since an operating bandwidth of the systems is much larger than coherence bandwidth of typical wireless channels (around the order of 0.1 MHz).

Fig.~\ref{fig:broadband} illustrates an example which can represent such scenario where there are a variety of parallel subchannels. Our results suggest that pairs of $({\sf FW}^{(1)}, {\sf BACK}^{(2)})$ and $({\sf FW}^{(2)}, {\sf BACK}^{(1)})$, for instance, can provide a significant gain with interaction. Another interesting observation is that even though forward-and-backward parallel ICs are identical, there exist many pairs of forward-backward subchannels that can yield capacity improvements. In tomorrow's communication systems, a broader system bandwidth is expected to support a variety of multimedia services. Hence, it is believed that our results will provide detailed guidelines as to how to design future communication systems.

%%%%%%%%%%%%%%%%%%%%%%%%%%%%%%%%%%%%%%%%%%%%%%%%%%%%%%%%%%%%%%%%%%%%%%%%%%%%%%%%%%%%%%%%%%
%In Fig.~\ref{fig:netgainregime},
%We have seen that there is net feedback gain when the forward IC well matches with the backward IC. We can also check that the gain is significant when there is strong asymmetry between $\alpha$ and $\tilde{\alpha}$. While this strong asymmetry is not likely to occur in narrowband systems, it can often occur in broadband systems where there are a multitude of subchannels with a wide dynamic range of channel gains. For example, in 3GPP-LTE and WiMAX systems, we can easily expect rich diversity on channel gains, since an operating bandwidth of the systems (around 20 MHz) is much larger than coherence bandwidth of typical wireless channels (around the order of 0.1 MHz).

%Good matching between the forward and backward channels can often occur when there are a variety of parallel subchannels, and this channel diversity can easily come from broadband systems.
%%(see Fig.~\ref{fig:broadband}).
%%Notice that broadband systems are likely to have many parallel subchannels with a dynamic range of channel variations.
%This shows the potential of our feedback idea to broadband systems, since it provides a guideline for good pairs of forward-backward channels.
%%%%%%%%%%%%%%%%%%%%%%%%%%%%%%%%%%%%%%%%%%%%%%%%%%%%%%%%%%%%%%%%%%%%%%%%%%%%%%%%%%%%%%%%%%%

%{\bf Full duplex vs. half duplex}
In this paper, we investigate the benefit of interaction for a full duplex system. This is only for illustrative purpose. As suggested in Remarks~\ref{remark:scheme1} and~\ref{remark:whyfreelunch}, the nature of the interaction gain comes from exploiting the past received signals, partially decoded symbols and users' own information as side information. This nature is not limited to the full duplex system. So one can readily see that the interaction gain carries over to the half duplex system. While the detailed capacity region of the half-duplex system is distinct, the channel regimes in which feedback offers a gain remain unchanged. In other words, we have the same picture as in Fig.~\ref{fig:freeride_vs_nofreeride}.

\subsection{Translation to the Gaussian Channel}
\label{disc:translationtoGaussian}

The deterministic-channel achievability proposed in this work gives insights into an achievable scheme in the Gaussian channel. This is inspired by several observations that can be made from Scheme 1 (see Example 1 in Fig.~\ref{fig:example1}) and Scheme 2 (see Example 2 in Figs.~\ref{fig:example2_stage1} and~\ref{fig:example2_stage2}).

\emph{(Extracting feedback signals $\rightarrow $ quantize-\&-binning)}: 
%The first observation is regarding how to construct  feedback signals. 
Note in Fig.~\ref{fig:example1} that the fedback signal $a \oplus B$ at user $\tilde{1}$ can be interpreted as a quantized version of the received signal $(A,a \oplus B)$ at the level below the clean signal $A$. This motivates the use of \emph{quantize-and-binning}~\cite{Salman:IT11,Lim:it11} in the Gaussian channel. There are two points to make. The first is that the \emph{binning} scheme~\cite{ElGamalBook} can be employed solely without quantization in this example. Binning the received signal $(A, a \oplus B)$ may construct a linear combination of the two components: $a \oplus B \oplus A$. 
%The distinction is that it has now contain $A_1$ as well. 
The distinction is then user $\tilde{1}$ feeds back $\tilde{A} \oplus (a \oplus B \oplus A)$ instead. Nonetheless, user 1 can still get $\tilde{A}$ of interest, as $A$ is also known. On the other hand, user 2 obtains $a \oplus A \oplus \tilde{A}$ instead of $a \oplus \tilde{A}$. This is not an issue either. We can still achieve interference alignment and neutralization in forward transmission at stage II. User 1 sending $a \oplus A \oplus \tilde{A}$ on top (with the help of the decoded symbol $\tilde{A}$) and user 2 sending $B' \oplus (a \oplus A \oplus \tilde{A})$, user $\tilde{2}$ can still obtain $B'$ interference-free. Also user $\tilde{1}$ can get $a$ with the help of $A$ which has already been received in stage I. The second point to note is that the binning-only approach might not work properly for other channel parameter regimes. This is because mixing all the equations may include some undesirable symbols that prevent the optimal transmission. In that case, both quantization and binning are desired to be employed with a careful choice of a quantization level, set to exclude undesirable symbols. 
%should be carefully chosen so that itdepending on channel parameter regimes. In this example, we quantize the received signals below the level of the clean signal. 
%%, since the private signal is not desirable at the other user. 
%For other regimes, a careful choice needs be made on the quantization level accordingly, in order to well mimic the deterministic-channel achievability.

\emph{(XORing with interference neutralization $\rightarrow$ superposition with dirty paper coding)}: Observe in Fig.~\ref{fig:example1} that the fedback signal $a \oplus B$ is XORed with a backward symbol $\tilde{A}$. This motivates the use of \emph{superposition coding} in the Gaussian channel. On the other hand, user $\tilde{2}$ sends $B$ on bottom for interference neutralization. To this end, we employ quantization scheme for extracting $B$ from the received signal $(B,b \oplus A)$ and utilize \emph{dirty paper coding}~\cite{Costa:it83} for nulling. 

\emph{(Interference alignment and neutralization $\rightarrow$ structured coding)}: Note at the second stage in Fig.~\ref{fig:example1} that user 2 computes the XOR of $b$ (its own symbol) and $\tilde{B}$ (decoded from the received signal in backward transmission) and then sends the XOR on a proper level (the top level) for nulling. This motivates the use of \emph{structured coding}~\cite{NazerGastpar:it11}, as computation needs to be made across appropriate symbols and the computed signal should be placed in a structured location for nulling.  

%Partial decoding, Markov encoding, Han-Kobayashi message splitting, Structured coding

%Second, the $b_1$ sent by user 1 at stage 2 can be interpreted as a block-Marov-encoded signal of $a_1 \oplus b_1$ conditioned on the previously-sent information $a_1$. Lastly, $a_3$ is new fresh information superimposed on $b_1$. These observations motivate us to employ quantize-map-and-forward
%%~\cite{Salman:IT11,LimKimElGamalChung:it11}
%for feedback, and block Markov encoding and superposition schemes at forward-message senders. 

\emph{(Retrospective decoding)}: To be the best of our knowledge, this is a novel feature that has never been introduced in network information theory literature. Hence, it requires a new achievability technique which includes a careful decoding order as well as sets proper symbols to decode for each time slot. Also note that decoded symbols in an intermediate time slot are part of the entire symbols. See Fig.~\ref{fig:example2_stage2} for instance. Here $a_L$ (a decoded symbol in time $L+2$) is part of the entire symbols $(A_L, a_L)$ sent in time $L$. Hence, this scheme needs to be properly combined with Han-Kobayashi message splitting~\cite{HanKoba:it81}.

\subsection{Unified Achievability}
The noisy network coding~\cite{Lim:it11} together with Han-Kobayashi message splitting is a fairly generic scheme that yields reasonably good performances for a variety of multi-user channels. It implements many achievablility techniques such as quantize-and-binning and superposition coding. However, it has a room for improvement as it does not incorporate dirty paper coding and structured coding.
%; hence it suffers from performance degradation in interference networks. 
An effort has been made by Nazer-Gastpar~\cite{NazerGastpar:it11} for implementing structured codes. 

But this approach still has a room for improvement, as it does not allow for the key operation that appears in our achievability: \emph{retrospective decoding}. As suggested in Example 2, the key operation seems required for achieving the optimal performance. There seems no way to achieve the perfect feedback bound without an intermediate decoding of partial symbols which admits a carefully-designed backward ordering. One future work of interest is to develop a generic achievable scheme that can be applied to general discrete memoryless networks as well as unifies all of the techniques mentioned earlier: (1) quantize-and-binning; (2) superposition coding (or block Markov coding); (3) structured coding; (4) Han-Kobayashi message-splitting; (5) retrospective decoding. This development is expected to open the door to characterizing and/or approximating many of interesting interference networks.

%t facilitates decoding other symbols in later time. 

%
\subsection{Unified Converse}

In this work, we develop a new converse technique which well captures the tension between feedback and independent message transmissions. Hence, unlike the prior upper bounds such as cutset~\cite{ElGamalBook}, genie-aided bounds~\cite{ElGamal:it82,Kramer:it02, SuhTse, Rini:it11, WangTse:it11, Vinod:it11, Sahai:ITW09}, generalized network sharing bounds~\cite{Kamath:allerton14}, it gives rise to the tight capacity characterization of interactive multi-user channels like the two-way IC. Encouragingly, our novel bound~\eqref{eq:S3outerbound_4} subsumes the following bounds as special cases: the nonfeedback-case counterpart $R_1 + R_2 \leq H(Y_1|V_1) + H(Y_2 |V_2)$~\cite{ElGamal:it82}; the rate-limited-feedback-case counterpart $R_1 + R_2 \leq H(Y_1|V_1) + H(Y_2|V_2) + C_{ {\sf FB}1}^{\sf bit pipe} + C_{ {\sf FB}2}^{\sf bit pipe}$~\cite{AlirezaSuhAves}. Here $C_{ {\sf FB}i}^{\sf bit pipe}$ denotes the capacity of the bit-piped feedback link that connects user $\tilde{i}$ to user $i$. One future work of interest is to extend this bound to arbitrary discrete memoryless networks in which many nodes interact with each other.

\subsection{Role of Interaction in General Networks}
\label{disc:extensiontogeneralnetworks}

%{\bf Extension 2}: It would be interesting to find communication scenarios in which interaction provide a gain in sum capacity for the regime of $( \frac{2}{3} \leq \alpha \leq 2, \frac{2}{3} \leq \at \leq  2)$ or for the regime of $( \alpha \geq 2, \at \geq  2)$?
This work focuses on an interference channel setting in which each user wishes to deliver its own message to its counterpart. As mentioned earlier, the nature of interaction gain is not limited to this particular setting. So it would be interesting to explore the role of interaction for a variety of different settings. While initial efforts along this research direction have been made for a multicast channel setting~\cite{SuhGoelaGastpar:isit12}, function computation settings~\cite{SuhGastpar:isit13,ShinSuh:allerton14}, and multi-hop network settings~\cite{ChenOzgurDiggavi:allerton14}, an explicit comparison between non-interactive vs interactive scenarios was not made yet. One research direction of interest is to investigate the capacity regions of such channels, thereby discovering two-way scenarios  in which one can achieve a huge interaction gain.

\section{Conclusion}
\label{sec:conclusion}

We characterized the capacity region of the two-way deterministic IC. As a consequence, we discovered an interesting fact that one can even get to perfect feedback capacities in both directions. In the process of obtaining this result, we found a new role of feedback: Feedback enables exploiting even the future information as side information via retrospective decoding. 
%The channel asymmetry, which yields an interaction gain (see Fig.~\ref{fig:freeride_vs_nofreeride}), often occurs in broadband systems where there are many subchannels exhibiting rich diversity on channel gains. Hence, pairing appropriate forward and backward ICs, one can easily obtain such a significant gain. 
Our future work includes: (1) Translating to the Gaussian channel; (2) Discovering other two-way scenarios in which one can achieve a huge interaction gain; (3) Generalizing our new achievability to broader network contexts.

\newpage

\appendices

\section{Achievability Proof of Theorem~\ref{thm:capacityregion}: Generalization to Arbitrary $(n,m,\nt,\mt)$}
\label{append:achievability_remaining}

One key idea for generalization is to use the network decomposition in~\cite{SuhGoelaGastpar:it16} (also illustrated via Example 3 in Fig.~\ref{fig:example3}). The idea provides a conceptually simpler proof by decomposing a general $(n,m)$ (or $(\nt,\mt)$) channel into multiple elementary subchannels and taking a proper matching across forward and backward subchannels. See Theorem~\ref{thm:ND} (stated below) for the identified elementary subchannels, which we will use to complete the proof in the subsequent subsections. 
%and see Lemma~\ref{lemma:buildingblocks} (stated below) for achievable rates for some pairs of fundamental building blocks. 

\begin{theorem}[Network Decomposition~\cite{SuhGoelaGastpar:it16}]
\label{thm:ND} For an arbitrary $(n,m)$ channel, the following network decomposition holds:
\begin{align}
&(n,m) \longrightarrow (1,0)^{ n-2m} \times (2,1)^m, \qquad \;  \alpha \in [0, 1/2]; \label{eq:ND1} \\
& (n,m) \longrightarrow  (2,1)^{2n-3m} \times (3,2)^{2m-n}, \;\;  \alpha \in [1/2,2/3]; \label{eq:ND2} \\
& (n,m) \longrightarrow  (0,1)^{m-2n} \times (1,2)^{n},  \qquad \;\;  \alpha \geq 2. \label{eq:ND3}
\end{align}
Here the symbol $\times$ indicates the concatenation of orthogonal channels and $(i,j)^\ell$ denotes the $\ell$-fold concatenation of the $(i,j)$ channel.
\end{theorem}

\subsection{Proof of (R1) $\alpha >2, \at >2 $ \& (R2) $\alpha \in (0, \frac{2}{3}), \at \in (0, \frac{2}{3})$ } 

The following achievability w.r.t. the elementary subchannels identified in Theorem~\ref{thm:ND} forms the basis of the proof for the regimes of (R1) and (R2). 

\begin{lemma} The following rates are achievable:
\label{lemma:buildingblocks1}
\begin{enumerate}
\item[(i)] For the pair of $(n,m)=(0,1)^i$ and $(\nt,\mt) = (1,2)^j$ where $ i \leq 2j$: $(R,\rt) = (i,2j-i)$;  
\item[(ii)] For the pair of $(n,m)=(2,1)^i$ and $(\nt,\mt) = (1,0)^j\times(2,1)^k$ where $i \leq 2j + 2k$: $(R,\rt) = (3i,2j+2k-i)$;
\item[(iii)] For the pair of $(n,m)=(2,1)^i$ and $(\nt,\mt) = (2,1)^j\times(3,2)^k$ where $i \leq 2j + 4k$: $(R,\rt) = (3i,2j+4k-i)$. 
\end{enumerate}
\end{lemma}
\begin{proof}
The proof builds upon the perfect feedback scheme in~\cite{SuhTse}. See Appendix~\ref{append:lemma_buildingblocks1} for the detailed proof.
\end{proof}

For the considered regimes, the claimed achievable rate region reads:
\begin{align*}
\{(R, \rt): R  \leq C_{\sf pf}, \rt \leq  \tilde{C}_{\sf pf}, R + \tilde{R} \leq C_{\sf no} + \tilde{C}_{\sf no}\}.
\end{align*} 
We see that there is no feedback gain in sum capacity. This means that one bit of a capacity increase due to feedback costs exactly one bit. Depending on whether or not $C_{\sf pf}$ (or $\tilde{C}_{\sf pf}$) exceeds $C_{\sf no} + \tilde{C}_{\sf no}$, we have four subcases, each of which forms a different shape of the region. See Fig.~\ref{fig:R1R2shape}.

\begin{figure}[t]
\begin{center}
{\epsfig{figure=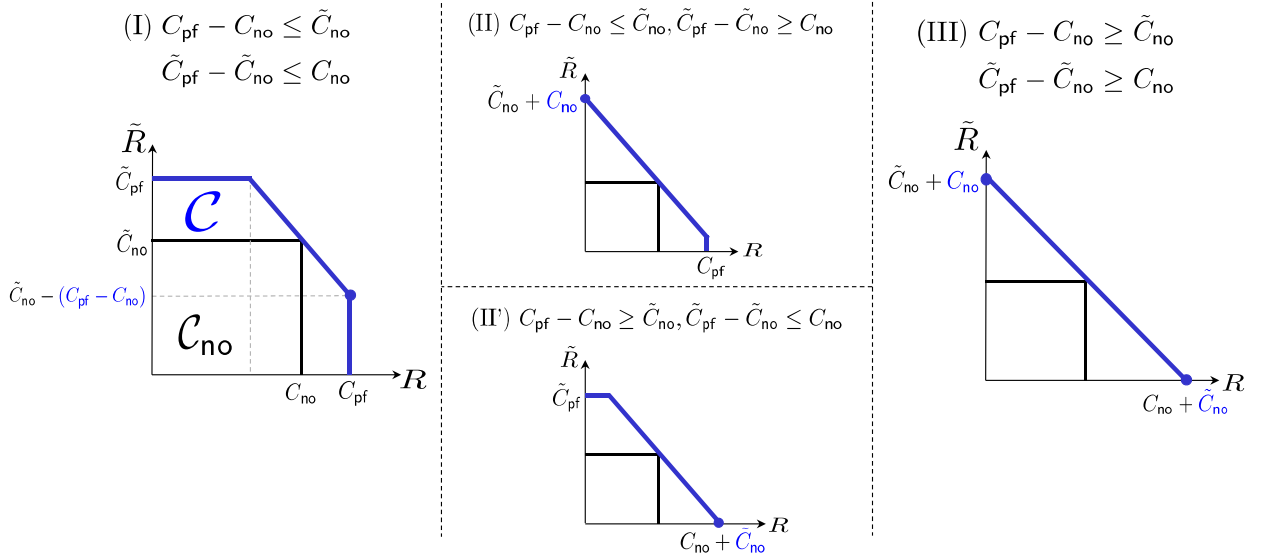, angle=0, width=1.0\textwidth}}
\end{center}
\caption{Four types of shapes of an achievable rate region for the regimes of (R1) $\alpha>2,\tilde{\alpha}>2 $ and (R2) $\alpha<\frac{2}{3},\tilde{\alpha}<\frac{2}{3}$. } \label{fig:R1R2shape}
\end{figure}

{\bf (I) $ (C_{\sf pf} - C_{\sf no} \leq \tilde{C}_{\sf no}), (\tilde{C}_{\sf pf} - \tilde{C}_{\sf no} \leq C_{\sf no})$}: The first case is the one in which the amount of feedback for maximal improvement, reflected in $C_{\sf pf} - C_{\sf no}$ (or $\tilde{C}_{\sf pf} - \tilde{C}_{\sf no}$), is smaller than the available resources offered by the backward IC (or the forward IC). In other words, in this case, we have a sufficient amount of resources such that one can achieve the perfect feedback bound in one direction. By symmetry, it suffices to focus on one corner point that favors the rate of forward transmission: $(R,\rt) = (C_{\sf pf}, \tilde{C}_{\sf no} - \blue{ (C_{\sf pf} - C_{\sf no}  )} )$.

%\begin{figure}[t]
%\begin{center}
%\epsfig{figure=./Figures/fig_R1achievability.eps, angle=0, width=1.0\textwidth}
%\end{center}
%\caption{Illustration of achievability for the (R1) regime via an example of $(n,m)=(1,3)$, $(\tilde{n},\tilde{m})=(1,3)$. By symmetry, we focus on one corner point: $(R,\rt) = (C_{\sf pf}, \tilde{C}_{\sf no} - (C_{\sf pf} - C_{\sf no} )  ) = ( m, 2 \nt - (m -2n) ) = (3, 1)$. This is an instance in which we have a sufficient amount of resources that enables achieving the perfect feedback bound in one direction: $C_{\sf pf}-C_{\sf no} = m -2n \leq 2 \tilde{n} = \tilde{C}_{\sf no}$ and $\tilde{C}_{\sf pf}-\tilde{C}_{\sf no} = \mt - 2 \nt \leq 2n = C_{\sf no}$.} 
%\label{fig:R1achievability}
%\end{figure}    

\emph{ (R1) $\alpha >2, \at >2$ }: For this regime, the network decomposition~(\ref{eq:ND3}) yields: 
\begin{align*}
(n,m) \longrightarrow&  (0,1)^{C_{\sf pf}-C_{\sf no}} \times (1,2)^{n},\\
(\nt,\mt) \longrightarrow& (0,1)^{\mt-2\nt} \times (1,2)^{\nt}.
\end{align*} 
Here we use the fact that $C_{\sf pf}-C_{\sf no}=m-2n$ in the considered case. We now apply Lemma~\ref{lemma:buildingblocks1}-(i) for the pair of $(0,1)^{C_{\sf pf}-C_{\sf no}}$ and $(1,2)^{\nt}$. Note that the condition in Lemma~\ref{lemma:buildingblocks1}-(i) holds: $C_{\sf pf} - C_{\sf no} \leq \tilde{C}_{\sf no} = 2 \nt$. This then gives: $R^{(1)} = C_{\sf pf} - C_{\sf no}$; $\tilde{R}^{(1)} = 2 \nt - (C_{\sf pf } - C_{\sf no})$. For the remaining subchannels, we apply the nonfeedback scheme, yielding: $R^{(2)} = 2n$; $\tilde{R}^{(2)} = 0$. Aggregating these two, we achieve the claimed corner point: 
\begin{align*}
R&= C_{\sf pf}-C_{\sf no} + 2  n=C_{\sf pf} - C_{\sf no} + C_{\sf no} = C_{\sf pf},\\
\tilde{R}&=2\nt-(C_{\sf pf}-C_{\sf no})=\tilde{C}_{\sf no}-(C_{\sf pf}-C_{\sf no}). 
\end{align*}

\emph{ (R2) $\alpha \in (0, \frac{2}{3}), \at \in (0, \frac{2}{3})$ }: Applying the network decompositions~(\ref{eq:ND1}) and~(\ref{eq:ND2}) to this regime, we get:
\begin{align*}
(n,m) \longrightarrow& \left\{
  \begin{array}{ll}
    (1,0)^{n-2m} \times (2,1)^{C_{\sf pf}-C_{\sf no}}, & \hbox{$\alpha\in\left(0,1/2 \right]$;} \\
     (2,1)^{C_{\sf pf}-C_{\sf no}} \times (3,2)^{2m-n}, & \hbox{$\alpha\in\left( 1/2, 2/3 \right)$;}
  \end{array}
\right. \\
(\nt,\mt) \longrightarrow& 
\left\{
  \begin{array}{ll}
    (1,0)^{\nt-2\mt} \times (2,1)^{\mt}, & \hbox{$\alpha\in\left(0,1/2 \right]$;} \\
     (2,1)^{2\nt-3\mt} \times (3,2)^{2\mt-\nt}, & \hbox{$\alpha\in\left( 1/2, 2/3 \right)$.}
  \end{array}
\right. 
\end{align*} 
Here we use the fact that $C_{\sf pf}-C_{\sf no} = m$ for $\alpha\in(0,\frac{1}{2}]$ and takes $2n-3m$ for $\alpha\in(\frac{1}{2},\frac{2}{3})$. When $\alpha \in (0, \frac{1}{2}]$ and $\at \in (0, \frac{1}{2}]$, we apply Lemma~\ref{lemma:buildingblocks1}-(ii) for the pair of $(2,1)^{C_{\sf pf}-C_{\sf no}}$ and $(1,0)^{\nt-2\mt} \times (2,1)^{\mt}$, yielding $R^{(1)} = 3 (C_{\sf pf} - C_{\sf no}) = 3 m$ and $\tilde{R}^{(1)} = 2 (\nt - 2 \mt) + 2 \mt - m$.  Notice that the condition in Lemma~\ref{lemma:buildingblocks1}-(ii) is satisfied:  $C_{\sf pf}-C_{\sf pf}\leq \ct_{\sf no}=2(\nt-2\mt)+2\mt$. For the rest, we apply the nonfeedback scheme to achieve $R^{(2)} = 2 (n-2m)$. This then gives:
\begin{align*}
R&= 3  m + 2 (n-2m) =2n-m=C_{\sf pf},\\
\tilde{R}&=2 (\nt-2\mt)+2 \mt-m=2(\nt-\mt)-m=\tilde{C}_{\sf no}-(C_{\sf pf}-C_{\sf no}). 
\end{align*} 
When $\alpha \in (0, \frac{1}{2}]$ and $\at \in (\frac{1}{2}, \frac{2}{3})$, we apply Lemma~\ref{lemma:buildingblocks1}-(iii) for the pair of $(2,1)^{C_{\sf pf}-C_{\sf no}}$ and $(2,1)^{2\nt-3\mt} \times (3,2)^{2\mt-\nt}$, yielding $R^{(1)} = 3 (C_{\sf pf} - C_{\sf no} ) = 3m$ and $\tilde{R}^{(1)} = 2 ( 2 \nt - 3 \mt) + 4 (2 \mt - \nt) -m$. Note that the associated condition holds: $C_{\sf pf}-C_{\sf pf}\leq \ct_{\sf no}=2(2\nt-3\mt)+4(\mt-\nt)$. For the remaining subchannels, we apply the nonfeedback scheme to achieve $R^{(2)} = 2 (n - 2m)$.  
This then gives: 
\begin{align*}
R&=  3 m  + 2 (n-2m) = 2n-m=C_{\sf pf},\\
\tilde{R}&=2 (2\nt-3\mt)+4(2\mt-\nt)-m=2\mt-m=\tilde{C}_{\sf no}-(C_{\sf pf}-C_{\sf no}). 
\end{align*}  
The proof for the other regimes of $[\alpha \in (\frac{1}{2}, \frac{2}{3}), \tilde{\alpha} \in (0, \frac{2}{3})]$ and $[\alpha \in (\frac{1}{2}, \frac{2}{3}), \tilde{\alpha} \in (\frac{1}{2}, \frac{2}{3})]$ follows similarly. 

As seen from all the cases above, one key observation to make is that the capacity increase due to feedback $C_{\sf pf} - C_{\sf no}$ plus the backward transmission rate is always $\tilde{C}_{\sf no}$, meaning that there is \emph{one-to-one tradeoff} between feedback and independent message transmissions, i.e., one bit of feedback costs one bit. 

{\bf (II) $C_{\sf pf}-C_{\sf no} \leq \tilde{C}_{\sf no}, \tilde{C}_{\sf pf}-\tilde{C}_{\sf no} \geq C_{\sf no}$}: Also in this case, one can readily prove the same one-to-one tradeoff relationship in achieving one corner point $(R,\tilde{R}) = (C_{\sf pf}, \tilde{C}_{\sf no} - (C_{\sf pf} - C_{\sf no}))$. Hence, we omit the detailed proof. On the other hand, there is a limitation in achieving the other counterpart. Note that the maximal feedback gain $\tilde{C}_{\sf pf} - \tilde{C}_{\sf no}$ for backward transmission does exceed the resource limit $C_{\sf no}$ offered by the forward channel. This leads the maximal achievable rate for backward transmission to be saturated by $\tilde{R} \leq \tilde{C}_{\sf no} + C_{\sf no}$. So the other corner point reads $(R,\rt) = (0, \tilde{C}_{\sf no} + C_{\sf no})$ instead. For completeness, we will show this is indeed the case as below. By symmetry, we omit the case of (II'). 

\emph{ (R1) $\alpha >2, \at >2$ }: For this regime,
\begin{align*}
(n,m) \longrightarrow&  (0,1)^{m-2n} \times (1,2)^{\frac{C_{\sf no}}{2}}\\
(\nt,\mt) \longrightarrow& (0,1)^{C_{\sf no}}
\times (0,1)^{(\tilde{C}_{\sf pf}-\tilde{C}_{\sf no})-C_{\sf no}} 
\times (1,2)^{\nt}.
\end{align*} 
Here we use the fact that $\tilde{C}_{\sf pf}-\tilde{C}_{\sf no}=\mt-2\nt$ and $\frac{C_{\sf no}}{2}=n$ in the considered case. We now apply a symmetric version of Lemma~\ref{lemma:buildingblocks1}-(i) for the pair of $(1,2)^{\frac{C_{\sf no}}{2}}$ and $(0,1)^{C_{\sf no}}$. This then gives: 
$R^{(1)} = 2  \frac{C_{\sf no}}{2} - C_{\sf no} = 0$; $\tilde{R}^{(1)}= C_{\sf no}$. For the rest, we apply the nonfeedback scheme to achieve: $R^{(2)} = 0$; $\tilde{R}^{(2)} = 2 \tilde{n} = 2 \tilde{C}_{\sf no}$. Hence, we achieve the claimed corner point: $(R,\tilde{R})=(0,\tilde{C}_{\sf no}+C_{\sf no})$.

\emph{ (R2) $\alpha \in (0, \frac{2}{3}), \at \in (0, \frac{2}{3})$ }: For this regime, the network decompositions~(\ref{eq:ND1}) and~(\ref{eq:ND2}) yield: 
\begin{align*}
(n,m) \longrightarrow&  
\left\{
  \begin{array}{ll}
    (1,0)^{n-2m} \times (2,1)^{m}, & \hbox{$\alpha\in\left(0,1/2 \right]$;} \\
     (2,1)^{2n-3m} \times (3,2)^{2m-n}, & \hbox{$\alpha\in\left(\frac{1}{2},\frac{2}{3}\right)$;}
  \end{array}
\right. \\
(\nt,\mt) \longrightarrow& 
\left\{
  \begin{array}{ll}
   (1,0)^{\nt-2\mt} \times (2,1)^{C_{\sf no}}\times(2,1)^{\ct_{\sf pf}-\ct_{\sf no}-C_{\sf no}}, & \hbox{$\alpha\in\left(0,1/2 \right]$;} \\
     (2,1)^{C_{\sf no}}\times(2,1)^{\ct_{\sf pf}-\ct_{\sf no}-C_{\sf no}} \times (3,2)^{2\mt-\nt}, & \hbox{$\alpha\in\left(\frac{1}{2},\frac{2}{3}\right)$.}
  \end{array}
\right.
\end{align*} 
Here we use the fact that $\ct_{\sf pf}-\ct_{\sf no}=\mt$ for $\at\in(0,\frac{1}{2}]$ and takes $2\nt-3\mt$ for $\at\in(\frac{1}{2},\frac{2}{3})$. When $\alpha \in (0, \frac{1}{2}]$ and $\at \in (0, \frac{1}{2}]$, we apply a symmetric version of Lemma~\ref{lemma:buildingblocks1}-(ii) for the pair of $(1,0)^{n-2m}\times(2,1)^m$ and $(2,1)^{C_{\sf no}}$. So we get: $R^{(1)} = 2 (n-2m) + 2 m - C_{\sf no} = 0$; $\tilde{R}^{(1)} = 3 C_{\sf no}$. For the rest, we apply the nonfeedback scheme to achieve: $R^{(2)}= 0; \tilde{R}^{(2)} = 2 ( \nt -2 \mt) + 2 ( \tilde{C}_{\sf pf} - \tilde{C}_{\sf no} - C_{\sf no}) = 2 \tilde{n}- 2 \mt - 2 C_{\sf no} =  \tilde{C}_{\sf no} - 2 C_{\sf no}$. Hence we achieve: $(R, \rt)= (0,\ct_{\sf no}+C_{\sf no})$.

When $\alpha \in (\frac{1}{2}, \frac{2}{3})$ and $\at \in (0, \frac{1}{2}]$, we apply a symmetric version of Lemma~\ref{lemma:buildingblocks1}-(iii) for the pair of $(2,1)^{2n-3m} \times (3,2)^{2m-n}$ and $(2,1)^{C_{\sf no}}$, thus giving: $R^{(1)} = 2 (2n-3m) + 4 (2m-n) - C_{\sf no} = 0$; $\tilde{R}^{(1)} = 3 C_{\sf no}$. For the rest, we apply the nonfeedback scheme to achieve: $R^{(2)} = 0$; $\tilde{R}^{(2)} = 2 (\tilde{C}_{\sf pf} - \tilde{C}_{\sf no} - C_{\sf no}) + 4 ( 2 \tilde{m} - \nt) = 2 \mt - 2 C_{\sf no}$. 
Hence we prove: $(R,\tilde{R}) = (0, \tilde{C}_{\sf no} + C_{\sf no})$. 
The proof of the other regimes $[\alpha \in (\frac{1}{2}, \frac{2}{3}), \tilde{\alpha} \in (0, \frac{2}{3})]$ and $[\alpha \in (\frac{1}{2}, \frac{2}{3}), \tilde{\alpha} \in (\frac{1}{2}, \frac{2}{3})]$ follows similarly.

{\bf (III) $C_{\sf pf}-C_{\sf no} \geq \tilde{C}_{\sf no}, \tilde{C}_{\sf pf}-\tilde{C}_{\sf no} \geq C_{\sf no}$}: This is the case in which there are limitations now in achieving both $R= C_{\sf pf}$ and $\rt = \tilde{C}_{\sf pf}$. 
Due to the same argument as above, what we can maximally achieve for $R$ (or $\rt$) in exchange of the other channel is $C_{\sf no} + \tilde{C}_{\sf no}$ which implies $(R,\rt) = (C_{\sf no} + \tilde{C}_{\sf no},0)$ or $(0, C_{\sf no} + \tilde{C}_{\sf no})$. The proof follows exactly the same as above; hence, we omit it.

\subsection{Proof of (R3) $\alpha >2, \at \in [\frac{2}{3}, 2]$} 

A tedious yet straightforward computation demonstrates that the claimed achievable rate region evaluated in the regime (R3) is:
\begin{align*}
\{(R, \rt): R  \leq C_{\sf pf}, \rt \leq  \tilde{C}_{\sf no}, R + \tilde{R} \leq C_{\sf no} + 2 \tilde{n} \}.
\end{align*} 
Unlike the (R1) and (R2) regimes, there is an interaction gain. Note that the sum-rate bound exceeds $C_{\sf no} + \tilde{C}_{\sf no}$ in the regime. The backward IC has no feedback gain. The network decomposition (3) together with the fact that $C_{\sf pf}-C_{\sf no}=m-2n$ in the regime gives:
\[
(n,m)\longrightarrow(0,1)^{C_{\sf pf}-C_{\sf no}}\times (1,2)^n.
\]
We find that the shape of the region depends on where $C_{\sf pf}- C_{\sf no}$ lies in between $2 \nt - \tilde{C}_{\sf no}$ and $2 \nt$. See Fig.~\ref{fig:R3shape}. 

\begin{figure}[t]
\begin{center}
{\epsfig{figure=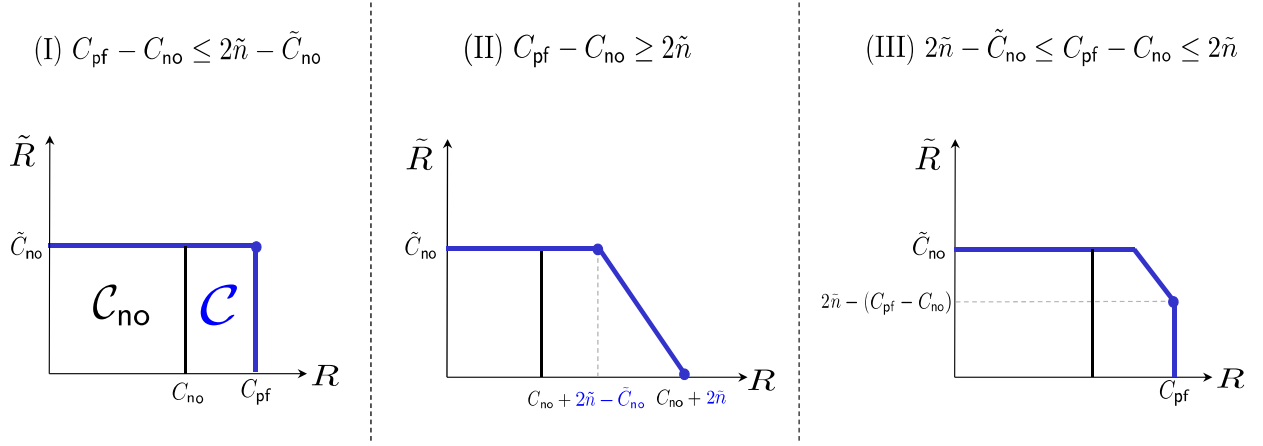, angle=0, width=1.0\textwidth}}
\end{center}
\caption{Three types of shapes of an achievable rate region for the regime (R3) $\alpha>2,\tilde{\alpha} \in [\frac{2}{3},2]$.} \label{fig:R3shape}
\end{figure}

{\bf (I) $ C_{\sf pf} - C_{\sf no} \leq 2 \nt - \tilde{C}_{\sf no}$}: The first case is the one in which the amount of feedback for maximal improvement, reflected in $C_{\sf pf} - C_{\sf no}$, is small enough to achieve the maximal feedback gain without degrading the performance of backward transmission. Now let us prove how to achieve $(R,\rt) = (C_{\sf pf}, \tilde{C}_{\sf no})$.

The decomposition idea is to pair up $(0,1)^{C_{\sf pf}-C_{\sf no}}$ and $(\nt, \mt)$ while applying the nonfeedback scheme for the remaining forward subchannel $(1,2)^n$. 
To give an achievability idea for the first pair, let us consider a simple example of $(n,m)=(0,1)$ and $(\nt,\mt)=(3,2)$. See Fig.~\ref{fig:R3achievability}.

\begin{figure}[t]
\begin{center}
\epsfig{figure=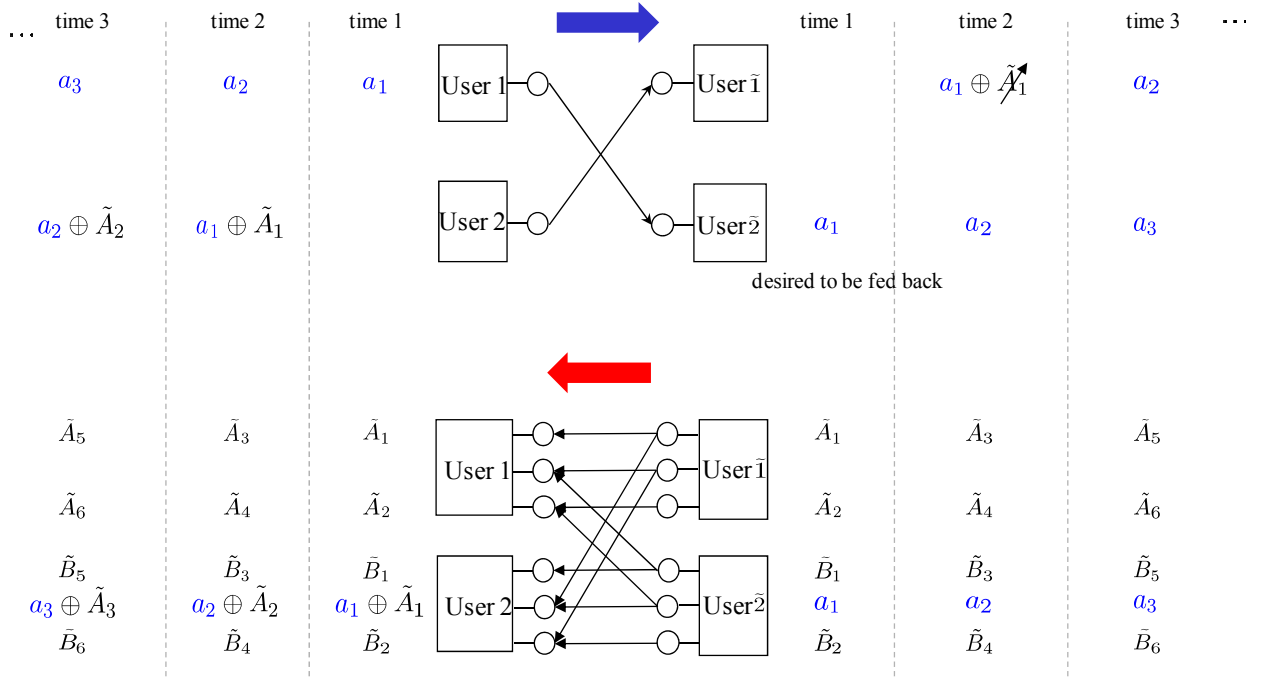, angle=0, width=1.0\textwidth}
\end{center}
\caption{Illustration of achievability for the (R3) regime via an example of $(n,m)=(0,1)$, $(\tilde{n},\tilde{m})=(3,2)$.  This is an instance in which we have a sufficient amount of resources that enables achieving the perfect feedback bound in the forward IC: $C_{\sf pf}-C_{\sf no} = 1 \leq 2 = 2 \nt - \tilde{C}_{\sf no}$. Hence, we achieve $(R,\rt) = (C_{\sf pf}, \tilde{C}_{\sf no}) =   (1, 4)$.} 
\label{fig:R3achievability}
\end{figure}

In each time, user 1 sends its own symbol $a_i$. Unlike the previous regimes (R1) and (R2), an interesting observation is made in feedback transmission. In the backward IC, $\tilde{C}_{\sf no} = \max ( 2 \nt - \mt, \mt )$ ($=4$ in this example) levels are utilized to send the backward symbols. For feedback, user $\tilde{2}$ sends user 1's received symbols $a_i$'s back to user 2 through the remaining \emph{direct-link} level. Here one can make two key observations. The first is that such feedback signal $a_i$ is interfered with by user $\tilde{1}$'s transmission but it turns out the interference does not cause any problem. Notice in the example that a feedback signal, say $a_1$, is mixed with $\tilde{A}_1$ and hence user 2 receives $a_1 \oplus \tilde{A}_1$ instead of $a_1$ which is desired to be fed back. Nonetheless user 2 sending the $a_1 \oplus \tilde{A}_1$ in time 2, user $\tilde{1}$ can decode $a_1$ of interest with the help of its own symbol $\tilde{A}_1$. This implies that feedback and independent backward message transmissions do not interfere with each other and thus one can maximally utilize available resource levels: the total number of direct-link levels $2\tilde{n}$. So the $2\tilde{n}-\tilde{C}_{\sf no}$ levels can be exploited for feedback. In the general case of $(0,1)^{C_{\sf pf} - C_{\sf no}}$, the maximal feedback gain $C_{\sf pf} - C_{\sf no}$ does not exceed the limit on the exploitable levels $2 \tilde{n} - \tilde{C}_{\sf no}$ under the considered regime. Hence, we achieve $R^{(1)} = C_{\sf pf} - C_{\sf no}$. Now the second observation is that the feedback transmission of $a_i$'s does not cause any interference to user 1. This ensures $\tilde{R}^{(1)} = \tilde{C}_{\sf no}$. On the other hand, for the remaining suchanneles $(1,2)^n$, we apply the nonfeedback scheme to achieve $R^{(2)} = 2n$. Combining all of the above, we get: 
\begin{align*}
R&=C_{\sf pf}-C_{\sf no}+2 n=C_{\sf pf}-C_{\sf no}+C_{\sf no}= C_{\sf pf}\\
\tilde{R}&=\tilde{C}_{\sf no}.
\end{align*}

{\bf (II) $ C_{\sf pf} - C_{\sf no} \geq 2 \nt$}: In this case, we do not have a sufficient amount of resources for achieving $R=C_{\sf pf}$. The maximally achievable forward rate is saturated by $C_{\sf no} + 2 \nt$ and this occurs when $\tilde{R}=0$.  On the other hand, under the constraint of $\tilde{R}=\tilde{C}_{\sf no}$, what one can achieve for $R$ is $C_{\sf no}+ ( 2 \nt - \tilde{C}_{\sf no})$. 

{\bf (III) $2 \nt - \tilde{C}_{\sf no}< C_{\sf pf} - C_{\sf no} < 2 \nt $}: This is the case in which we have a sufficient amount of resources for achieving $R=C_{\sf pf}$, but not enough to achieve $\tilde{R}=\tilde{C}_{\sf no}$ simultaneously. Hence, aiming at $R=C_{\sf pf}$, $\tilde{R}$ is saturated by $2\tilde{n} - (C_{\sf pf} - C_{\sf no})$.

\subsection{Proof of (R4) $\alpha \in (0,\frac{2}{3}], \at \in [\frac{2}{3}, 2]$} 

For the regime of (R4), the claimed achievable rate region is:
\begin{align*}
\{(R, \rt): R  \leq C_{\sf pf}, \rt \leq  \tilde{C}_{\sf no}, R + \tilde{R} \leq C_{\sf no} + 2 \tilde{m} \}.
\end{align*} 
This rate region is almost the same as that of (R3). The only difference is that the sum-rate bound now reads $C_{\sf no} + 2 \mt$ instead of $C_{\sf no} + 2 \nt$. Hence, the shape of the region depends now on where $C_{\sf pf}-C_{\sf no}$ lies in between $2\tilde{m}-\tilde{C}_{\sf no}$ and $2\tilde{m}$. See Fig.~\ref{fig:R4shape}. Here we will describe the proof for the case (I) $C_{\sf pf} - C_{\sf no} \leq 2 \mt - \tilde{C}_{\sf no}$ in which we have a sufficient amount of resources in achieving $(R, \rt) = (C_{\sf pf}, \tilde{C}_{\sf no})$. For the other cases of (II) and (III), one can make the same arguments as those in the (R3) regime; hence, we omit them.

Here what we need to demonstrate are two-folded. First, feedback and independent backward message transmissions do not interfere with each other. Second, the maximum number of resource levels utilized for sending feedback and independent backward symbols is limited by the total number of cross-link levels: $2 \mt$. The idea for feedback strategy is to employ Scheme 1 that we illustrated via Example 1 in Section~\ref{sec:scheme1}. We will show that the above two indeed hold when we use this idea. 

Note in Fig.~\ref{fig:example1} the tension between forward-symbol feedback and backward symbols, e.g., $a \oplus B$ vs. $\tilde{A}$. Scheme 1 based on XORing with interference neutralization leads us to completely resolve the tension. Observe that user 1 could decode $\tilde{A}$ of interest since user $\tilde{2}$ transmitted $B$ through the \emph{second cross-link} level to neutralize the inference $B$ at the bottom level at user 1. This contributes one bit (the number of the second cross-link level) to the backward symbol rate (w.r.t. $\tilde{A}$). At the same time, user 2 could obtain $a \oplus \tilde{A}$ (which would be used for the purpose of refinement in stage II) through the \emph{first cross-link} level. This contributes one bit (the number of the first cross-link level) to the feedback rate (w.r.t. $a \oplus \tilde{A}$). Similarly $b \oplus A$ and $\tilde{B}$ were successfully transmitted to user 1 and 2 respectively, and the contributed 2 bits correspond to the number of the \emph{remaining cross-link} levels. We can now see that feedback and independent backward symbols do not cause any interference to each other and the total transmission rate is limited by the total number of cross link levels: $2 \mt$.  Since the maximal amount of feedback $C_{\sf pf} - C_{\sf no}$ plus the backward symbol rate $\tilde{C}_{\sf no}$ does not exceed $2 \mt$ in the considered case, we can indeed achieve $(R,\rt) = (C_{\sf pf}, \tilde{C}_{\sf no})$.

\begin{figure}[t]
\begin{center}
{\epsfig{figure=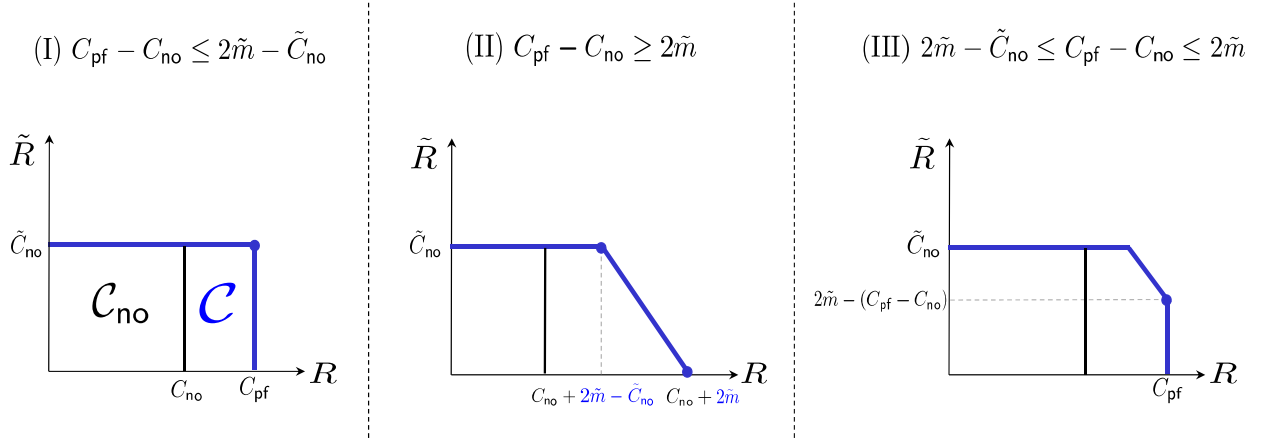, angle=0, width=1.0\textwidth}}
\end{center}
\caption{Three types of shapes of an achievable rate region for the regime (R4) $\alpha\in(0,\frac{2}{3}],\tilde{\alpha} \in [\frac{2}{3},2]$.} \label{fig:R4shape}
\end{figure}

\subsection{Proof of (R5) $\alpha \in (0,\frac{2}{3}), \at > 2$} 

For the regime of (R5), the claimed achievable rate region is:
\begin{align*}
\{(R, \rt): R  \leq C_{\sf pf}, \rt \leq  \tilde{C}_{\sf pf}, R + \tilde{R} \leq 2n + \tilde{C}_{\sf no}, R + \tilde{R} \leq C_{\sf no} + 2 \tilde{m} \}.
\end{align*} 
Remember that $C_{\sf pf} - C_{\sf no}$ indicates the maximum amount of feedback w.r.t. forward symbols and we interpret $2 \mt - \tilde{C}_{\sf pf}$ as the remaining resource levels that can potentially be utilized to aid forward transmission. Whether or not $C_{\sf pf} - C_{\sf no} \leq 2  \mt - \tilde{C}_{\sf pf}$ (i.e., we have enough resource levels to achieve $R=C_{\sf pf}$), the shape of the above claimed region is changed. Note that the last inequality in the rate region becomes inactive when $C_{\sf pf} - C_{\sf no} \leq 2  \mt - \tilde{C}_{\sf pf}$. Similarly the third inequality is inactive when $\tilde{C}_{\sf pf} - \tilde{C}_{\sf no} \leq 2 n - {C}_{\sf pf}$ (i.e., we have enough resources for achieving $\tilde{R} = \tilde{C}_{\sf pf}$). One can readily verify that $C_{\sf pf} - C_{\sf no} > 2  \mt - \tilde{C}_{\sf pf}$ and $\tilde{C}_{\sf pf} - \tilde{C}_{\sf no}  >  2 n - {C}_{\sf pf}$ do not hold simultaneously. Hence, it suffices to consider the following three cases: 
\begin{align*}
 \textrm{(I) } &C_{\sf pf} - C_{\sf no} \leq 2  \mt - \tilde{C}_{\sf pf}, \quad \tilde{C}_{\sf pf} - \tilde{C}_{\sf no} \leq 2 n - {C}_{\sf pf}; \\
 \textrm{(II) } &C_{\sf pf} - C_{\sf no} \leq 2  \mt - \tilde{C}_{\sf pf}, \quad \tilde{C}_{\sf pf} - \tilde{C}_{\sf no} > 2 n - {C}_{\sf pf}; \\
\textrm{(III) } &C_{\sf pf} - C_{\sf no} > 2  \mt - \tilde{C}_{\sf pf}, \quad \tilde{C}_{\sf pf} - \tilde{C}_{\sf no} \leq 2 n - {C}_{\sf pf}.
\end{align*}

As mentioned earlier in Example 3, the key idea for the proof is to use the network decomposition. Specifically, the following lemma that describes achievability for the elementary subchannels in the considered regime forms the basis of the proof.

\begin{lemma} The following rates are achievable:
\label{lemma:buildingblocks}
\begin{enumerate}
\item[(i)] For the pair of $(n,m)=(2,1)$ and $(\nt,\mt) = (0,1)$: $(R,\rt) = (3,1)$;
\item[(ii)] For the pair of $(n,m)=(2,1)^i$ and $(\nt,\mt) = (1,2)^j$ where $i \leq 2j$: $(R,\rt) = (3i,2j)$; 
\item[(iii)] For the pair of $(n,m)=(3,2)^i$ and $(\nt,\mt) = (0,1)^j$ where $2i \geq j$: $(R,\rt) = (4i,j)$; 
\item[(iv)] For the pair of $(n,m)=(2,1)$ and $(\nt,\mt) = (0,1)^2$: $(R,\rt) = (2,2)$;  
\item[(v)] For the pair of $(n,m)=(2,1)^2$ and $(\nt,\mt) = (0,1)$: $(R,\rt) = (6,0)$.
\end{enumerate}
\end{lemma}
\begin{proof}
See Appendix~\ref{append:lemma_buildingblocks}.
\end{proof}

{\bf (I) $C_{\sf pf}-C_{\sf no}\leq 2\tilde{m}-\tilde{C}_{\sf pf}, \tilde{C}_{\sf pf}-\tilde{C}_{\sf no}\leq 2n-{C}_{\sf pf}$}: In this case, the rate region claims that we can get all the way to perfect feedback capacities: $(R,\rt) = (C_{\sf pf}, \tilde{C}_{\sf pf})$. First consider the regime of $\alpha\in(0,\frac{1}{2}]$ in which the network decompositions~\eqref{eq:ND1} and~\eqref{eq:ND3} yield: 
\begin{align*}
(n,m)\longrightarrow& {(2,1)^{\tilde{C}_{\sf pf}-\tilde{C}_{\sf no}}}
\times{(2,1)^{C_{\sf pf}-C_{\sf no}-(\tilde{C}_{\sf pf}-\tilde{C}_{\sf no})}}
\times(1,0)^{n-2m}; \\
(\tilde{n},\tilde{m})\longrightarrow&{(0,1)^{\tilde{C}_{\sf pf}-\tilde{C}_{\sf no}}} \times{(1,2)^{\tilde{n}}}.
\end{align*}
Here we use the fact that $C_{\sf pf} - C_{\sf no} = m$ and that $\tilde{C}_{\sf pf}-\tilde{C}_{\sf no}\leq C_{\sf pf}-{C}_{\sf no} = 2n - C_{\sf pf}$ in the considered regime.
%\begin{figure}[t]
%\begin{center}
%{\epsfig{figure=./Figures/fig_R5shape1.eps, angle=0, width=0.3\textwidth}}
%\end{center}
%\caption{The shape of an achievable rate region for the regime (R5) $\alpha\in(0,\frac{2}{3}],\tilde{\alpha}>2 $ and the case (I) $C_{\sf pf}-C_{\sf no}\leq 2\mt-\tilde{C}_{\sf pf}, \tilde{C}_{\sf pf}-\tilde{C}_{\sf no}\leq 2n-{C}_{\sf pf}$.} \label{fig:R5shape1}
%\end{figure}
We now apply Lemma~\ref{lemma:buildingblocks}-(i) for the pair of ${(2,1)^{\tilde{C}_{\sf pf}-\tilde{C}_{\sf no}}}$ and ${(0,1)^{\tilde{C}_{\sf pf}-\tilde{C}_{\sf no}}}$. Also we apply Lemma~\ref{lemma:buildingblocks}-(ii) for the pair of ${(2,1)^{C_{\sf pf}-C_{\sf no}-(\tilde{C}_{\sf pf}-\tilde{C}_{\sf no})}}$ and ${(1,2)^{\tilde{n}}}$. Note that $C_{\sf pf}-C_{\sf no}-(\tilde{C}_{\sf pf}-\tilde{C}_{\sf no})\leq 2\tilde{n}$ in the considered regime: $C_{\sf pf}-C_{\sf no}\leq 2\tilde{m}-\tilde{C}_{\sf pf}$.
Lastly we apply the nonfeedback scheme for the remaining subchannel $(1,0)^{n-2m}$. This yields: 
\begin{align*}
R&={3\times (\tilde{C}_{\sf pf}-\tilde{C}_{\sf no})}
+{3\times \{C_{\sf pf}-C_{\sf no}-(\tilde{C}_{\sf pf}-\tilde{C}_{\sf no})\}}+2\times (n-2m) =2n-m=C_{\sf pf},\\
\tilde{R}&={1\times(\tilde{C}_{\sf pf}-\tilde{C}_{\sf no})}
+{2\times\tilde{n}} =\tilde{m}=\tilde{C}_{\sf pf}.
\end{align*} 

Next consider the regime of $\alpha\in[\frac{1}{2},\frac{2}{3}]$. In this regime, there are two subcases depending on whether or not $C_{\sf pf}-C_{\sf no}\geq \tilde{C}_{\sf pf}-\tilde{C}_{\sf no}$. When $C_{\sf pf}-C_{\sf no}\geq \tilde{C}_{\sf pf}-\tilde{C}_{\sf no}$,  
\begin{align*}
(n,m)\longrightarrow& {(2,1)^{\tilde{C}_{\sf pf}-\tilde{C}_{\sf no}}}
\times{(2,1)^{C_{\sf pf}-C_{\sf no}-(\tilde{C}_{\sf pf}-\tilde{C}_{\sf no})}}
\times(3,2)^{2m-n};\\
(\tilde{n},\tilde{m})\longrightarrow&{(0,1)^{\tilde{C}_{\sf pf}-\tilde{C}_{\sf no}}}
\times{(1,2)^{\tilde{n}}}.
\end{align*}
We apply Lemma~\ref{lemma:buildingblocks}-(i) for the pair of ${(2,1)^{\tilde{C}_{\sf pf}-\tilde{C}_{\sf no}}}$ and ${(0,1)^{\tilde{C}_{\sf pf}-\tilde{C}_{\sf no}}}$; apply Lemma~\ref{lemma:buildingblocks}-(ii) for the pair of ${(2,1)^{C_{\sf pf}-C_{\sf no}-(\tilde{C}_{\sf pf}-\tilde{C}_{\sf no})}}$ and ${(1,2)^{\tilde{n}}}$ (note that $C_{\sf pf}-C_{\sf no}-(\tilde{C}_{\sf pf}-\tilde{C}_{\sf no})\leq 2\tilde{n}$ in the considered regime $C_{\sf pf}-C_{\sf no}\leq 2\tilde{m}-\tilde{C}_{\sf pf}$); apply the nonfeedback scheme for $(3,2)^{2m-n}$. This gives:
\begin{align*}
R&={3\times (\tilde{C}_{\sf pf}-\tilde{C}_{\sf no})}
+{3\times \{C_{\sf pf}-C_{\sf no}-(\tilde{C}_{\sf pf}-\tilde{C}_{\sf no})\}}+4\times (2m-n) =2n-m=C_{\sf pf},\\
\tilde{R}&={1\times(\tilde{C}_{\sf pf}-\tilde{C}_{\sf no})}
+{2\times\tilde{n}} =\tilde{m}=\tilde{C}_{\sf pf}.
\end{align*}
For the other case $C_{\sf pf}-C_{\sf no}< \tilde{C}_{\sf pf}-\tilde{C}_{\sf no}$,
\begin{align*}
(n,m)&\longrightarrow{(2,1)^{C_{\sf pf}-C_{\sf no}}}
\times{(3,2)^{2m-n}},\\
(\tilde{n},\tilde{m})&\longrightarrow{(0,1)^{{C}_{\sf pf}-{C}_{\sf no}}}
\times{(0,1)^{\tilde{C}_{\sf pf}-\tilde{C}_{\sf no}-({C}_{\sf pf}-{C}_{\sf no})}}
\times(1,2)^{\tilde{n}}.
\end{align*}
Using Lemma~\ref{lemma:buildingblocks} and making similar arguments as earlier, one can  show that
%
%
%For the pair of a forward-subchannel ${(2,1)^{{C}_{\sf pf}-{C}_{\sf no}}}$ and a backward-subchannel ${(0,1)^{{C}_{\sf pf}-{C}_{\sf no}}}$, we apply lemma~\ref{lemma:lemscheme2}; for the pair of a forward-subchannel ${(3,2)^{2m-n}}$ and a backward-subchannel ${(0,1)^{\tilde{C}_{\sf pf}-\tilde{C}_{\sf no}-({C}_{\sf pf}-{C}_{\sf no})}}$, we apply lemma~\ref{lemma:lemR3} ($\tilde{C}_{\sf pf}-\tilde{C}_{\sf no}\leq 2n -C_{\sf pf}$ implies $\tilde{C}_{\sf pf}-\tilde{C}_{\sf no}-(C_{\sf pf}-C_{\sf no})\leq 2(2m-n)$); for the remaining backward-subchannel $(1,2)^{\nt}$, we apply the non-feedback scheme. Then we achieve: 
\begin{align*}
R&={3\times ({C}_{\sf pf}-{C}_{\sf no})}
+{4\times (2m-n)} =2n-m=C_{\sf pf},\\
\tilde{R}&={1\times({C}_{\sf pf}-{C}_{\sf no})}
+{1\times\{\tilde{C}_{\sf pf}-\tilde{C}_{\sf no}-({C}_{\sf pf}-{C}_{\sf no})\}}
+2\times\tilde{n} =\tilde{m}=\tilde{C}_{\sf pf}.
\end{align*}

{\bf (II) $C_{\sf pf}-C_{\sf no}\leq 2\tilde{m}-\tilde{C}_{\sf pf}, \tilde{C}_{\sf pf}-\tilde{C}_{\sf no}> 2n-{C}_{\sf pf}$}: In this case, there are two corner points to achieve. The first corner point is $(R,\rt) = (C_{\sf pf}, \tilde{C}_{\sf no} + 2n - C_{\sf pf})$. The second corner point depends on where $\tilde{C}_{\sf pf} - \tilde{C}_{\sf no}$ lies in between $2n- C_{\sf no}$, $2n$ and beyond. See Fig.~\ref{fig:R5shape2}. For the cases of (II-1) and (II-2), the corner point reads $(R,\rt) = (2n - (\tilde{C}_{\sf pf} - \tilde{C}_{\sf no} ), \tilde{C}_{\sf pf})$, while for (II-3), $(R,\rt) = (0, \tilde{C}_{\sf no} + 2n)$.

\begin{figure}[t]
\begin{center}
{\epsfig{figure=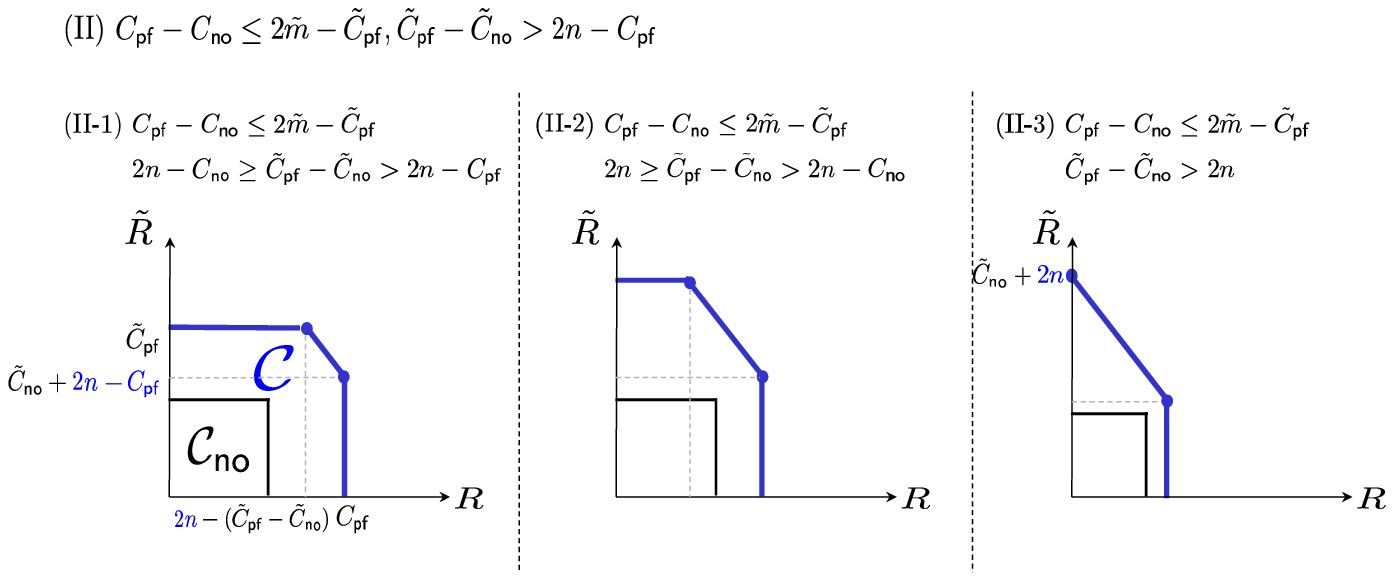, angle=0, width=1.0\textwidth}}
\end{center}
\caption{Three types of shapes of an achievable rate region for the regime (R5) $\alpha\in(0,\frac{2}{3}],\tilde{\alpha} >2$ and the case (II) $C_{\sf pf}-C_{\sf no}\leq 2\tilde{m}-\tilde{C}_{\sf pf}, \tilde{C}_{\sf pf}-\tilde{C}_{\sf no}> 2n-{C}_{\sf pf}$.} \label{fig:R5shape2}
\end{figure}

Let us first prove $(R,\rt) = (C_{\sf pf}, \tilde{C}_{\sf no} + 2n - C_{\sf pf})$. 
For the regime $\alpha\in (0,\frac{1}{2}]$, 
\begin{align*}
(n,m)\longrightarrow&{(2,1)^{C_{\sf pf}-C_{\sf no}}}
\times(1,0)^{n-2m},\\
(\tilde{n},\tilde{m})\longrightarrow&{(0,1)^{{C}_{\sf pf}-{C}_{\sf no}}}
\times(0,1)^{\tilde{C}_{\sf pf}-\tilde{C}_{\sf no}-({C}_{\sf pf}-{C}_{\sf no})}\times(1,2)^{\tilde{n}}.
\end{align*}
Note that $\tilde{C}_{\sf pf}-\tilde{C}_{\sf no}>C_{\sf pf}-C_{\sf no}$ in the considered regime $\tilde{C}_{\sf pf}-\tilde{C}_{\sf no}> 2n -C_{\sf pf}$. 
We apply Lemma~\ref{lemma:buildingblocks}-(i) for the pair of ${(2,1)^{{C}_{\sf pf}-{C}_{\sf no}}}$ and ${(0,1)^{{C}_{\sf pf}-{C}_{\sf no}}}$; apply the nonfeedback scheme for the rest. This yields:
\begin{align*}
R&={3\times(C_{\sf pf}-C_{\sf no})}+2\times (n-2m) =2n-m=C_{\sf pf},\\
\tilde{R}&={1\times({C}_{\sf pf}-{C}_{\sf no})}+2\times\tilde{n} =\tilde{C}_{\sf no}+2n-C_{\sf pf}.
\end{align*}
For the regime $\alpha\in [\frac{1}{2},\frac{2}{3}]$, 
\begin{align*}
(n,m)\longrightarrow&{(2,1)^{C_{\sf pf}-C_{\sf no}}}
\times{(3,2)^{2m-n}},\\
(\tilde{n},\tilde{m})\longrightarrow&{(0,1)^{{C}_{\sf pf}-{C}_{\sf no}}}
\times{(0,1)^{2(2m-n)}}
\times(0,1)^{\tilde{C}_{\sf pf}-\tilde{C}_{\sf no}-({C}_{\sf pf}-{C}_{\sf no})-2(2m-n)}
\times(1,2)^{\tilde{n}}.
\end{align*}
Note that $\tilde{C}_{\sf pf}-\tilde{C}_{\sf no}>C_{\sf pf}-C_{\sf no}+2(2m-n)$ in the considered regime $\tilde{C}_{\sf pf}-\tilde{C}_{\sf no}> 2n -C_{\sf pf}$. 
We apply Lemma~\ref{lemma:buildingblocks}-(i) for the pair of ${(2,1)^{{C}_{\sf pf}-{C}_{\sf no}}}$ and ${(0,1)^{{C}_{\sf pf}-{C}_{\sf no}}}$; apply Lemma~\ref{lemma:buildingblocks}-(iii) for the pair of ${(3,2)^{2m-n}}$ and ${(0,1)^{2(2m-n)}}$; apply the nonfeedback scheme for the rest. This yields: 
\begin{align*}
R&={3\times(C_{\sf pf}-C_{\sf no})}
+{4(2m-n)} =2n-m=C_{\sf pf},\\
\tilde{R}&={1\times(C_{\sf pf}-C_{\sf no})}
+{1\times 2(2m-n)}
+2\tilde{n} =\tilde{C}_{\sf no}+2n-C_{\sf pf}.
\end{align*}

We are now ready to prove the second corner point which favors $\rt$. 
Depending on the quantity of $\tilde{C}_{\sf pf} - \tilde{C}_{\sf no}$, we have three subcases.

{\bf (II-1) $2n-C_{\sf pf} < \tilde{C}_{\sf pf}-\tilde{C}_{\sf no} \leq  2n -C_{\sf no}$}:
For the regime $\alpha\in(0,\frac{1}{2}]$, 
\begin{align*}
(n,m)\longrightarrow&{(2,1)^{2n-C_{\sf no}-(\tilde{C}_{\sf pf}-\tilde{C}_{\sf no})}}
\times{(2,1)^{\tilde{C}_{\sf pf}-\tilde{C}_{\sf no}-(2n-C_{\sf pf})}}
\times(1,0)^{n-2m},\\
(\tilde{n},\tilde{m})\longrightarrow&{(0,1)^{2n-C_{\sf no}-(\tilde{C}_{\sf pf}-\tilde{C}_{\sf no})}}
\times{(0,1)^{2\{\tilde{C}_{\sf pf}-\tilde{C}_{\sf no}-(2n-C_{\sf pf})\}}}\times(1,2)^{\tilde{n}}.
\end{align*}
We apply Lemma~\ref{lemma:buildingblocks}-(i) for the pair of ${(2,1)^{2n-C_{\sf no}-(\tilde{C}_{\sf pf}-\tilde{C}_{\sf no})}}$ and ${(0,1)^{2n-C_{\sf no}-(\tilde{C}_{\sf pf}-\tilde{C}_{\sf no})}}$; apply Lemma~\ref{lemma:buildingblocks}-(iv) for the pair of ${(2,1)^{\tilde{C}_{\sf pf}-\tilde{C}_{\sf no}-(2n-C_{\sf pf})}}$ and ${(0,1)^{2\{\tilde{C}_{\sf pf}-\tilde{C}_{\sf no}-(2n-C_{\sf pf})\}}}$; apply the non-feedback scheme for the rest. This then gives: 
\begin{align*}
R&={3 \{2n-C_{\sf no}-(\tilde{C}_{\sf pf}-\tilde{C}_{\sf no})\}}
+{2 \{\tilde{C}_{\sf pf}-\tilde{C}_{\sf no}-(2n-C_{\sf pf})\}}
+2 (n-2m) =2n-(\tilde{C}_{\sf pf}-\tilde{C}_{\sf no}),\\
\tilde{R}&={\{2n-C_{\sf no}-(\tilde{C}_{\sf pf}-\tilde{C}_{\sf no})\}}
+{2\{\tilde{C}_{\sf pf}-\tilde{C}_{\sf no}-(2n-C_{\sf pf})\}}
+2 \tilde{n} =\tilde{C}_{\sf pf}.
\end{align*}
For the regime $\alpha\in[\frac{1}{2},\frac{2}{3}]$, 
\begin{align*}
(n,m)\longrightarrow&{(2,1)^{2n-C_{\sf no}-(\tilde{C}_{\sf pf}-\tilde{C}_{\sf no})}}
\times {(2,1)^{\tilde{C}_{\sf pf}-\tilde{C}_{\sf no}-(2n-C_{\sf pf})}}
\times{(3,2)^{2m-n}},\\
(\tilde{n},\tilde{m})\longrightarrow&{(0,1)^{2n-C_{\sf no}-(\tilde{C}_{\sf pf}-\tilde{C}_{\sf no})}}
\times {(0,1)^{2\{\tilde{C}_{\sf pf}-\tilde{C}_{\sf no}-(2n-C_{\sf pf})\}}}
\times {(0,1)^{2(2m-n)}}
\times(1,2)^{\tilde{n}}.
\end{align*}
We apply Lemma~\ref{lemma:buildingblocks}-(i) for the pair of ${(2,1)^{2n-C_{\sf no}-(\tilde{C}_{\sf pf}-\tilde{C}_{\sf no})}}$ and ${(0,1)^{2n-C_{\sf no}-(\tilde{C}_{\sf pf}-\tilde{C}_{\sf no})}}$; apply Lemma~\ref{lemma:buildingblocks}-(iv) for the pair of ${(2,1)^{\tilde{C}_{\sf pf}-\tilde{C}_{\sf no}-(2n-C_{\sf pf})}}$ and ${(0,1)^{2\{\tilde{C}_{\sf pf}-\tilde{C}_{\sf no}-(2n-C_{\sf pf})\}}}$; apply Lemma~\ref{lemma:buildingblocks}-(iii) for the pair of ${(3,2)^{2m-n}}$ and ${(0,1)^{2(2m-n)}}$; apply the nonfeedback scheme for the rest. This then gives: 
\begin{align*}
R&={3 \{2n-C_{\sf no}-(\tilde{C}_{\sf pf}-\tilde{C}_{\sf no})\}}
+{2  \{\tilde{C}_{\sf pf}-\tilde{C}_{\sf no}-(2n-C_{\sf pf})\}}
+{4(2m-n)} =2n-(\tilde{C}_{\sf pf}-\tilde{C}_{\sf no}),\\
\tilde{R}&={ \{2n-C_{\sf no}-(\tilde{C}_{\sf pf}-\tilde{C}_{\sf no})\}}
+{2 \{\tilde{C}_{\sf pf}-\tilde{C}_{\sf no}-(2n-C_{\sf pf})\}}
+{1 2(2m-n)}
+2\tilde{n} =\tilde{C}_{\sf pf}.
\end{align*}

{\bf (II-2) $ 2n - C_{\sf no} < \tilde{C}_{\sf pf}-\tilde{C}_{\sf no} \leq  2n$}: It turns out in this case proving achievability only via the network decomposition is a bit involved. So for illustrative purpose, we will first show achievability for one point that lies on the 45-degree line connecting the two corner points. Later we will slightly perturb the scheme to prove achievability for the second corner point that we intend to achieve. 

First consider the regime $\alpha\in (0,\frac{1}{2}]$. In this case, 
\begin{align*}
(n,m)\longrightarrow&{(2,1)^{C_{\sf pf}-C_{\sf no}}}
\times(1,0)^{n-2m},\\
(\tilde{n},\tilde{m})\longrightarrow&{(0,1)^{2(C_{\sf pf}-C_{\sf no})}}
\times(0,1)^{\tilde{C}_{\sf pf}-\tilde{C}_{\sf no}-2(C_{\sf pf}-C_{\sf no})}
\times(1,2)^{\tilde{n}}.
\end{align*}
Note that $\tilde{C}_{\sf pf}-\tilde{C}_{\sf no}> 2(C_{\sf pf}-C_{\sf no})$ in the considered regime $\tilde{C}_{\sf pf}-\tilde{C}_{\sf no}> 2n -C_{\sf no}$. We apply Lemma~\ref{lemma:buildingblocks}-(iv) for the pair of ${(2,1)^{C_{\sf pf}-C_{\sf no}}}$ and ${(0,1)^{2(C_{\sf pf}-C_{\sf no})}}$; apply the nonfeedback scheme for the rest. This then yields:
\begin{align*}
R&={2\times(C_{\sf pf}-C_{\sf no})}
+2\times (n-2m),\\
\tilde{R}&={2\times(C_{\sf pf}-C_{\sf no})}
+2\times\tilde{n}.
\end{align*}
As mentioned earlier, this is an intermediate point that lies on the 45-degree line connecting the two corner points. Now we tune the scheme which yields the above rate to prove the achievability of the second corner point. We use part of the forward channel for aiding backward transmission instead of sending its own traffic. Specifically we utilize $\tilde{C}_{\sf pf} - \tilde{C}_{\sf no} - 2 (C_{\sf pf} - C_{\sf no})$ number of bottom levels in the forward channel in an effort to relay backward-symbol feedback. This naive change incurs one-to-one tradeoff between feedback and independent message transmission, thus yielding: 
\begin{align*}
R&={2\times(C_{\sf pf}-C_{\sf no})}
+2\times (n-2m)
-\{\tilde{C}_{\sf pf}-\tilde{C}_{\sf no}-2(C_{\sf pf}-C_{\sf no})\}=2n-(\tilde{C}_{\sf pf}-\tilde{C}_{\sf no}),\\
\tilde{R}&={2\times(C_{\sf pf}-C_{\sf no})}
+2\times\tilde{n}
+\{\tilde{C}_{\sf pf}-\tilde{C}_{\sf no}-2(C_{\sf pf}-C_{\sf no})\}=\tilde{C}_{\sf pf}.
\end{align*} 

For the regime $\alpha\in[\frac{1}{2}, \frac{2}{3}]$,
\begin{align*}
(n,m)\longrightarrow&{(2,1)^{C_{\sf pf}-C_{\sf no}}}
\times{(3,2)^{2m-n}},\\
(\tilde{n},\tilde{m})\longrightarrow&{(0,1)^{2(C_{\sf pf}-C_{\sf no})}}
\times {(0,1)^{2(2m-n)}}
\times {(0,1)^{\tilde{C}_{\sf pf}-\tilde{C}_{\sf no}-2(C_{\sf pf}-C_{\sf no})-2(2m-n)}}
\times(1,2)^{\tilde{n}}.
\end{align*}
We apply Lemma~\ref{lemma:buildingblocks}-(iv) for the pair of ${(2,1)^{C_{\sf pf}-C_{\sf no}}}$ and ${(0,1)^{2(C_{\sf pf}-C_{\sf no})}}$; apply Lemma~\ref{lemma:buildingblocks}-(iii) for the pair of ${(3,2)^{2m-n}}$ and ${(0,1)^{2(2m-n)}}$; apply the nonfeedback scheme for the rest. This then gives:
\begin{align*}
R&={2\times(C_{\sf pf}-C_{\sf no})}
+{4\times(2m-n)},\\
\tilde{R}&={2\times(C_{\sf pf}-C_{\sf no})}
+{1\times 2(2m-n)}
+2\tilde{n}.
\end{align*}
This is an intermediate point that lies on the 45-degree line connecting the two corner points. Now sacrificing $\tilde{C}_{\sf pf} - \tilde{C}_{\sf no} - 2 (C_{\sf pf} - C_{\sf no}) - 2 (2m-n)$ number of resource levels in the forward channel for aiding backward transmission, we achieve:
\begin{align*}
R&={2(C_{\sf pf}-C_{\sf no})}
+{4 (2m-n)}
-\{\tilde{C}_{\sf pf}-\tilde{C}_{\sf no}-2(C_{\sf pf}-C_{\sf no})-2(2m-n)\} =2n-(\tilde{C}_{\sf pf}-\tilde{C}_{\sf no}),\\
\tilde{R}&={2 (C_{\sf pf}-C_{\sf no})}
+{ 2(2m-n)}
+2\tilde{n}
+\{\tilde{C}_{\sf pf}-\tilde{C}_{\sf no}-2(C_{\sf pf}-C_{\sf no})-2(2m-n)\} =\tilde{C}_{\sf pf}.
\end{align*}

{\bf (II-3) $\tilde{C}_{\sf pf}-\tilde{C}_{\sf no}> 2n$}:
In this case, we sacrifice all of the $2n$ direct links in the forward channel only for the purpose of helping backward transmission. This then gives: $(R,\tilde{R})=(0,\tilde{C}_{\sf no}+2n)$.

\begin{figure}[t]
\begin{center}
{\epsfig{figure=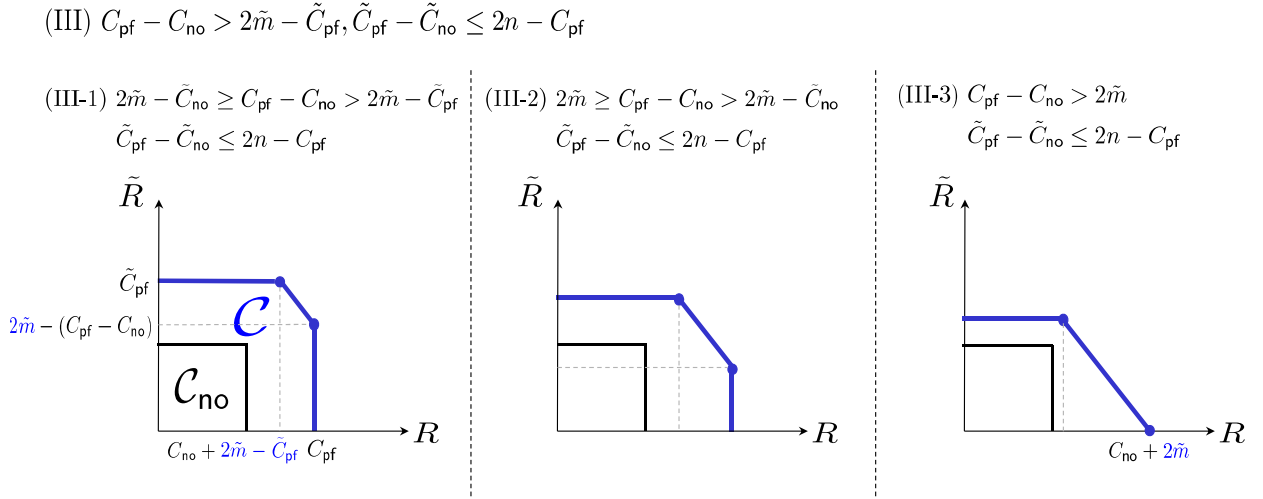, angle=0, width=1.0\textwidth}}
\end{center}
\caption{Three types of shapes of an achievable rate region for the regime (R5) $\alpha\in(0,\frac{2}{3}],\tilde{\alpha} >2$ and the case (III) $C_{\sf pf}-C_{\sf no}> 2\tilde{m}-\tilde{C}_{\sf pf}, \tilde{C}_{\sf pf}-\tilde{C}_{\sf no}\leq 2n-{C}_{\sf pf}$.} \label{fig:R5shape3}
\end{figure}

{\bf (III) $C_{\sf pf}-C_{\sf no}> 2\tilde{m}-\tilde{C}_{\sf pf}, \tilde{C}_{\sf pf}-\tilde{C}_{\sf no}\leq 2n-{C}_{\sf pf}$}:
Similarly this case requires the proof of two corner points. The first corner point is: $(R,\rt) = (C_{\sf no} + 2 \mt - \tilde{C}_{\sf pf}, \tilde{C}_{\sf pf})$. The second corner point depends on where $C_{\sf pf}- C_{\sf no}$ lies in. See Fig.~\ref{fig:R5shape3}. While the proof is similar to that in the previous case, we provide details for completeness. 

First focus on the proof of the first corner point $(R,\rt) = (C_{\sf no} + 2 \mt - \tilde{C}_{\sf pf}, \tilde{C}_{\sf pf})$. Notice that for the regime $\alpha\in(0,\frac{1}{2}],\at>2$, we encounter a contradiction as follows:
\begin{align*}
C_{\sf pf}-C_{\sf no}> 2\tilde{m}-\tilde{C}_{\sf pf}&\Rightarrow C_{\sf pf}-C_{\sf no}> \tilde{C}_{\sf pf}-\tilde{C}_{\sf no};\\
\tilde{C}_{\sf pf}-\tilde{C}_{\sf no}\leq 2n -C_{\sf pf}&\Rightarrow C_{\sf pf}-C_{\sf no}\leq \tilde{C}_{\sf pf}-\tilde{C}_{\sf no}.
\end{align*}
Hence, we will not consider this regime. For the regime $\alpha\in [\frac{1}{2}, \frac{2}{3}]$, 
\begin{align*}
(n,m)\longrightarrow&{(2,1)^{\tilde{C}_{\sf pf}-\tilde{C}_{\sf no}}}
\times {(2,1)^{2\tilde{n}}}
\times (2,1)^{C_{\sf pf}-C_{\sf no}-(\tilde{C}_{\sf pf}-\tilde{C}_{\sf no})-2\tilde{n}}
\times(3,2)^{2m-n},\\
(\tilde{n},\tilde{m})\longrightarrow&{(0,1)^{\tilde{C}_{\sf pf}-\tilde{C}_{\sf no}}}
\times{(1,2)^{\tilde{n}}}.
\end{align*}
Note that $C_{\sf pf}-C_{\sf no}> \tilde{C}_{\sf pf}-\tilde{C}_{\sf no}+2\tilde{n}$ in the considered regime $C_{\sf pf}-C_{\sf no}> 2\tilde{m}-\tilde{C}_{\sf pf}$. 
We now apply Lemma~\ref{lemma:buildingblocks}-(i) for the pair of ${(2,1)^{\tilde{C}_{\sf pf}-\tilde{C}_{\sf no}}}$ and ${(0,1)^{\tilde{C}_{\sf pf}-\tilde{C}_{\sf no}}}$; apply Lemma~\ref{lemma:buildingblocks}-(ii) for the pair of ${(2,1)^{2\tilde{n}}}$ and ${(1,2)^{\tilde{n}}}$; apply the nonfeedback scheme for the rest. This gives: 
\begin{align*}
R&={3(\tilde{C}_{\sf pf}-\tilde{C}_{\sf no})}
+{3 \cdot 2\tilde{n}}
+2\{C_{\sf pf}-C_{\sf no}-(\tilde{C}_{\sf pf}-\tilde{C}_{\sf no})-2\tilde{n}\}
+4 (2m-n) =C_{\sf no}+2\tilde{m}-\tilde{C}_{\sf pf},\\
\tilde{R}&={(\tilde{C}_{\sf pf}-\tilde{C}_{\sf no})}
+{2\tilde{n}} =\tilde{C}_{\sf pf}.
\end{align*}

Let us now prove the second corner point which favours $R$. As mentioned earlier, we have three subcases depending on the quantity of $C_{\sf pf} - C_{\sf no}$.

{\bf (III-1) $2\tilde{m}-\tilde{C}_{\sf pf} < C_{\sf pf}-C_{\sf no} \leq 2\tilde{m}-\tilde{C}_{\sf no}$}: In this case, we have:
\begin{align*}
(n,m)\longrightarrow&{(2,1)^{2\tilde{m}-\tilde{C}_{\sf no}-(C_{\sf pf}-C_{\sf no})}}
\times {(2,1)^{2\{C_{\sf pf}-C_{\sf no}-(2\tilde{m}-\tilde{C}_{\sf pf})\}}}
\times {(2,1)^{2\tilde{n}}}
\times(3,2)^{2m-n},\\
(\tilde{n},\tilde{m})\longrightarrow&{(0,1)^{2\tilde{m}-\tilde{C}_{\sf no}-(C_{\sf pf}-C_{\sf no})}}
\times {(0,1)^{C_{\sf pf}-C_{\sf no}-(2\tilde{m}-\tilde{C}_{\sf pf})}}
\times {(1,2)^{\tilde{n}}}.
\end{align*}
We now apply Lemma~\ref{lemma:buildingblocks}-(i) for the pair of ${(2,1)^{2\tilde{m}-\tilde{C}_{\sf no}-(C_{\sf pf}-C_{\sf no})}}$ and ${(0,1)^{2\tilde{m}-\tilde{C}_{\sf no}-(C_{\sf pf}-C_{\sf no})}}$; apply Lemma~\ref{lemma:buildingblocks}-(v) for the pair of ${(2,1)^{2\{C_{\sf pf}-C_{\sf no}-(2\tilde{m}-\tilde{C}_{\sf pf})\}}}$ and ${(0,1)^{C_{\sf pf}-C_{\sf no}-(2\tilde{m}-\tilde{C}_{\sf pf})}}$;  apply the nonfeedback scheme for the rest. This then yields:
\begin{align*}
R&={3 \{2\tilde{m}-\tilde{C}_{\sf no}-(C_{\sf pf}-C_{\sf no})\}}
+{3\cdot 2\{C_{\sf pf}-C_{\sf no}-(2\tilde{m}-\tilde{C}_{\sf pf})\}}
+{3\cdot 2\tilde{n}}
+4(2m-n)=C_{\sf pf},\\
\tilde{R}&={\{2\tilde{m}-\tilde{C}_{\sf no}-(C_{\sf pf}-C_{\sf no})\}}
+{2\tilde{n}} =2\tilde{m}-(C_{\sf pf}-C_{\sf no}).
\end{align*}

{\bf (III-2) $2\tilde{m} -\tilde{C}_{\sf no} < C_{\sf pf}-C_{\sf no} \leq  2\tilde{m}$}: In this case, we will take the two-step approach: first obtaining an intermediate point that lies on the 45-degree line and then perturbing the scheme to prove achievability of the second corner point of interest. 

In the considered regime, we have:
\begin{align*}
(n,m)\longrightarrow&{(2,1)^{2(\tilde{C}_{\sf pf}-\tilde{C}_{\sf no})}}
\times {(2,1)^{2\tilde{n}}}
\times (2,1)^{C_{\sf pf}-C_{\sf no}-2\times(\tilde{C}_{\sf pf}-\tilde{C}_{\sf no})-2\tilde{n}}
\times(3,2)^{2m-n},\\
(\tilde{n},\tilde{m})\longrightarrow&{(0,1)^{\tilde{C}_{\sf pf}-\tilde{C}_{\sf no}}}
\times {(1,2)^{\tilde{n}}}.
\end{align*}
Here we use the fact that $ C_{\sf pf}-C_{\sf no}> 2(\tilde{C}_{\sf pf}-\tilde{C}_{\sf no})+2\tilde{n}$ due to $C_{\sf pf}-C_{\sf no}> 2\tilde{m}-\tilde{C}_{\sf no}$. We now apply Lemma~\ref{lemma:buildingblocks}-(v) for the pair of ${(2,1)^{2(\tilde{C}_{\sf pf}-\tilde{C}_{\sf no})}}$ and ${(0,1)^{\tilde{C}_{\sf pf}-\tilde{C}_{\sf no}}}$; apply Lemma~\ref{lemma:buildingblocks}-(ii) for the pair of ${(2,1)^{2\tilde{n}}}$ and ${(1,2)^{\tilde{n}}}$; apply the nonfeedback scheme for the rest. This gives: 
\begin{align*}
R&={3\times2(\tilde{C}_{\sf pf}-\tilde{C}_{\sf no})}
+{3\times 2\tilde{n}}
+2\times\{C_{\sf pf}-C_{\sf no}-2\times(\tilde{C}_{\sf pf}-\tilde{C}_{\sf no})-2\tilde{n}\}
+4\times(2m-n),\\
\tilde{R}&={2\tilde{n}}.
\end{align*}
We now change the scheme that achieves the above rate pair to prove achievability of the second corner point. Specifically we utilize $C_{\sf pf}-C_{\sf no}-2\times(\tilde{C}_{\sf pf}-\tilde{C}_{\sf no})-2\tilde{n}$ number of cross links in the backward channel to help forward transmission. 
% $(1,2)^{\nt}$ backward subchannel to help forward transmission w.r.t. $(2,1)^{C_{\sf pf}-C_{\sf no}-2\times(\tilde{C}_{\sf pf}-\tilde{C}_{\sf no})-2\tilde{n}}$ number of forward symbols to achieve rate 3 in each $(2,1)$ network, instead of sending their own backward information, 
This way, we can achieve:
\begin{align*}
R&={3\times2(\tilde{C}_{\sf pf}-\tilde{C}_{\sf no})}
+{3\times 2\tilde{n}}
+2\times\{C_{\sf pf}-C_{\sf no}-2\times(\tilde{C}_{\sf pf}-\tilde{C}_{\sf no})-2\tilde{n}\}
+4\times(2m-n)\\
&+\{C_{\sf pf}-C_{\sf no}-2\times(\tilde{C}_{\sf pf}-\tilde{C}_{\sf no})-2\tilde{n}\} =2n-m=C_{\sf pf},\\
\tilde{R}&={2\tilde{n}}-\{C_{\sf pf}-C_{\sf no}-2\times(\tilde{C}_{\sf pf}-\tilde{C}_{\sf no})-2\tilde{n}\} =2\tilde{m}-(C_{\sf pf}-C_{\sf no}).
\end{align*}    

{\bf (III-3) $ C_{\sf pf}-C_{\sf no}> 2\tilde{m}$}:
In this case, all of the $2\tilde{m}$ cross-link levels in the backward channel are used solely for aiding forward transmission. So we can achieve: $(R,\tilde{R})=(C_{\sf no}+2\tilde{m},0)$.

\begin{figure}[t]
\begin{center}
\epsfig{figure=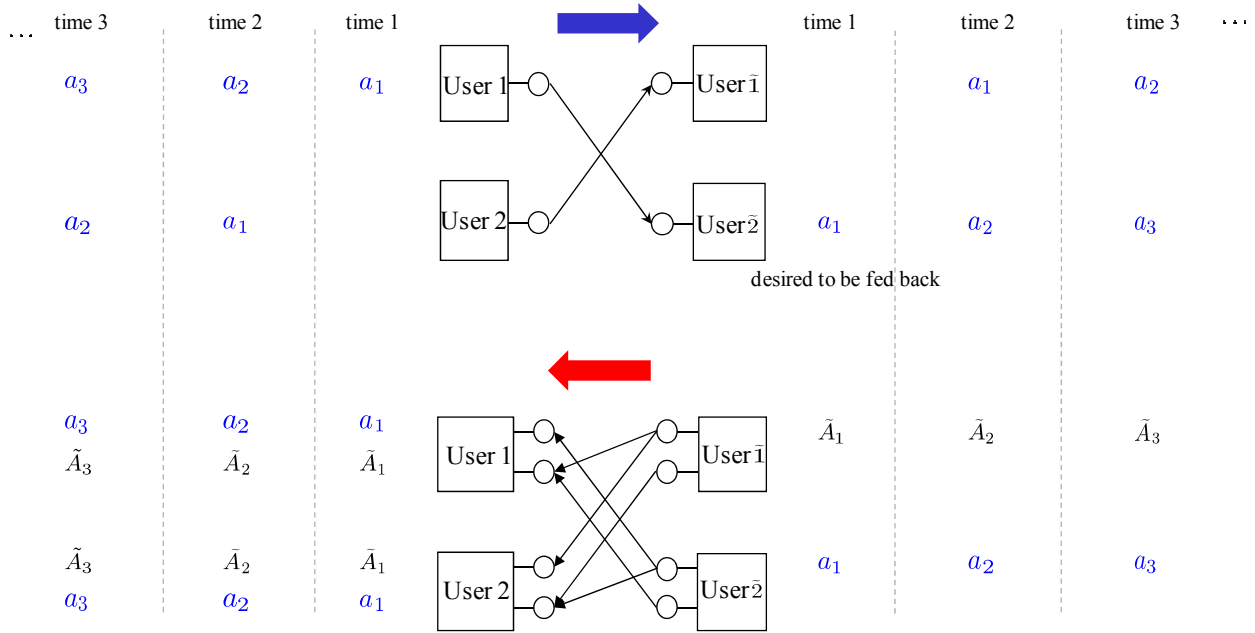, angle=0, width=1.0\textwidth}
\end{center}
\caption{For the pair of $(n,m)=(0,1)$, $(\tilde{n},\tilde{m})=(1,2)$, one can achieve $(R,\rt)=(1,1)$.} 
\label{fig:BuildingBlock1}
\end{figure}

\section{Proof of Lemma~\ref{lemma:buildingblocks1}}
\label{append:lemma_buildingblocks1}

\emph{(i):} We will illustrate achievability via the simplest example in which $(i,j)=(1,1)$. See Fig.~\ref{fig:BuildingBlock1}. In the forward channel $(0,1)$, only one user (say user 1) intends to send one symbol (say $a_i$) every time slot. User $\tilde{2}$ then feeds the symbol back to user 2 using the top level in the backward channel. Next user 2 delivers the fed back symbol to user $\tilde{1}$. This way, we achieve $R=1$. Now for $(n,m)=(0,1)^i$ and $(\nt,\mt)=(1,2)^j$, consider sending $i$ number of feedback symbols from user $\tilde{2}$ to user 2. Since the total number of resource levels at user 2 in the backward channel is $2j$, one can ensure $R=i$ as long as $i \leq 2j$. On the other hand, the remaining $2j-i$ resource levels at user 2 are used for backward traffic. In the example, $2j-i=1$, so one backward symbol $\tilde{A}_i$ is transmitted per time. Notice that this transmission also occupies a resource level at user 1. This prevents from squeezing more backward symbols, thus yielding $\rt = 2j - i$. Here one key observation to make is that feedback for increasing $R$ by one bit incurs one bit of degradation w.r.t. $\rt$, meaning that there is \emph{one-to-one tradeoff} between feedback and backward message transmissions, as demonstrated in~\cite{SuhTse}.

\begin{figure}[t]
\begin{center}
\epsfig{figure=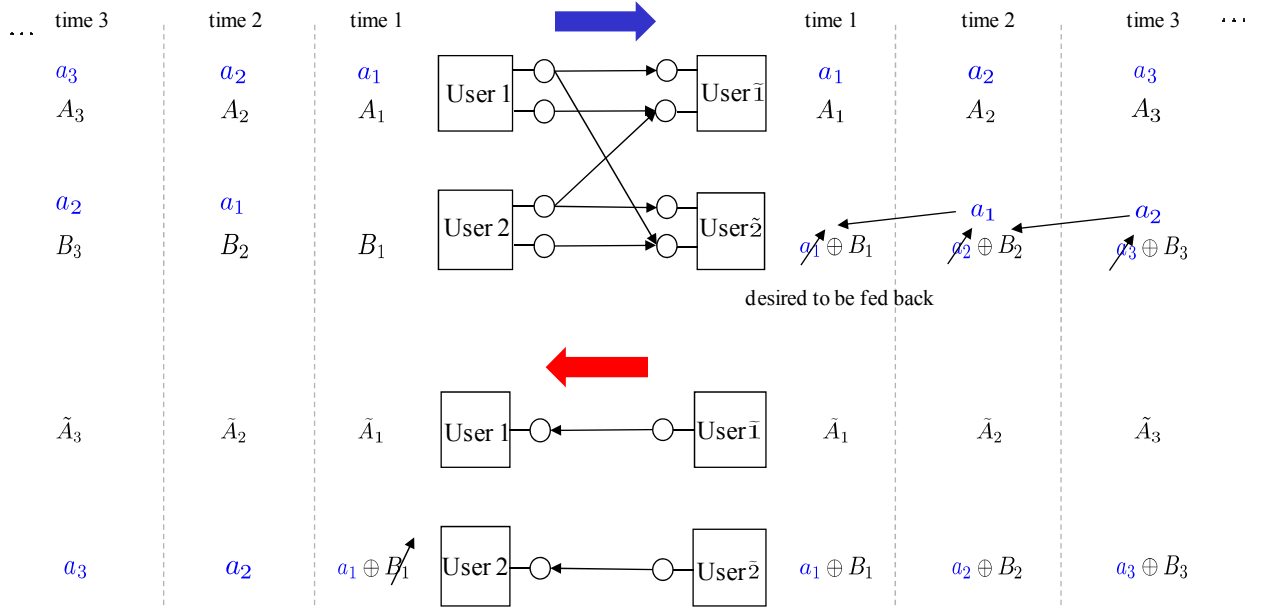, angle=0, width=1.0\textwidth}
\end{center}
\caption{For the pair of $(n,m)=(2,1)$, $(\tilde{n},\tilde{m})=(1,0)$, one can achieve $(R,\rt)=(3,1)$.} 
\label{fig:BuildingBlock2}
\end{figure}    

\emph{(ii):} We will describe achievability via a special case of $(i,j,k)=(1,1,0)$: $(n,m)=(2,1)$ and $(\nt, \mt)=(1,0)$. See Fig.~\ref{fig:BuildingBlock2}. In the forward channel $(2,1)$, user 1 sends two symbols $(a_i,A_i)$ per time, while user 2 sends only one symbol $B_i$ on bottom. The symbol $B_i$ is interfered with $a_i$. User $\tilde{2}$ then sends the interfered signal $a_i \oplus B_i$ back to user 2, which in turn enables user 2 to decode $a_i$. User 2 forwarding the $a_i$ through the top level in the next time allows user $\tilde{2}$ to refine the corrupted symbol. For instance, at time 2, user $\tilde{2}$ can decode $B_1$ by subtracting $a_1$ from $a_1 \oplus B_1$. This way, we achieve $R = \frac{3N-1}{N} \rightarrow 3 $ as code length $N$ tends to infinity. Now for $(n,m)=(2,1)^i$ and $(\nt,\mt)=(1,0)^j$, consider sending $i$ number of feedback symbols either from user $\tilde{2}$ to user 2 (as in the example) or from user $\tilde{1}$ to user 1 (this is the case in which user 2 sends more compared to user 1). Then, we can achieve $R=3i$ as long as $i$ does not exceed the total number $2j$ of resource levels in the backward channel which corresponds to the nonfeedback sum-rate. For the remaining resource levels $2j-i$, we employ the nonfeedback scheme to achieve $\rt = 2j - i$. As in the previous case $(i)$, we see one-to-one tradeoff. One can apply the same argument for $(n,m)=(2,1)^i$ and $(\nt,\mt)=(2,1)^k$ to observe the same one-to-one tradeoff relationship. The only distinction is that in this case, the nonfeedback sum-rate of the backward channel is $2k$. Hence, we achieve $(R,\rt)=(3i,2k-i)$ under $i \leq 2k$. Now for the general $(i,j,k)$ case, combining the above two, we get $(R,\rt)=(3i, 2j+2k-i)$ if $i \leq 2j + 2k$. 

\emph{(iii):} For $(n,m)=(2,1)^i$ and $(\nt,\mt)=(2,1)^j$, the proof in the $(ii)$ case yields $(R,\rt)=(3i, 2j-i)$ under $i \leq 2j$. For $(n,m)=(2,1)^i$ and $(\nt,\mt)=(3,2)^k$, using the same argument and the fact that the nonfeedback sum-rate of the backward channel is $4k$, one can show that $(R,\rt)=(3i, 4k - i)$ under $i \leq 4k$. Now for the general $(i,j,k)$ case, combining the above two, we can achieve $(R,\rt)=(3i,2j+4k-i)$ as long as $i\leq 2j+4k$. This completes the proof.

\section{Proof of Lemma~\ref{lemma:buildingblocks}}
\label{append:lemma_buildingblocks}

\emph{(i):} See \emph{Scheme 2} in Section IV.

\emph{(ii):} Obviously $(n,m)=(2,1)^i=(2i,i)$ and $(\nt,\mt)=(1,2)^j=(j,2j)$. Note in this case that $C_{\sf pf}=3i > 2i = C_{\sf no}$ and $\tilde{C}_{\sf no}=2j$ and hence the channel belongs to the (R4) regime. Since the condition $i \leq 2j$ corresponds to $C_{\sf pf}-C_{\sf no}\leq 2\tilde{m}-\tilde{C}_{\sf no}$, the achievability for the (R4) regime yields $(R,\tilde{R})=(C_{\sf pf},\tilde{C}_{\sf no})$.

\emph{(iii):} Obviously $(n,m)=(3,2)^i=(3i,2i)$ and $(\nt,\mt)=(0,1)^j=(0,j)$. Note in this case that $C_{\sf no}=4i$ and $\tilde{C}_{\sf pf}=j > 0 =\tilde{C}_{\sf no}$ and hence the channel belongs to the (R3') regime (the symmetric counterpart of (R3)). Since the condition $2i \geq j$ corresponds to $\tilde{C}_{\sf pf}-\tilde{C}_{\sf no}\leq 2n-{C}_{\sf no}$, the achievability for the (R3') regime yields $(R,\tilde{R})=(C_{\sf no},\tilde{C}_{\sf pf})$.  

\emph{(iv):} In the forward channel, each user sends one bit on bottom every time, while the upper level is utilized to relay backward-symbol feedback. See Fig.~\ref{fig:R5achievability}. Here the backward-symbol feedback, say $\tilde{b}_1$, does not cause any interference to user $\tilde{2}$ as it is already known. Hence, two feedback symbols can be delivered every time and this yields $(R, \rt) = (2,2)$. 
 
%
%the backward information to make each backward network achieve the perfect feedback bound. Sending this information does not interfere user $\tilde{1}$ and $\tilde{2}$ from decoding forward symbols since they can use their backward information as side information. Then $(R,\tilde{R})=(2,2)$ can be achieved (see fig.~\ref{fig:R5achievability}). 

\begin{figure}[t]
\begin{center}
{\epsfig{figure=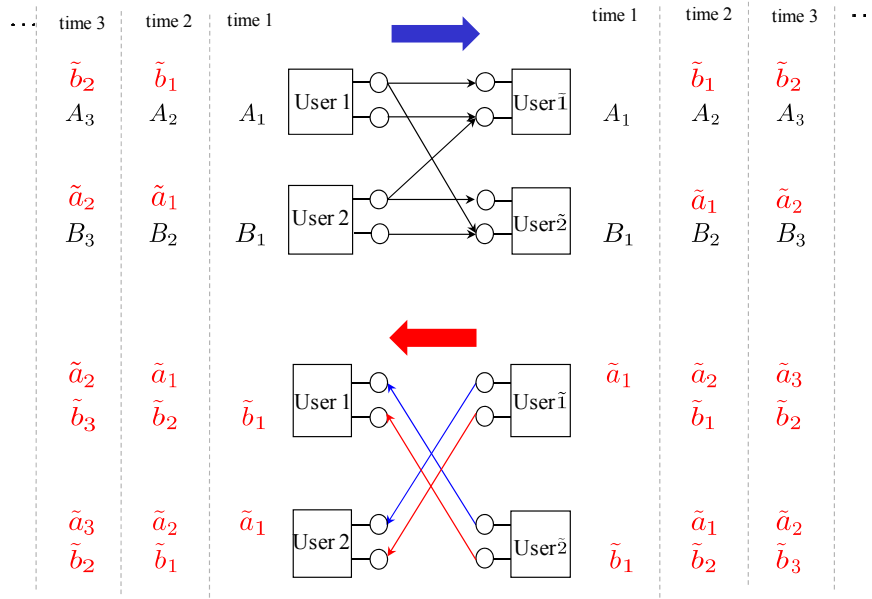, angle=0, width=0.7\textwidth}}
\end{center}
\caption{For the pair of $(n,m)=(2,1)$ and $(\nt,\mt)=(0,1)^2$, one can achieve $(R,\tilde{R})=(2,2)$.} \label{fig:R5achievability}
\end{figure}

\emph{(v):} The backward channel is used solely for feeding back forward-symbol feeddback. One can easily verify that this enables us to achieve $R= C_{\sf pf} = 2 \times 3 = 6$.   

\bibliographystyle{ieeetr}
\bibliography{bib_feedback}

\end{document}